\documentclass[11pt,a4paper]{article}
\usepackage[a4paper]{geometry}
\usepackage{amssymb,latexsym,mathtools,amsfonts,amsthm}
\usepackage{graphicx}
\usepackage{amsbsy}
\usepackage{hyperref}\hypersetup{hidelinks}
\usepackage{fullpage}
\usepackage{enumitem}
\usepackage{color}
\usepackage{tikz}
\usepackage{subfigure}
\usepackage{booktabs,longtable}
\usepackage[alphabetic,initials]{amsrefs}
\usepackage{hyperref}
\usepackage{curves}
\usepackage{array}
\usepackage[english]{babel}
\usepackage{enumerate}
\usepackage{mathrsfs}
\usepackage{empheq}
\usepackage{comment}
\usepackage{braket}
\hypersetup{colorlinks = true, linkcolor = {blue} }
\usepackage{cancel}
 \usepackage{lipsum}
 \usepackage{mathtools}
 \usepackage{MnSymbol}
\usepackage{xcolor}
 \usepackage{subfigure}
\usepackage{soul}

\usepackage{tikz} % Optional PGF libraries
\usetikzlibrary{shapes,snakes}
\usepackage{pgf}
\usepackage{pgflibraryarrows}
\usepackage{pgflibrarysnakes}
\usetikzlibrary{decorations.text}
\usepgfmodule{shapes}
\usetikzlibrary{decorations.pathmorphing}
\usetikzlibrary{decorations.markings}
\usetikzlibrary{patterns}
\usetikzlibrary{automata}
\usetikzlibrary{positioning}
\usetikzlibrary {arrows.meta}
\usepackage{pgfplots}
\usetikzlibrary{backgrounds}
\tikzset{partial ellipse/.style args={#1:#2:#3}{insert path={+ (#1:#3) arc (#1:#2:#3)} }}
\tikzset{->-/.style={decoration={ markings, mark=at position #1 with {\arrow{>}}},postaction={decorate}}}

\tikzset{
  midarrow/.style={
    decoration={
      markings,
      mark=at position 0.5 with {\arrow[>=stealth]{>}}
    },
    postaction={decorate}
  }
}

\definecolor{Dgreen}{RGB}{0,153,0}

\newcommand*{\mailto}[1]{\href{mailto:#1}{\nolinkurl{#1}}}

\newtheorem{theorem}{Theorem}[section]

\newtheorem{remark}[theorem]{Remark}

\newtheorem{prop}[theorem]{Proposition}
\newtheorem{thm}[theorem]{Theorem}

\theoremstyle{definition}
\newtheorem{defn}[theorem]{Definition}
\newtheorem{Rhp}[theorem]{RH problem}
\newcommand{\oo}{\mathcal{O}}
\newcommand{\ccc}{\mathbb{C}}
\newcommand{\AAA}{{\mathcal A}}
\newcommand{\R}{\mathbb{R}}
\newcommand{\Z}{\mathbb{Z}}
\newcommand{\C}{\mathbb{C}}

\newcommand{\abs}[1]{\lvert#1 \rvert}

\newcommand{\ddd}{\mathrm{d}}

\newcommand{\ii}{\mathrm{i}}
\newcommand{\E}{\mathrm{e}}

\newcommand{\re}{\mathrm{Re}}
\newcommand{\I}{\mathrm{Im}}
\newcommand{\sech}{\mathop{\mathrm{sech}}}
\DeclareMathOperator{\res}{Res}

\newcommand{\lb}{\lambda}
\def\d{{\rm d}}
\def\e{{\rm e}}
\def\i{{\rm i}}
\def\1{\operatorname{Id}}

\def\res{\mathop{{\rm Res}}}
\def\Tr{\mathop{{\rm Tr}}}
\def\Ai{\mathrm{Ai}}
\def\Re{\operatorname{Re}}
\def\Im{\operatorname{Im}}
\def\identity{\operatorname{Id}}
\def\exp{\operatorname{exp}}
\def\Bes{\mathrm{Bes}}
\def\Jac{\operatorname{nd}}

\def\XXint#1#2#3{{\setbox0=\hbox{$#1{#2#3}{\int}$}
     \vcenter{\hbox{$#2#3$}}\kern-.5\wd0}}

\numberwithin{equation}{section}

\begin{document}
\title{A dense focusing Ablowitz-Ladik soliton gas and its asymptotics}
\author{Meisen Chen\footnotemark[1],\quad Engui Fan\footnotemark[2], \quad Zhaoyu Wang\footnotemark[4]~\footnotemark[5],  \quad Yiling Yang\footnotemark[6],
\quad Lun Zhang\footnotemark[2]~\footnotemark[3]}

\footnotetext[1]{School of Mathematics and Statistics, Fujian Normal University, Fuzhou 350117,  China. E-mail: \texttt{chenms@fjnu.edu.cn.}}
\footnotetext[2]{School of Mathematical Sciences, Fudan University, Shanghai 200433, China.
E-mail: \texttt{\{faneg, lunzhang\}@fudan.edu.cn.}}
\footnotetext[3]{Shanghai Key Laboratory for Contemporary Applied Mathematics, Fudan University, Shanghai 200433, China.}
\footnotetext[4]{Department of Mathematics, Shanghai University, Shanghai 200444, China. E-mail: \texttt{zhaoyuwang@shu.edu.cn.}}
\footnotetext[5]{Newtouch Center for Mathematics of Shanghai University, Shanghai 200444, China.}
\footnotetext[6]{College of Mathematics and Statistics, Chongqing University, Chongqing 401331, China. E-mail: \texttt{ylyang19@fudan.edu.cn.}}

\maketitle
\begin{abstract}
In this paper, we propose a soliton gas solution for the focusing Ablowitz-Ladik system. This solution is defined as the large $N$ limit of the $N$-soliton solution, and arises from a continuous spectrum of poles that accumulate within two disjoint intervals on the imaginary axis. We show that this gas solution admits a Fredholm determinant representation. By further exploring its Riemann-Hilbert characterization, we are able to establish the large-space asymptotics at $t=0$ and large-time asymptotics of the gas solution. 
\\[6pt]
% We establish the soliton gas Riemann-Hilbert problem for the focusing Ablowitz-Ladik (AL) system, 
% construct both soliton and soliton gas solutions in terms of Fredholm determinants, 
% and analyze the asymptotic behaviors of the soliton gas solution $q_n(t)$.
% 	Specifically, we investigate the large-space asymptotic behaviors at $t=0$: as $n\to+\infty$, $q_n(0)$ decays exponentially, while as $n\to-\infty$, $q_n(0)$ approaches an algebro-geometric solution at a rate of $\oo(n^{-1})$.
% 	Furthermore, we study the large-time asymptotic behaviors of the soliton gas solution in different regions, including a fast decaying region, two hyperelliptic wave regions, and two transition regions.

    \noindent {\bf AMS Subject Classification 2020:} 37K60; 39A14; 35Q51; 35B40; 35Q15;  35Q35.
\end{abstract}

\tableofcontents

\section{Introduction and main results} %\label{sec1}
%Introduction for AL
%In this paper, we  consider the soliton gas solutions of
The Ablowitz-Ladik (AL) system is defined by 
\begin{align}\label{e1.1}
	\ii\frac{\ddd }{\ddd  t}q_n=q_{n+1}-2q_n+q_{n-1}+\sigma\lvert q_n\rvert^2(q_{n+1}+q_{n-1}),\qquad n\in \mathbb{Z},
\end{align}
where  $\sigma= \pm 1$, corresponding to the focusing/defocusing case, respectively. The AL system was introduced by Ablowitz and Ladik \cite{AL1974,AL1976} via discretizing the $2\times 2$ Zakharov-Shabat Lax pair of the cubic nonlinear Schr\"{o}dinger equation 
\begin{align*}
    \ii u_t = u_{xx}+\sigma|u|^2 u.
\end{align*}
In addition, the AL system also possesses numerous physical applications, ranging from the dynamics of anharmonic lattices \cite{TH1990}, self-trapping on a dimer \cite{KC1986} to Heisenberg spin chains \cite{I1982, N1987}.
%and  nonlinear waveguides \cite{BKA2009}. 
%The integrability of the AL system has also been proved by constructing a bi-Hamiltonian structure \cite{AKV2020,EL2006}.
%Besides, being used as numerical difference schemes for its continuous counterpart, the AL system possesses numerous physical applications, including the dynamics of anharmonic lattices \cite{TH1990}, self-trapping on a dimer \cite{KC1986} and Heisenberg spin chains \cite{I1982, N1987}.
%and  nonlinear waveguides \cite{BKA2009}.

In this paper, we are concerned with the focusing AL system, i.e., $\sigma=1$ in  \eqref{e1.1}, which reads
\begin{align}\label{e1.2}
	\ii\frac{\ddd }{\ddd  t}q_n=q_{n+1}-2q_n+q_{n-1}+\lvert q_n\rvert^2(q_{n+1}+q_{n-1}).
\end{align}
It is well known that the  AL system is integrable as can be seen from its  bi-Hamiltonian structure \cite{AKV2020,EL2006}. There have been intensive studies of \eqref{e1.2} from multiple perspectives, which include the initial-boundary problem \cite{XF2018}, inverse scattering transform (IST) with nonzero background \cite{ABP2007, P2016,PV2016, VK1992}, integrable decomposition \cite{ GDZ2007},  quasi-periodic solutions \cite{LN2012, MEKL1995, OY2014}, and large-time asymptotics in the presence of solitons \cite{CFW2025}. Among various investigations of the focusing AL system, we particularly mention the soliton solutions -- one of the most fundamental types of particular solutions for integrable systems.  In the framework of IST, these solutions are related to  poles of the transmission coefficient. A one-soliton solution arises from a pair of conjugate simple poles, while an $N$-soliton solution corresponds to $N$ pairs of conjugate simple poles. The one-soliton solution to the AL system  \eqref{e1.2} associated with the discrete eigenvalue $\lb_1 \in \C$ such that $|\lb_1|>1$ is given by \cite{AL1974,APT2004}
%and its norming constant $\varLambda_1$
\begin{align}
    q_n(t) = \i \frac{\sinh(2 \log|\lb_1|)}{|\lb_1|} \e^{-2 \i(-n\arg \lb_1- \tilde{w} t - \psi^+) } \sech( 2 (n+1) \log|\lb_1| - \tilde{v}t -\delta_0),
\end{align}
where the constants $\tilde{w}$, $\tilde{v}$ and $\delta_0$ depend on $\lb_1$, and $\psi^+ = -\arg \varLambda_1 -2 \arg \lb_1 \pm \pi$ with $\varLambda_1$ being the normalizing constant; see also \cite{LSG2024} for the higher-order soliton solution and \eqref{eqtau} below for the $N$-soliton solution.

The interpretation of soliton ensembles as particle-like entities has been a catalyst for a multitude of novel research endeavors ever since its discovery. In \cite{Z1971}, the notion of “soliton gas” was first introduced by  Zakharov in the context of Korteweg-de Vries (KdV) equation. By evaluating the efficient modification of the soliton velocity within a rarefied gas, Zakharov derived an integro-differential kinetic equation for the soliton gas. The kinetic equation describes the evolution of the spectral distribution function of solitons due to soliton-soliton collisions. 
%The kinetic equation for the soliton gas of other integrable equations like Whitham equation and nonlinear Schr\"{o}dinger equation can be found in \cite{E03,EK05,EKPZ11,ET20}.
The kinetic equation for the soliton gas has been extended to other integrable equations like focusing nonlinear Schr\"{o}dinger (NLS) equation \cite{EK05, EKPZ11} and defocusing and resonant NLS equations \cite{CER2021}.
In addition, Dyachenko \textit{et al.} formulated a Riemann-Hilbert (RH) problem for the soliton gas of the KdV equation in \cite{DZZ16}. Later, Girotti \textit{et al.} established large distance and large time asymptotics of KdV soliton gas dynamics in \cite{kdvgas} via the RH approach; see also \cite{mkdvgas} for the analysis of a dense modified Korteweg-de Vries (mKdV) soliton gas.

% Since then,  many creative and worthwhile papers have studied the kinetic equation for the soliton gas of integrable systems \cite{E03,EK05,EKPZ11,ET20}. Research on soliton gas via the Riemann-Hilbert (RH) method started from \cite{DZZ16}, where Dyachenko and his collaborators first formulated an RH problem for the soliton gas of the KdV equation. Following their idea, researchers began to use  RH problem for the soliton gas of integrable systems to investigate its properties. Notably,  Girotti and her collaborators in \cite{kdvgas} used the RH problem for the   KdV soliton gas  to establish an asymptotic description of soliton gas dynamics for large distance and large time. Further, in \cite{mkdvgas}, Girotti and her collaborators analyzed the case of a dense modified Korteweg–de Vries (mKdV) soliton gas and its large-time behavior in the presence of a single trial soliton.

%上面强调连续的结果，下面说离散的
Although substantial advances have been achieved in the studies of soliton gases for continuous integrable systems, the corresponding analogue for discrete integrable systems, however, has seen very little progress. In \cite{AA2025}, Aggarwal considered the Toda lattice $(\mathbf{p}(t);\,\mathbf{q}(t))$ at thermal equilibrium in the sense that the variables $(p_i)$ and $(\E^{q_i -q_{i+1}})$ are independent Gaussian and Gamma random variables, respectively, which can be interpreted as a soliton gas of the Toda system. The present work represents the first attempt to investigate the soliton gas of discrete integrable systems via the RH approach in the context of the focusing AL system \eqref{e1.2}, pioneering a large-$n$ and large-$t$ asymptotic analysis framework.  Our main results are stated in the next section.

% introduce the RH method into the study of soliton gases in the AL system \eqref{e1.2}, pioneering a large-$n$ and large-$t$ asymptotic analysis framework for soliton gases in discrete integrable systems. 
%In this paper, we aim to  analyse the case of a dense AL soliton gas and its large-time behavior.

\subsection{Main results}
As in the case of continuous integrable system, the soliton gas solution of \eqref{e1.2} is defined as the large $N$ limit of its $N$-soliton solution $q^{[N]}_{n}$. Here, $q^{[N]}_{n}$ is determined by the discrete spectrum consisting of the points $\i \lb_j^{\pm1}$, $\lb_j>0$, $j=1,\cdots,N$, together with the associated norming constants $\varLambda_j \in \R \setminus \{ 0\}$ and reflection coefficient $r(\lb)$.  By taking $N \to \infty$, the soliton gas solution $q_n$ arises from a continuous spectrum of poles that accumulate within the interval $\Sigma_1 \cup \Sigma_2$, where %$\Sigma_1=(\i\eta_1, \i\eta_2)$ and $\Sigma_2=(\i\eta_2^{-1}, \i\eta_1^{-1})$. Here the parameters satisfy the condition $1 < \eta_1 < \eta_2$. 
\begin{align}
    \Sigma_1:=(\i\eta_1, \i\eta_2),\qquad\Sigma_2:=(\i\eta_2^{-1}, \i\eta_1^{-1}),\qquad1 < \eta_1 < \eta_2.\label{def sigma1} 
\end{align}
Our first result shows that $q_n$ can be expressed in terms of a Fredholm determinant, which is a natural extension of the Fredholm determinant representation of the $N$-soliton solution $q^{[N]}_{n}$ established in \eqref{eqtau} below.

\begin{theorem}\label{thm1}
	%Under Assumption \ref{assum1},
	%we set $\mathcal{K}$ as an integral operator on $L^2(\Sigma_1)$
	%\begin{align}
	%	\mathcal{K}[f](\zeta)=\int_{\i\eta_1}^{\i\eta_2}K(\zeta,s)f(s)\ddd s.
	%\end{align}
	%with the kernel
	%\begin{align*}
	%	K(\zeta,s)=\frac{\sqrt{r(\zeta)r(s)}(-\zeta s)^{-(n+2)/2}e^{\i t(\zeta+s+\zeta^{-1}+s^{-1})/2}}{2\pi\i(1+\zeta s)}.
	%\end{align*}
	%As a result, we obtain the soliton gas solution in the form
	%\begin{align}
	%	q_n(t)=\i^{n+1}\E^{2\i t}\frac{\d}{\d t}\Im\ln\det(\identity_{L^2([\i\eta_1,\i\eta_2])}+ \mathcal{K}).
	%\end{align}
    %Deonte $\Sigma_1=(\i\eta_1,\i\eta_2)$, $\eta_2>\eta_1>1$. 
Suppose that $r(\lb)$ is a positive and continuous function on $[\i\eta_1,\i\eta_2]$ and let $\mathcal{K}: L^2(\Sigma_1)\to L^2(\Sigma_1)$ be an integral operator defined by
	\begin{align}
		\mathcal{K}_{n,t}[f](\zeta)=\int_{\i\eta_1}^{\i\eta_2}K_{n,t}(s;\zeta)f(s) \ddd s,
	\end{align}
where
	\begin{align*}
		K_{n,t}(s;\zeta)=-\frac{\sqrt{r(\zeta)r(s)}(-\zeta s)^{-n/2}e^{\i t(\zeta+s+\zeta^{-1}+s^{-1})/2}}{2\pi(1+\zeta s)}
	\end{align*}
is the associated kernel. Then,
	\begin{align}\label{e1.4}
		q_n(t):=\i^{n+1}\E^{2\i t}\frac{\d}{\d t}\Im\ln\det\left(\identity_{L^2(\Sigma_1)}+ \mathcal{K}_{n+1,t}\right)
	\end{align}
is the soliton gas solution for the focusing AL system \eqref{e1.2}. In addition, we have
    \begin{align}\label{fcn}
\prod_{k=n}^\infty(1+\abs{q_k(t)}^2)=\left(\re\left[\frac{\det(\identity_{L^2(\Sigma_1)}+\mathcal{K}_{n+2,t})}{\det(\identity_{L^2(\Sigma_1)}+\mathcal{K}_{n,t})}\right]\right)^{-1}.
    \end{align}
\end{theorem}

%\begin{remark}
%	In Theorem \ref{thm1}, we see that the soliton gas solution can be written in the Fredholm representation.
%    Especially, at the right-hand side of \eqref{e1.4}, $\i^{n+1}e^{2\i t}$ is the phase and $\mathcal{K}_{n,t}$ is the Fredholm operator. 
%	In the following asymptotic analysis, we also find that it is equivalent to analyze the asymptotic behavior of the term $\frac{\d}{\d t}\Im\ln\det(\identity_{L^2([\i\eta_1,\i\eta_2])}+ \mathcal{K}_{n+1,t})$.
%\end{remark}

%By analyzing the RH problem for the soliton gas of  \eqref{e1.2},   we obtain the following theorem in Section \ref{sec 3}, which characterizes the asymptotic behavior of the soliton gas potential $q_n(0)$ as $n\to\pm\infty$.
%Under Assumption \ref{assum1}, in Section \ref{sec fredhlom}, we present the soliton gas associated to $\{r(\lb),\eta_1,\eta_2\}$  as follows, which is a limit of $N$-soliton solution as $N\to\infty$.
%Under Assumption \ref{assum1}, we prove that the soliton gas RH problem \ref{RHP0}, is exactly solved.
%In addition, the solution recovered from this RH problem solves \eqref{e1.2}.
%Further, in Section \ref{sec fredhlom}, we present the soliton gas solution in the form of the Fredhlom determinant as follows, which is a limit of $N$-soliton solution as $N\to\infty$.

One cannot evaluate the Fredholm determinants \eqref{e1.4} explicitly for fixed $n$ and $t$, and therefore it is natural to try to approximate them for large $n$ and $t$. We start with the asymptotics of $q_n(0)$ as $n\to \pm \infty$.

% On the other hand,  by taking $N \to \infty$ through RH gymnastics,  we obtain the asymptotic behavior of $q_n(0)$ as $n\to \pm \infty$.
\begin{thm}\label{Thm 1}
Let %the intervel $\Sigma_1=(\i\eta_1,\i\eta_2)$, oriented upwards with $\eta_2>\eta_1>1$, and 
$q_n(t)$ be the soliton gas solution of the focusing AL system \eqref{e1.2} defined in \eqref{e1.4}.
%related to spectral data $\{r(\lb),r(\bar\lb^{-1});\Sigma_1\cup \Sigma_2\}$. 
Assume that $r(\lb)$ is analytic and positive in a neighborhood of $\Sigma_1$, we have  
%, then, since $\Sigma_1$ and $\Sigma_2$ are symmetric by the mapping: $\lb\mapsto\bar\lb^{-1}$, $\overline{r(\bar\lb^{-1})}$ is naturally analytic in some neighborhood of $\Sigma_2$.
% For the soliton gas solution, when $t=0$,
% as $n\to+\infty$, we have
\begin{align}\label{eq:qnasy}
    q_n(0)=\mathcal{O}(\E^{-cn}), \qquad n \to +\infty,
\end{align}
where $c$ is a positive constant, and 
%As $n\to-\infty$, we have
\begin{align}\label{eq:qn0negn}
	q_n(0)=\frac{\ii^{n+1}}{2}\left(\sqrt{\frac{\eta_2}{\eta_1}}-\sqrt{\frac{\eta_1}{\eta_2}}\right)\Jac\left(\frac{K(k)((n+1)\Omega+\Delta)}{\pi},k\right)+\oo(n^{-1}), \qquad n\to -\infty,
\end{align}
where  $\Jac( z,k)$ given in  \eqref{ndjac} is the subsidiary  Jacobi elliptic function with modulus
    \begin{equation}\label{def:modulus}
        k=\frac{(\eta_1-1)(\eta_2+1)}{(\eta_2-1)(\eta_1+1)},
    \end{equation}
$K$ is the complete elliptic integral defined in  \eqref{ei1},  $\Omega$ is a constant defined in  \eqref{eOmega}, and
\begin{align*}
    \Delta=\frac{ l_2 }{(2\Pi(l_1^2,k)-K(k))\sqrt{(1-l_1^2)(1-l_2^2)}}\int_{\eta_1}^{\eta_2}\frac{\ln r(\i s)\ddd  s }{\sqrt{(\eta_2-s)(s-\eta_1)(s-\eta_1^{-1})(s-\eta_2^{-1})}}
\end{align*}
with 
\begin{align}\label{def:lj}
    l_j:=\frac{\eta_j-1}{\eta_j+1}, \qquad j=1,2.
\end{align}
\end{thm}

\begin{remark}
Since $\Sigma_2$ is the image of $\Sigma_1$ under the mapping $\lb\mapsto\bar\lb^{-1}$, our assumption on $r(\lb)$ in Theorem  \ref{Thm 1} also implies that 
$\overline{r(\bar\lb^{-1})}$ is analytic in some neighborhood of $\Sigma_2$.

    % Here, we only need to assume that $r(\lb)$ is analytically extended to the neighborhood of $\Sigma_1$,  since $\Sigma_1$ and $\Sigma_2$ are symmetric by the mapping: $\lb\mapsto\bar\lb^{-1}$, $\overline{r(\bar\lb^{-1})}$ is naturally analytic in some neighborhood of $\Sigma_2$.
\end{remark}
%Delete
%Furthermore, \hl{in Section} \ref{sec 4}, we provide a global large-time asymptotic description of the soliton gas potential $q_n(t)$ of \eqref{e1.2}.

We finally derive large time asymptotics of $q_n(t)$, which exhibits qualitatively different behaviors in different regions of the $(n+1,t)$-half plane. Depending on the ratio $(n+1)/t$, there are three main regions: one fast decaying region and two genus-1 hyperelliptic wave regions, together with two transition regions between adjacent regions, as described below and illustrated in Figure \ref{fig:region}.

\begin{defn}\label{def region} 
With the constants $k$ and $l_1$ given in \eqref{def:modulus} and \eqref{def:lj}, we define 
	\begin{itemize}
		\item the fast decaying region: $\left\{ (n+1,t): \ 	\frac{n+1}{t}>-\frac{\eta_{1}-\eta_{1}^{-1}}{\ln\eta_{1}}\right\}$;
		\item the first transition region $T_I$: $\left\{(n+1,t):\  \bigcup_{m=0}^{\infty}T^{(m)}_I\right\}$, where $$T^{(m)}_I=\left\{	 (n+1,t):\ -\frac{2m+1}{\ln \eta_{1}}\frac{\ln t}{t}<\frac{n+1}{t}+\frac{\eta_{1}-\eta_{1}^{-1}}{\ln\eta_{1}}< - \frac{2m-1}{\ln \eta_{1}}\frac{\ln t}{t}\right\}$$
       with $m\in \{0\}\cup \mathbb{N}$;
		\item the first genus-1 hyperelliptic wave region $H_I$: $\left\{(n+1,t):\ 	\xi_{\mathrm{crit}}<\frac{n+1}{t}<-\frac{\eta_{1}-\eta_{1}^{-1}}{\ln\eta_{1}}\right\}$;
		\item the second transition region $T_{II}$: $\left\{ (n+1,t):\ |	\frac{n+1}{t} - \xi_{\mathrm{crit}}| < C t^{-2/3}\right\}$, where $C$ is any positive constant;
		\item the second genus-1 hyperelliptic wave region $H_{II}$: $\left\{(n+1,t):\ 	\frac{n+1}{t}<\xi_{\mathrm{crit}}\right\}$.
	\end{itemize}
Here, 
\begin{align}\label{def xi}
	%\xi_{\mathrm{crit}}=\frac{-\frac{(\eta_2 - 1)^2}{\eta_2}\left(\frac{(\eta_2 + 1)^2}{\eta_2}+\frac{(\eta_1 - 1)^2}{\eta_1}\left(\frac{2\Pi(l_1^2,k)}{K(k)}-1\right)+\frac{(\eta_1 + 1)^2}{\eta_1}\left(\frac{E(k)}{K(k)}-1\right)\right)}{\eta_2+\eta_2^{-1}-2 - 4\frac{\Pi(l_1^2,k)}{K(k)}},
    \xi_{\mathrm{crit}}:=\frac{\eta_2 +\frac{1}{\eta_2}+\eta_1 +\frac{1}{\eta_1}}{2}+\frac{-\eta_2^2-\frac{1}{\eta_2^2}+2+\frac{8l_1^2}{(k^2-l_1^2)(l_1^2-1)}\left(\frac{E(k)}{K(k)}+\frac{k^2-l_1^2}{l_1^2}-\frac{k^2-l_1^4}{l_1^2}\frac{\Pi(l_1^2,k)}{K(k)}\right)}{\eta_2+\eta_2^{-1}+2 - 4\frac{\Pi(l_1^2,k)}{K(k)}},
\end{align}
where
%$k$ is given in \eqref{def:modulus}, 
$K$, $E$ and $\Pi$ are the complete elliptic integrals defined in \eqref{ei1} and \eqref{ei2}. 
%A direct numerical calculation shows that $\xi_{\mathrm{crit}}<-\frac{\eta_{1}-\eta_{1}^{-1}}{\ln\eta_{1}}$. 

\end{defn}

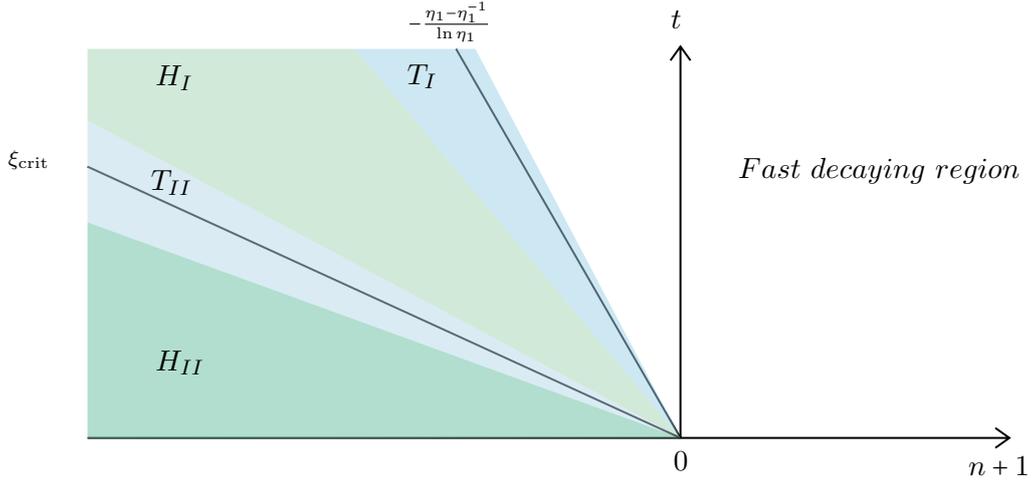
\begin{figure}
    \centering

\tikzset{every picture/.style={line width=0.75pt}} %set default line width to 0.75pt

\begin{tikzpicture}[x=0.75pt,y=0.75pt,yscale=-1,xscale=1]
%uncomment if require: \path (0,310); %set diagram left start at 0, and has height of 310

%Shape: Axis 2D [id:dp19081316815157778]
\draw  (81.39,281) -- (541.6,281)(377.39,84.13) -- (377.39,281) (534.6,276) -- (541.6,281) -- (534.6,286) (372.39,91.13) -- (377.39,84.13) -- (382.39,91.13)  ;
%Straight Lines [id:da6650657193618597]
\draw    (265.39,85.33) -- (377.39,281) ;
%Straight Lines [id:da6317823028841635]
\draw    (81.39,144.53) -- (377.39,281) ;
%Shape: Polygon [id:ds7735808411530944]
%\draw  [color={rgb, 255:red, 0; green, 0; blue, 0 }  ,draw opacity=0 ][fill={rgb, 255:red, 74; green, 144; blue, 226 }  ,fill opacity=0.3 ] (265.39,85.33) -- (377.39,281) -- (214.99,85.33) -- cycle ;
%Shape: Polygon [id:ds22924589138265705]
%\draw  [color={rgb, 255:red, 0; green, 0; blue, 0 }  ,draw opacity=0 ][fill={rgb, 255:red, 155; green, 155; blue, 155 }  ,fill opacity=0.5 ] (214.99,85.33) -- (377.39,281) -- (81.39,121.35) -- (81.39,85.33) -- cycle ;
%Shape: Polygon [id:ds35103513033656397]
%\draw  [color={rgb, 255:red, 0; green, 0; blue, 0 }  ,draw opacity=0 ][fill={rgb, 255:red, 74; green, 144; blue, 226 }  ,fill opacity=0.5 ] (81.39,121.35) -- (377.39,281) -- (81.39,172.55) -- cycle ;
%Shape: Polygon [id:ds9825007318231097]
%\draw  [color={rgb, 255:red, 0; green, 0; blue, 0 }  ,draw opacity=0 ][fill={rgb, 255:red, 74; green, 74; blue, 74 }  ,fill opacity=0.6 ] (81.39,172.55) -- (377.39,281) -- (81.39,281.35) -- cycle ;

\draw  [color={rgb, 255:red, 0; green, 0; blue, 0 }  ,draw opacity=0 ][fill={rgb, 255:red, 155; green, 207; blue, 232  }  ,fill opacity=0.5 ] (275,85.33) -- (377.39,281) -- (214.99,85.33) -- cycle ;

\draw  [color={rgb, 255:red, 0; green, 0; blue, 0 }  ,draw opacity=0 ][fill={rgb, 255:red, 164; green, 211; blue, 179  }  ,fill opacity=0.5 ] (214.99,85.33) -- (377.39,281) -- (81.39,121.35) -- (81.39,85.33) -- cycle ;
%Shape: Polygon [id:ds35103513033656397]
\draw  [color={rgb, 255:red, 0; green, 0; blue, 0 }  ,draw opacity=0 ][fill={rgb, 255:red, 182; green, 215; blue,232 }  ,fill opacity=0.5 ] (81.39,121.35) -- (377.39,281) -- (81.39,172.55) -- cycle ;
%Shape: Polygon [id:ds9825007318231097]
\draw  [color={rgb, 255:red, 0; green, 0; blue, 0 }  ,draw opacity=0 ][fill={rgb, 255:red, 99; green, 190; blue, 157 }  ,fill opacity=0.5 ] (81.39,172.55) -- (377.39,281) -- (81.39,281.35) -- cycle ;

% Text Node
\draw (240,60) node [anchor=north west][inner sep=0.75pt]  [font=\scriptsize]  {$-\frac{\eta_{1} -\eta_{1} ^{-1}}{\ln\eta_1}$};
% Text Node
\draw (40.4,134.6) node [anchor=north west][inner sep=0.75pt]  [font=\scriptsize]  {$\xi_{\mathrm{crit}}$};
% Text Node
\draw (114,92.2) node [anchor=north west][inner sep=0.75pt]    {$H_I$};
% Text Node
\draw (239.6,91.4) node [anchor=north west][inner sep=0.75pt]    {$T_I$};
% Text Node
\draw (404.8,138.8) node [anchor=north west][inner sep=0.75pt]    {$Fast\ decaying\ region$};
% Text Node
\draw (372.4,286.4) node [anchor=north west][inner sep=0.75pt]    {$0$};
% Text Node
\draw (519.6,288.6) node [anchor=north west][inner sep=0.75pt]    {$n+1$};
% Text Node
\draw (371.2,64.4) node [anchor=north west][inner sep=0.75pt]    {$t$};
% Text Node
\draw (111.6,144.8) node [anchor=north west][inner sep=0.75pt]    {$T_{II}$};
% Text Node
\draw (114,235.4) node [anchor=north west][inner sep=0.75pt]    {$H_{II}$};
\end{tikzpicture}
    \caption{\footnotesize Five different asymptotic regions given in Definition \ref{def region}.}
    \label{fig:region}
\end{figure}

Large time asymptotics of $q_n(t)$ in these regions are shown in the following theorem.
\begin{thm}\label{Thm 2}
%Suppose that $r: \Sigma_1 \to \mathbb{R}$ is a continuous, positive function on $\Sigma_1$ with analytic extension to a neighbourhood of $\Sigma_1$. 
Assume that $r(\lb)$ is analytic and positive in a neighborhood of $\Sigma_1$. As $t\to +\infty$, we have the following asymptotics of $q_n(t)$ in the regions given in Definition \ref{def region}. 
% Then the large-time asymptotics of the soliton gas AL solution $q_n(t)$ 
%  are shown as follows.
	\begin{itemize}
		\item[\rm (1)] For $\frac{n+1}{t}$ in the fast decaying region,  there exists a  positive constant $c$ such that 
		\begin{align}\label{eq:decay}
		    q_n(t)=\oo(\E^{-ct}).
		\end{align}
		\item[\rm (2)]  
    For $\frac{n+1}{t} \in T_I^{(m)} \subseteq T_I$, $m\in \{0\}\cup \mathbb{N}$, we have
			\begin{align}\label{mr1}
		    q_n(t)=\oo	\left( \min(\E^{(2m-1)\ln t +  t(\frac{n+1}{t}+\frac{\eta_1-\eta_1^{-1}}{\ln \eta_1})\ln \eta_1},\E^{-(2m+1)\ln t - t(\frac{n+1}{t}+\frac{\eta_1-\eta_1^{-1}}{\ln \eta_1})\ln \eta_1} )\right).
		\end{align}
	
		\item[\rm (3)] For $\frac{n+1}{t} \in H_I$, we have
		\begin{align}\label{mr2}
		    q_n(t)=\frac{\E^{2\ii t(1+\frac{\pi(n+1)}{4t})}}{2}\left(\sqrt{\frac{\alpha(\xi)}{\eta_1}}-\sqrt{\frac{\eta_1}{\alpha(\xi)}}\right)	\Jac\left(\frac{K(k(\xi))(t\Omega+\Delta)}{\pi},k(\xi)\right)+\oo(t^{-1}),
		\end{align}
where  $\alpha(\xi)$ is a real-valued function  defined in \eqref{def alpha} with $\xi=(n+1)/t$, $\Jac(z,k(\xi))$ given in  \eqref{ndjac} is the subsidiary Jacobi elliptic function with modulus
    \begin{equation}\label{def:kxi}
    k(\xi)=\frac{(\eta_1-1)(\alpha(\xi)+1)}{(\alpha(\xi)-1)(\eta_1+1)},
    \end{equation}
    and $\Omega$ and $\Delta$ are constants given in \eqref{def Omegj} and \eqref{def Delta}, respectively.

\item[\rm (4)]  For $\frac{n+1}{t} \in T_{II}$, we have
			\begin{align}\label{mr3}
		    q_n(t)=\frac{\E^{2\ii t(1+\frac{\pi(n+1)}{4t})}}{2}\left(\sqrt{\frac{\eta_2}{\eta_1}}-\sqrt{\frac{\eta_1}{\eta_2}}\right)\Jac\left(\frac{K(k)(t\Omega+\Delta)}{\pi},k\right)+\oo(t^{-1/3}),
		\end{align}
where  $k$, $\Omega$ and $\Delta$ are given in \eqref{def:kxi}, \eqref{def Omegj} and \eqref{def Delta}, respectively.

    % by the same expressions as for $\frac{n+1}{t} \in H_{I}$ under $\alpha(\xi)=\eta_2$.
		\item[\rm (5)]  For $\frac{n+1}{t} \in H_{II}$, we have  \begin{align}\label{mr4}
		    q_n(t)=\frac{\E^{2\ii t(1+\frac{\pi(n+1)}{4t})}}{2}\left(\sqrt{\frac{\eta_2}{\eta_1}}-\sqrt{\frac{\eta_1}{\eta_2}}\right)\Jac\left(\frac{K(k)(t\Omega+\Delta)}{\pi},k\right)+\oo(t^{-1}),
		\end{align}
	where $k$, $\Omega$ and $\Delta$ are given in \eqref{def:kxi}, \eqref{def Omegj} and \eqref{def Delta}, respectively.
	\end{itemize}
\end{thm}

Since $\alpha(\xi)$ takes different values in different regions, it is readily seen from \eqref{mr2} and \eqref{mr4} that the regions $H_{I}$ and $H_{II}$ correspond to the genus-1 hyperelliptic wave region with modulated and  constant coefficients, respectively. In addition,  as $\frac{n+1}{t} \to -\frac{\eta_{1}-\eta_{1}^{-1}}{\ln\eta_{1}}$, the saddle point of the phase function coalesces with the endpoint $\i \eta_1$  of $\Sigma_1$, while as $\frac{n+1}{t} \to \xi_{\mathrm{crit}}$, it approaches the other endpoint $\i \eta_2$; see Figure \ref{fig reg} below for an illustration. This leads to the introduction of transition regions $T_I$ and $T_{II}$, which requires new ingredients of analysis. It also explains why the error estimates differ in the regions $T_{II}$ and $H_{II}$, although the leading asymptotics are the same.

In Theorem \ref{Thm 2}, we only present the leading asymptotics of the soliton gas solution $q_n$ with error estimates for large $t$, which is the main focus of the present work. Our asymptotic analysis, leading to the proof of the above theorem, also in principal allows us to establish correction terms explicitly in the expansion. All the ingredients for obtaining such results are presented, but we will not write the details down neither comment them any further. It is worthwhile to mention that detailed calculations of the sub-leading term in large time asymptotics of the soliton gas solution for the mKdV equation in the transition regions can be found in \cite{mkdvtran}.

We emphasize that providing uniform asymptotics in transition regions is usually a difficult problem, which is fully resolved only in some particular cases; cf. \cite{mis,RN6,HM1981,wang2023defocusing,xu2024transient}. In the context of focusing AL system, the analysis in transition regions $T_I$ and $T_{II}$ involves RH problems relevant to the generalized Laguerre polynomial and Painlev\'e XXXIV equation, respectively, which is different from the classical local parametrices used in other regions. The reason why we divide $T_I$ into different subregions $T_I^{(m)}$ is that the index $m$ is exactly the degree of Laguerre polynomials related to the local analysis. Thus, although $T_I$ is relatively small in size, it is itself composed of several layers. It is worthwhile to point out that the transitional asymptotics of the mKdV equation for step-like initial data is also related to the RH problem built from Laguerre polynomials \cite{BM2018}, and the Painlev\'e \uppercase\expandafter{\romannumeral34\relax} transcendents play an important role in asymptotic studies of critical behaviors arsing from integrable differential equations \cite{FLYZ2006}, random unitary ensembles \cite{ikj2008}  and orthogonal polynomials \cite{XZ2011}. 

The rest of this paper is organized as follows. In Section \ref{sec 2}, we consider a sequence of RH problems, indexed by $N$, characterizing the pure $N$-soliton solution $q^{[N]}_{n}$, whose spectrum is confined to the interval $\Sigma_1\cup\Sigma_2$. As $N\to +\infty$, the soliton gas solution $q_n$ is then characterized by the limiting RH problem for $Z$, which serves as the starting point of further asymptotic analysis. In Section \ref{sec fredhlom}, we introduce the tau-functions and establish the Fredholm determinant representation of $q^{[N]}_{n}$. This structure is preserved in the large $N$ limit, which leads to the proof of Theorem \ref{thm1}. We perform Deift-Zhou steepest descent analysis \cite{RN6,RN10} to RH problems for $Z(\lb; n,0)$ and $Z(\lb; n+1,t)$ in Sections \ref{sec 3} and \ref{sec 4}, respectively. The main idea of the analysis is to transform RH problem for $Z$ into a solvable form consisting of the global and local RH problems. A key step in this procedure is to apply the $g$-function mechanism \cite{gfunction} to arrive at a global RH problem solvable in terms of the Jacobi theta function. As aforementioned, one needs special treatments for the analysis in the transition regions $T_I$ and $T_{II}$. The outcome of our analysis is the proofs of Theorems \ref{Thm 1} and \ref{Thm 2}, presented in Section \ref{sec:proof}.

\paragraph{Notations}
Throughout this paper, the following notations will be used.
\begin{itemize}
	\item The Pauli matrices:
	\begin{align}\label{def:Pauli}
		\sigma_1=\left(\begin{matrix}
			0&1\\1&0
		\end{matrix}\right),\quad\sigma_2=\left(\begin{matrix}
			0&-\ii\\\ii&0
		\end{matrix}\right),\quad\sigma_3=\left(\begin{matrix}
			1&0\\0&-1
		\end{matrix}\right).
	\end{align}
    For a $2 \times 2$ matrix $A$, we also set 
    \begin{equation*}
        \e^{\hat{\sigma}_j} A:= \e^{\sigma_j} A  \e^{-\sigma_j}, \qquad j=1,2,3.
    \end{equation*}
%\item The function $R(\lb)=\sqrt{(\lb-\ii\eta_1)(\lb-\ii\eta_2)(\lb-\ii/\eta_1)(\lb-\ii/\eta_2)}$.
%\item \hl{The} phase function for the RH problem of the focusing AL system \eqref{e1.2}:
%\begin{align}\label{phi}
%	\phi(\lb; n,t)=-\ii t(\lb+\lb^{-1}-2)+n\ln\lb.
%\end{align}
\item  Assume $1-k^2 \in \mathbb{C}\setminus (-\infty,0]$ and $1-k^2 \sin^2 m \in \mathbb{C}\setminus (-\infty,0]$, except that one of them may be $0$,
and $1-\alpha^2 \in \mathbb{C}\setminus {0}$. Then,
the complete and incomplete elliptic integrals are defined by
\begin{align}
	&K(k)=\int_{0}^{1}\frac{\ddd  s}{\sqrt{(1-s^2)(1-k^2s^2)}},\quad E(k)=\int_{0}^{1}\sqrt{\frac{1-k^2s^2}{1-s^2}}\ddd  s,\label{ei1}\\
    &F(m,k)=\int_{0}^{m}\frac{\ddd  s}{\sqrt{(1-s^2)(1-k^2s^2)}},\\
	&\Pi(\alpha^2,k)=\int_{0}^{1}\frac{\ddd  s}{(1-\alpha^2 s^2)\sqrt{(1-s^2)(1-k^2s^2)}}.\label{ei2}
\end{align}
\item If $A$ is a matrix, then $(A)_{ij}$ stands for its $(i,j)$-th entry.
\item In Sections \ref{sec 3} and \ref{sec 4}, we adopt the same notations (such as $g$, $\delta$, $T$, $Z^{(1)}$, $Z^{(\infty)}$, $Z^{(p)}$, \dots  ) during the analysis,  which should be understood in different contexts. We believe this will not lead
to any confusion.
\end{itemize}

%\section{Lax pair and the soliton gas RH problem}
\section{The soliton gas solution and its RH characterization}

\label{sec 2}
In this section, we review some basic results of the focusing AL system \eqref{e1.2}, including an RH problem for the $N$-soliton solution $q^{[N]}_{n}$
\cite{APT2004,CFW2025}. By taking $N\to+\infty$, we 
interpret the $(1,2)$-th entry of the limiting RH problem at the origin as the soliton gas solution. We justify this claim by checking that this entry indeed solves the focusing AL system \eqref{e1.2} and can be obtained as the large $N$ limit of $q^{[N]}_{n}$. 

% derive RH problem \ref{RHP0} corresponding to the soliton gas solution of \eqref{e1.2}.
\begin{comment}
It is known that the Lax pair of the focusing AL system \eqref{e1.2} reads
\begin{align}
	&SX=AX,\qquad \frac{\ddd  X}{\ddd  t}=BX,
\end{align}
where 
\begin{equation}\label{def:rsoprator}
    Sf(n)=f(n+1)
\end{equation}
is the right-shift operator, and
\begin{align*}
	&A=z^{\sigma_3}+Q,\ \ B=\frac{\sigma_3}{\ii}\left(\frac{(z-z^{-1})^2}{2}-Q(S^{-1}Q)+z^{\sigma_3}Q-
	(S^{-1}Q)z^{\sigma_3}\right),
\end{align*}
with 
\begin{align*}
	Q:=Q(n,t)=\left(\begin{matrix}
		0&q_n(t)\\-\overline{q_n(t)}&0
	\end{matrix}\right)
\end{align*}
being a skew-Hermitian matrix.   For ease of handling and to reduce symmetry, set $\lb=z^2$.
\end{comment}

From \cite{CFW2025}, it follows that 
\begin{align}\label{eq:Nsoliton}
		q^{[N]}_n(t)=M_{12}(0;n+1,t),
\end{align}
where $M(\lb;n,t)$ solves the following RH problem. 
% The RH problem associated with the $N$-soliton solution of \eqref{e1.2} is established in \cite{CFW2025} as follows.
%In \cite{CFW2025}, we established the RH problem  associated to the $N$-soliton solution of \eqref{e1.2} as follows.
%In the following, we see that RH problem \ref{rNsolition} is associated to the $N$-soliton solution of \eqref{e1.2}, while RH problem \ref{RHP0} corresponds to the soliton gas solution of \eqref{e1.2}.  Furthermore, as $N\to+\infty$, the solution of RH problem \ref{rNsolition} tends to that of RH problem \ref{RHP0} for any $(n,t)\in\mathbb{Z}\times\mathbb{R}_+$.
\begin{Rhp}\label{rNsolition}\
	\begin{itemize}
		\item $M(\lb)=M(\lb;n,t)$ is meromorphic in $\ccc$ with simple poles at  $\{\ii\lambda_j, \ii\lambda_j^{-1}\}_{j=1}^N$ with $\lb_j>1$.
        % and at the corresponding conjugate points $\ccc\setminus\{\ii\lambda_j^{-1}\}_{j=1}^N$.%is analytic in $\ccc\setminus\{\ii\lambda_j,\ii\lambda_j^{-1}\}_{j=1}^N$ with $\lb_j>1$.
		\item
        $M(\lb)$  satisfies the residue conditions 
		\begin{align*}
			&\underset{\lambda=\i\lb_j}{\res}M(\lb)= \lim_{\lb\to \ii \lb_j}M(\lb)\left(\begin{matrix}
				0&0\\-\ii\varLambda_j\E^{-\phi(\lb)}&0
			\end{matrix}\right),\\
		&\underset{\lambda=\i\lb_j^{-1}}{\res}M(\lb)=\lim_{\lb\to \ii \lb_j^{-1}}M(\lb)\left(\begin{matrix}
			0&-\ii\varLambda_j\lb_j^{-2}\E^{\phi(\lb)}\\0&0
		\end{matrix}\right),
		\end{align*}
	where $\varLambda_j$, $j=1,\cdots,N$, are nonzero  norming constants and
    % $\varLambda_j$, $j=1,\cdots,N$, and phase function $\phi(\lb):=\phi(\lb;n,t)$ given by
        \begin{align}\label{phi}
	\phi(\lb)=\phi(\lb;n,t):=-\ii t(\lb+\lb^{-1}-2)+n\ln\lb.
\end{align}
		\item As $\lb\to\infty$, we have
		$
			M(\lambda)=I+\oo(\lambda^{-1}).
		$
	\end{itemize}
\end{Rhp}
%\begin{assumption}
%	Given $\eta_2>\eta_1>1$, set $\Sigma_1=\ii[\eta_1,\eta_2]$ and $\Sigma_2=\ii[\eta_2^{-1},\eta_1^{-1}]$ as two upwards oriented contours.	Assume that $r(\lb)$ is holomorphic in the neighborhood of $\ii[\eta_1,\eta_2]$ and  $r(\lambda)\in H^1(\ii[\eta_1,\eta_2])$.
%\end{assumption}
%In the above RH problem, under the Assumption \ref{assum1}, if we let $\lb_j$ and $\Lambda_j$ satisfy the conditions in the following proposition, then, ones find that the summations in \eqref{e2.2} converge to some Riemann integrals, which play a key role in the limit shown in Proposition \ref{p2.3}. 

Suppose that  $r(\lb)$ is positive and continuous on $[\i\eta_1,\i\eta_2]$,
and set
	\begin{align}\label{lbva}
		\lb_j=\eta_1+\frac{j-1}{N}(\eta_2-\eta_1),\quad \varLambda_j=\frac{r(\i\lb_j)}{2\pi N},\qquad  j=1,\cdots,N.
	\end{align}
It is then readily seen that
\begin{subequations}\label{e2.2}
        \begin{align}
		&\lim_{N\to+\infty}\sum_{j=1}^{N}\frac{\i\Lambda_j}{\lb-\i\lb_j}=\frac{1}{2\pi\i}\int_{\Sigma_1}\frac{\i r(s)}{\lb-s}\d s,
        \qquad \lb\in\ccc\setminus \Sigma_1
        \\
		&\lim_{N\to+\infty}\sum_{j=1}^{N}\frac{\i\Lambda_j\lb_j^{-2}}{\lb-\i\lb_j}=\frac{1}{2\pi\i}\int_{\Sigma_2}\frac{\i r(\bar s^{-1})}{\lb-s}\d s,  \qquad \lb\in\ccc\setminus \Sigma_2,
	\end{align}
    \end{subequations}
where $\Sigma_i$, $i=1,2$, are defined in \eqref{def sigma1}. 
This, together with RH problem \ref{rNsolition}, leads us to consider the following limiting RH problem. 
% If $\lb_j$ and $\Lambda_j$, $j=1,\cdots,N$, in the above RH problem satisfy the conditions stated in the following proposition, then the sums in \eqref{e2.2} converge to certain Riemann integrals, which describe the continuous spectral condition of the soliton gas RH problem \ref{RHP0} with the jump contour $\Sigma_1\cup\Sigma_2$ oriented upwards.
% The soliton gas solution $q_n(t)$--the limit of $N$-soliton solution $q^{[N]}_n(t)$--can be recovered. 
% Additionally, Proposition \ref{Pro 1} demonstrates that the soliton gas solution exactly solves the focusing AL system.
% \begin{prop}\label{p2.2}
% 	Suppose that $r(\lb)$ is continuous on $[\i\eta_1,\i\eta_2]$.
% 	Let
% 	\begin{align*}
% 		\lb_j=\eta_1+\frac{j-1}{N}(\eta_2-\eta_1),\quad \varLambda_j=\frac{r(\i\lb_j)}{2\pi N},\quad  j=1,\cdots,N.
% 	\end{align*}
% 	Then, for $\lb\in\ccc\setminus(\Sigma_1\cup\Sigma_2)$, we have the following identities:
%     \begin{subequations}\label{e2.2}
%         \begin{align}
% 		&\lim_{N\to+\infty}\sum_{j=1}^{N}\frac{\i\Lambda_j}{\lb-\i\lb_j}=\frac{1}{2\pi\i}\int_{\Sigma_1}\frac{\i r(s)}{\lb-s}\d s, 
%         \\
%         &\lim_{N\to+\infty}\sum_{j=1}^{N}\frac{\i\Lambda_j\lb_j^{-2}}{\lb-\i\lb_j}=\frac{1}{2\pi\i}\int_{\Sigma_2}\frac{\i r(\bar s^{-1})}{\lb-s}\d s.
% 	\end{align}
%     \end{subequations}
% \end{prop}
% \begin{proof}
% 	This proof is obtained by directly computing the Riemann integrals on $\Sigma_1$ and $\Sigma_2$, respectively.
% \end{proof}

% The soliton gas RH problem for the focusing AL system \eqref{e1.2} is stated as follows:
\begin{Rhp}\label{RHP0}\
	\begin{itemize}
		\item $Z(\lb)=Z(\lb;n,t)$ is analytic in $\ccc\setminus(\Sigma_1\cup\Sigma_2)$.
		\item For $\lb\in\Sigma_1\cup\Sigma_2$, $Z(\lb)$ satisfies the jump condition 
        $	Z_+(\lb)=Z_-(\lb)J(\lb),$
		where
		\begin{align}\label{J}
			J(\lb)=\begin{cases}
				\left(\begin{matrix}
					1&0\\\ii r(\lb)\E^{-\phi(\lb)}&1
				\end{matrix}\right),&\lb\in\Sigma_1,\\
				\left(\begin{matrix}
					1&\ii r(\bar\lb^{-1})\E^{\phi(\lb)}\\0&1
				\end{matrix}\right),&\lb\in\Sigma_2.
			\end{cases}
		\end{align}
		\item As $\lb\to\infty$, we have
		$
			Z(\lb)=I+\oo(\lb^{-1}).
		$
	\end{itemize}
\end{Rhp}

Our next proposition shows that the above RH problem characterizes the soliton gas solution of the focusing AL system \eqref{e1.2}.  
\begin{prop}\label{Pro 1}
Suppose that $r(\lb)$ is positive and continuous on $[\i\eta_1,\i\eta_2]$, then RH problem \ref{RHP0} is uniquely solvable and satisfies the symmetry relation
	\begin{align}
Z(\lb)=\sigma_2\overline{Z(0)^{-1}Z(\bar\lb^{-1})}\sigma_2.\label{sy Z}
	\end{align}
Furthermore, 
%the soliton gas solution corresponding to the RH problem \ref{RHP0} is given by the reconstruction formula
	\begin{align}\label{rec}
		q_n(t):= Z_{12}(0;n+1,t)
	\end{align}
is the soliton gas solution of the focusing AL system \eqref{e1.2} and for any $(n,t)\in\Z\times\R$, we have
	\begin{align}\label{NtoSG}
		\lim_{N\to\infty}q^{[N]}_n(t)=q_n(t).
	\end{align}

\end{prop}
\begin{proof}
First, it is easy to check that the orientation of $\Sigma_1\cup\Sigma_2$ is reversed by the mapping $\lb\to\bar\lb^{-1}$. Since the jump matrix satisfies the skew-Hermitian symmetry relation $\left(J(\bar\lb^{-1})\right)^\dagger=J(\lb)^{-1}$ for $\lb\in \Sigma_1\cup\Sigma_2$, Lemma 4.4 in \cite{CFW2025} then guarantees that RH problem \ref{RHP0} is uniquely solvable.

Next, to ensure $q_n$ defined in \eqref{rec} is a solution of the AL system \eqref{e1.2}, we will show that, starting from $Z$, one can define a $2\times2$ matrix-valued function $\Phi$ satisfying the Lax pair equations
\begin{align}\label{lax Z}
    S\Phi=A \Phi,\qquad \frac{\ddd  \Phi}{\ddd  t}=B\Phi, 
\end{align}
where the coefficients $A$ and $B$ are obtained from $Z$, and $S$ is the right-shift operator defined by
\begin{equation}\label{defs}
    Sf(n)=f(n+1).
\end{equation}
The compatibility condition of \eqref{lax Z} then gives us the AL system \eqref{e1.2}.

% and whose compatibility condition is the AL system \eqref{e1.2}. Here, we recall that $S$ is the right-shift operator: $Sf(n)=f(n+1)$.

To proceed, we note that the symmetry relation \eqref{sy Z} reads 
\begin{align*}
   	Z(\lb)=\sigma_2\overline{Z(0)^{-1}Z(\lb)}\sigma_2, \qquad |\lb|=1,
\end{align*}
from which it follows that 
\begin{align*}
   & \sigma_2\overline{Z(\lb)}^{-1}\sigma_2=Z(\lb)^{-1}\sigma_2\overline{Z(0)}^{-1}\sigma_2,\qquad Z(0)=Z(\lb)\sigma_2\overline{Z(\lb)}^{-1}\sigma_2.
\end{align*}
Combining these two identities yields
\begin{align*}
    Z(0)=\sigma_2\overline{Z(0)}^{-1}\sigma_2.
\end{align*}
Denote
\begin{align}
	Z_{12}(0;n)=q_{n-1}, \qquad  Z_{22}(0;n)=c_{n-1}.\label{Zij}
\end{align}
Together with the fact $\det[Z]\equiv1$, we have
\begin{align}\label{def:Z21}
Z_{21}(0;n)=\overline{q_{n-1}},\qquad Z_{11}(0;n)=\frac{1+|q_{n-1}|^2}{c_{n-1}}.
\end{align}
Since $\E^{\phi(\lb)/2}$ is not analytic in $\mathbb{C}$, we introduce the $z$-plane by setting
\begin{align*}
    \lb=z^2,
\end{align*}
and define
\begin{align*}
	\Phi(z):=Z(z^2)\E^{\phi(z^2)\sigma_3/2},
\end{align*}
which admits the jump $\E^{\phi(z^2)\hat{\sigma}_3/2} V$ for $z^2\in \Sigma_1\cup\Sigma_2$ and satisfies $\Phi=(I+\mathcal{O}(z^{-2}))\E^{\phi(z^2)\sigma_3/2}$ as $z\to\infty$. Note that the jump of $\Phi$ is independent of $n$ and $t$, the function
\begin{align*}
	\Phi(z;n+1,t)\Phi(z;n,t)^{-1}z^{-\sigma_3}=Z(z^2;n+1,t)z^{\sigma_3}Z(z^2;n,t)^{-1}z^{-\sigma_3}
\end{align*}
% has no jump on $z$-plane, and only may have poles on $\infty$ and $0$.  
is analytic in the $z$-plane except for possible poles at $z=0$ and $z=\infty$. 

To analyze its behavior near these points, we denote
\begin{align*}
    Z(\lb; n)=I+Z_1/\lb+\oo(\lb^{-2}), \qquad \lb\to \infty. 
\end{align*}
As $z\to\infty$, we have
\begin{align*}
	\Phi(z;n+1,t)\Phi(z;n,t)^{-1}z^{-\sigma_3} &= (I+Z_1(n+1)/z^2+\mathcal{O}(z^{-4}))(I+z^{\sigma_3}Z_1(n)z^{-\sigma_3}/z^2+\mathcal{O}(z^{-3}))^{-1},\\
	& =	I+\begin{pmatrix}
		0& -Z_{1,12}(n)\\
		0 &0\\
	\end{pmatrix}+\mathcal{O}(z^{-2}).
\end{align*}
As $z\to0$,  it follows from \eqref{def:Z21} that
\begin{align*}
	&\Phi(z;n+1,t)\Phi(z;n,t)^{-1}z^{-\sigma_3}
	=(Z(0;n+1)+\mathcal{O}(z^{2}))(z^{\sigma_3}Z(0;n)z^{-\sigma_3}+\mathcal{O}(z))^{-1}\\
	&=(Z(0;n+1)+\mathcal{O}(z^{2}))\left(z^{-2}\begin{pmatrix}
		0& 0\\
		-\overline{q_{n-1}} &0\\
	\end{pmatrix}+\mathcal{O}(1)\right)
	=z^{-2}Z(0;n+1)\begin{pmatrix}
		0& 0\\
		-\overline{q_{n-1}} &0\\
	\end{pmatrix}+\mathcal{O}(1).
\end{align*}
Consequently, one has
\begin{align*}
	\Phi(z;n+1,t)\Phi(z;n,t)^{-1}z^{-\sigma_3}=z^{-2}Z(0;n+1)\begin{pmatrix}
		0& 0\\
		-\overline{q_{n-1}} &0\\
	\end{pmatrix}+	I+\begin{pmatrix}
	0& -Z_{1,12}(n)\\
	0 &0\\
	\end{pmatrix},
\end{align*}
or equivalently,  
\begin{align}
    S\Phi=A\Phi,\label{laxphi 1}
\end{align}
where 
\begin{align*}
A
	= z^{\sigma_3}-z^{-1}\begin{pmatrix}
		\overline{q_{n-1}}q_{n}& Z_{1,12}(n)\\
		\overline{q_{n-1}}c_{n} &0\\
	\end{pmatrix}.
\end{align*}
Similarly, consider
\begin{align*}
	\frac{\ddd  \Phi(z;n,t)}{\ddd  t}\Phi(z;n,t)^{-1}&=(Z(z^2;n,t)\E^{\phi(z^2;n,t)\sigma_3/2})_t\E^{-\phi(z^2;n,t)\sigma_3/2}Z(z^2;n,t)^{-1}\\
	&=\frac{\ddd  Z(z^2;n,t)}{\ddd  t}Z(z^2;n,t)^{-1}-\frac{\i(z^2-2+z^{-2})}{2}Z(z^2;n,t)\sigma_3Z(z^2;n,t)^{-1},
\end{align*}
which is analytic in the $z$-plane except for possible poles at $z=0$ and $z=\infty$. As $z\to\infty$, we have
\begin{align*}
	\frac{\ddd  \Phi(z;n,t)}{\ddd  t}\Phi(z;n,t)^{-1}&=-\frac{\i(z^2-2+z^{-2})}{2}(I+Z_1(n)/z^2+\mathcal{O}(z^{-4}))\sigma_3(I+Z_1(n)/z^2+\mathcal{O}(z^{-4}))^{-1}+\mathcal{O}(z^{-2})\\
	&=-\frac{\i}{2}z^2\sigma_3-\frac{\i}{2}(-2\sigma_3+\sigma_3Z_1(n)^{-1}+Z_1(n)\sigma_3)+\mathcal{O}(z^{-2})\\
	&=-\frac{\i}{2}(z^2-2)\sigma_3-\frac{\i}{2}\begin{pmatrix}
		 Z_{1,11}(n)+ Z_{1,22}(n)& -2Z_{1,12}(n)\\
		2Z_{1,21}(n) &-(Z_{1,11}(n)+ Z_{1,22}(n))\\
	\end{pmatrix}+\mathcal{O}(z^{-2}),
\end{align*}
and as $z\to0$,  we have
\begin{align*}
		\frac{\ddd  \Phi(z;n,t) }{\ddd  t}\Phi(z;n,t)^{-1}=&-\frac{\i(z^2-2+z^{-2})}{2}(Z(0;n)+\mathcal{O}(z^{2}))\sigma_3(Z(0;n)+\mathcal{O}(z^{2}))^{-1}\\
		=&-\frac{\i(z^2-2+z^{-2})}{2}Z(0;n)\sigma_3Z(0;n)^{-1}+\mathcal{O}(1)\\
		=&-\frac{\i}{2}z^{-2}\begin{pmatrix}
			1+2|q_{n-1}|^2& -2\frac{1+|q_{n-1}|^2}{c_{n-1}}q_{n-1}\\
			2\overline{q_{n-1}}c_{n-1} &-(1+2|q_{n-1}|^2)\\
		\end{pmatrix}+\mathcal{O}(1).
\end{align*}
It thereby holds that
\begin{align}
	\frac{\ddd  \Phi}{\ddd  t}&=B\Phi,\label{laxphi 2}
\end{align}
where
%\begin{align*}
%\hat{B}&=\left(-\frac{\i}{2}z^{-2}\begin{pmatrix}
%		1+2|q_{n-1}|^2& -2\frac{1+|q_{n-1}|^2}{c_{n-1}}q_{n-1}\\
%		2\overline{q_{n-1}}c_{n-1}&-(1+2|q_{n-1}|^2)\\
%	\end{pmatrix} \right.\\
%	&\left.-\frac{\i}{2}(z^2-2)\sigma_3-\frac{\i}{2}\begin{pmatrix}
%	Z_{1,11}(n)+ Z_{1,22}(n)& -2Z_{1,12}(n)\\
%	2Z_{1,21}(n) &-(Z_{1,11}(n)+ Z_{1,22}(n))\\
%	\end{pmatrix}\right) .
%\end{align*}
\begin{align*}
B=&-\frac{\i}{2}z^{-2}\begin{pmatrix}
		1+2|q_{n-1}|^2& -2\frac{1+|q_{n-1}|^2}{c_{n-1}}q_{n-1}\\
		2\overline{q_{n-1}}c_{n-1}&-(1+2|q_{n-1}|^2)\\
	\end{pmatrix}-\frac{\i}{2}(z^2-2)\sigma_3 \\
	&-\frac{\i}{2}\begin{pmatrix}
	Z_{1,11}(n)+ Z_{1,22}(n)& -2Z_{1,12}(n)\\
	2Z_{1,21}(n) &-(Z_{1,11}(n)+ Z_{1,22}(n))
	\end{pmatrix}.
\end{align*}

The compatibility condition of equations \eqref{laxphi 1} and \eqref{laxphi 2} then yields
\begin{align}
    \frac{\ddd A}{\ddd  t}+AB=(SB)A.\label{ccequ}
\end{align}
Comparing the coefficients of the $\mathcal{O}(z)$ and $\mathcal{O}(z^{-3})$ terms on both sides of \eqref{ccequ}, we obtain
\begin{align}
    Z_{1,21}(n)=\overline{q_{n-1}}c_n, \label{cz1}
\qquad
   c_n(1+|q_{n-1}|^2)=c_{n-1}.
\end{align}
Substituting the above two equalities into the coefficient of the $\mathcal{O}(z^{-1})$ term in \eqref{ccequ} yields the AL system \eqref{e1.2}.

Finally, it remains to show \eqref{NtoSG}. Let $\gamma_1$ be a closed curve encircling the interval $\Sigma_1$ in a clockwise manner and set $\gamma_2:=\overline{\gamma_1^{-1}}$ which encircles $\Sigma_2$.  We then define 
\begin{align*}
		T^M(\lb)=\begin{cases}
			\left(\begin{matrix}
				1&0\\\sum_{j=1}^N\frac{\i\Lambda_j\E^{-\phi(\lb)}}{\lb-\i\lb_j}&1
			\end{matrix}\right),&\lb\in D_1,\\
			\left(\begin{matrix}
				1&\sum_{j=1}^{N}\frac{\i\Lambda_j\lb_j^{-2}}{\lb-\i\lb_j^{-1}}\E^{\phi(\lb)}\\0&1
			\end{matrix}\right),&\lb\in D_2,\\
			I,&\text{otherwise},
		\end{cases}\quad T^Z(\lb)=\begin{cases}
		\left(\begin{matrix}
			1&0\\-\int_{\Sigma_1}\frac{\i r(s)\E^{-\phi(\lb)}\d s}{2\pi\i(s-\lb)}&1
		\end{matrix}\right),&\lb\in D_1,\\
		\left(\begin{matrix}
			1&-\int_{\Sigma_2}\frac{\i r(s)\E^{\phi(\lb)}\d s}{2\pi\i(s-\lb)}\\0&1
		\end{matrix}\right),&\lb\in D_2,\\
		I,&\text{otherwise},
	\end{cases}
	\end{align*}
where $D_i$, $i=1,2$, denotes the region between $\gamma_i$ and $\Sigma_i$, and
% We first introduce a clockwise-oriented contour $\gamma_1$ that closely encircles $\Sigma_1$. Define $\gamma_2:=\overline{\gamma_1^{-1}}$, then $\gamma_2$ is also clockwise-oriented and encircles  $\Sigma_2$.
%  Let $D_j\;(j=1,2)$
% 	denote the region between $\gamma_j$ and $\Sigma_j$ 
% 	and introduce the matrix-valued functions 
%For the solutions of RH problem \ref{rNsolition} and \ref{RHP0}, we perform the following transformations:
	\begin{align}\label{e2.15}
		\tilde M(\lb):=M(\lb)T^M(\lb),\qquad
	\tilde Z(\lb):=Z(\lb)T^Z(\lb).
	\end{align}
It is readily seen that $\tilde M(\lb)$ has no pole at $\lb_j$ and $\tilde Z(\lb)$ has no jump on $\Sigma_1\cup\Sigma_2$.
%, the function $\mathcal{E}(\lb)$ has neither poles nor jumps on $\Sigma_1\cup\Sigma_2$. 
Indeed,  the first identity in \eqref{e2.15} implies that $\tilde M(\lb)$ may have  at most  simple poles at each $\lb_j$.
    %from the first formula in \eqref{e2.15}, 
    %Exactly, on the one hand, seeing the first formula in \eqref{e2.15}, 
  %  we see that $\tilde M(\lb)$ can have 
    %, and by directly computing the residue at any $\lb_j$ for $\tilde M(\lb)$, we notice that $\tilde M(\lb)$ has no pole at $\lb_j$.
However, a direct computation of the residue condition at each $\lb_j$ shows that $\tilde{M}(\lb)$ is pole-free at $\lb_j$. On the other hand, by the Sokhotski-Plemelj formula, for any $\lb\in\Sigma_1$, one has
    \begin{align*}
        \tilde Z_+(\lb)=\tilde Z_-(\lb)\left(\begin{matrix}
            1&0\\\int_{\Sigma_1}\frac{\i r(s)\E^{-\phi(\lb)}\d s}{2\pi\i(s-\lb_-)}-\int_{\Sigma_1}\frac{\i r(s)\E^{-\phi(\lb)}\d s}{2\pi\i(s-\lb_+)}+\ii r(\lb)\E^{-\phi(\lb)}&1
        \end{matrix}\right)=\tilde Z_-(\lb),
    \end{align*}
so $\tilde Z(\lb)$ has no jump on $\Sigma_1$. The same argument shows that $\tilde Z(\lb)$ has no jump on $\Sigma_2$ either. As a consequence, the function $\mathcal{E}(\lb):=\tilde M(\lb)\tilde Z(\lb)^{-1}$ satisfies an RH problem with the jump contour given by $\gamma_1 \cup \gamma_2$ and the jump matrix $J^{\mathcal{E}}$ reads
	\begin{align}\label{e2.33}
		J^{\mathcal{E}}(\lb)-I=\tilde Z_-(\lb)T_-^M(\lb)^{-1}T_-^Z(\lb)\tilde Z_-(\lb)^{-1}-I=\tilde Z_-(\lb)(T_-^Z(\lb)-T_-^M(\lb))\tilde Z_-(\lb)^{-1}.
	\end{align}
By \eqref{e2.2}, the right-hand side of \eqref{e2.33} tends to zero as $N\to\infty$. Then, the RH problem for $\mathcal{E}$ is a small-norm RH problem for large $N$. Thus, by the small-norm RH problem theory \cite{RN10}, as $N\to\infty$, $M(0;n,t)\to Z(0;n,t)$ for any $(n,t)\in\Z\times\R$, which leads to \eqref{NtoSG}. 

This completes the proof of Proposition \ref{Pro 1}.
\end{proof}

\section{Fredholm determinant representation of $q_n$}\label{sec fredhlom}
In this section, we intend to prove Theorem \ref{thm1} by establishing a Fredholm determinant representation of the soliton gas solution $q_n$. We start with a construction of the $N$-soliton solution $q^{[N]}_{n}$ in terms of Fredholm determinants.

% Having established solvability of the soliton gas RH problem, we now construct explicit expressions for the  
% soliton and soliton gas solutions in terms of Fredholm determinants. 
% In the first part of this section, we introduce the associated tau-functions and represent the finite $N$
% soliton solutions as Fredholm determinants. 
% %, where we also refer to the Hirota bilinear method.
% In the second part we pass to the limit $N\to\infty$ in these determinant formula to obtain the soliton gas solution and thereby complete the proof of Theorem \ref{thm1}.

\subsection{The Fredholm determinant representation of the $N$-soliton solution}\label{s3.1}
% Now, using the tau-function $\tau^{[N]}$ defined in \eqref{tau}, we construct the $N$-soliton solution of the focusing AL system \eqref{e1.2}. 
%\begin{align}\label{eqtau}
%		q^{[N]}_{n}(t)=\i^{n+1}\E^{2\i t}\frac{\d}{\d t}\Im\ln\tau^{[N]}(\i;n+1,t).
%	\end{align}
%    Moreover, the product 
 %   \begin{align}
 %       c^{[N]}_n(t)=\prod_{k=n}^\infty(1+\abs{q^{[N]}_n(t)}^2)\label{def cNnt}
%    \end{align}
 %   can be written in the form
  %  \begin{align}\label{ectau}
%        c^{[N]}_n(t)=\left(\re\left(\frac{\tau^{[N]}(\i;n+2,t)}{\tau^{[N]}(\i;n,t)}\right)\right)^{-1}.
%    \end{align}
    
Note that the first row of the solution of RH problem \ref{rNsolition} is given by
	\begin{align}\label{e2.14}
		(M_{11}(\lb),M_{12}(\lb))
	&=(1,0)-\sum_{j=1}^{N}\left(\begin{matrix}
		\frac{\ii d_j\lb_j\hat\beta_j}{\lb-\ii\lb_j},&\frac{\ii^{n+1}e^{2\i t}d_j\hat\alpha_j}{\lb-\ii\lb_j^{-1}}
	\end{matrix}\right).
	\end{align}
Here, the vectors $\hat\alpha=\left(\hat\alpha_1,\cdots,\hat\alpha_N\right)^T$ and $\hat\beta=\left(\hat\beta_1,\cdots,\hat\beta_N\right)^T$ are determined by the finite linear system
	\begin{align}\label{Ns}
		\left(\begin{matrix}
			I_N&-\Psi\\\Psi&I_N
		\end{matrix}\right)\left(\begin{matrix}
			\hat\alpha\\\hat\beta
		\end{matrix}\right)=\left(\begin{matrix}
			\mathbf{d} \\\mathbf{0}
		\end{matrix}\right),
	\end{align}
where $I_N$ is the $N\times N$ identity matrix, 
%$d$ is a column vector function
\begin{align}\label{defpsi}
\mathbf{d}=(d_1,\cdots,d_N)^T,\qquad \Psi=\Psi(n,t):=\left(\frac{d_jd_k}{1-\lb_j^{-1}\lb_k^{-1}}\right)_{j,k=1}^N, \end{align}
with 
\begin{align}\label{def:dj}
d_j=d_j(n,t):=\sqrt{\Lambda_j\lb_j^{-n-2}\E^{t(\lb_j^{-1}-\lb_j)}}.
\end{align}
As it is assumed that $r(\lb)$ is positive on $[\i\eta_1,\i\eta_2]$, we see from \eqref{lbva} that $d_j>0$. 
%Furthermore, the whole solution of RH problem \ref{rNsolition} is determined by the symmetry
%    \begin{align}
%M(\bar\lb^{-1})=\sigma_2\overline{M(0)^{-1}M(\lb)}\sigma_2
%    \end{align}
%The second row of the solution is irrelevant for constructing the soliton and soliton-gas solutions and is therefore omitted.
Besides \eqref{eq:Nsoliton}, we also have
% In addition, the $N$-soliton solution of \eqref{e1.2} is recovered from 
% \begin{align}\label{eqN}
% q^{[N]}_n(t)=M_{12}(0;n+1,t).
% \end{align}
% Meanwhile, define the product 
\begin{align}
  c^{[N]}_n(t):=\prod_{k=n}^\infty(1+\abs{q^{[N]}_n(t)}^2)= M_{11}(0;n,t)^{-1},\label{def cNnt}
   \end{align}
which can be proved by using an argument analogous to the derivation of  \eqref{def:Z21}.

%, this product can be reconstructed by
% \begin{align}\label{concn}
% c^{[N]}_n(t).
% \end{align}
%Let $\{\lb_j,\Lambda_j\}_{j=1}^N$ be as in \eqref{lbva}. %Then,  we have $\lb_j>1$ and $\Lambda_j>0$.
Since $\Psi$ is real and symmetric, the linear system \eqref{Ns} is equivalent to 
	\begin{align}
		\left(\begin{matrix}
			\hat\alpha\\\hat\beta
		\end{matrix}\right)=\left(\begin{matrix}
			(I_N+\Psi^2)^{-1}&(I_N+\Psi^2)^{-1}\Psi\\-(I_N+\Psi^2)^{-1}\Psi&(I_N+\Psi^2)^{-1}
		\end{matrix}\right)\left(\begin{matrix}
			\mathbf{d}\\\mathbf{0}
		\end{matrix}\right).
	\end{align}
Note that $I_N+\Psi^2=(I_N+\i\Psi)(I_N-\i\Psi)$, we introduce the tau-function 
	\begin{align}\label{tau}
		\tau^{[N]}(\eta)=\tau^{[N]}(\eta;n,t):=\det(I_N-\eta\Psi(n,t)),\qquad N\ge1.
	\end{align}
	Using the fact that
	\begin{align}\label{erecoN}
		\Psi_{jk}-(S^2\Psi)_{jk}=\frac{d_jd_k-(S^2d)_j(S^2d)_k}{1-\lb_j^{-1}\lb_k^{-1}}=d_jd_k=(\mathbf{d}\mathbf{d}^T)_{jk},
	\end{align}
where $S$ is the right-shift operator defined in \eqref{defs}, we obtain from the Weinstein–Aronszajn identity that
	\begin{align*}
		S^2\tau^{[N]}(\eta)&=\det(I_N-\eta S^2\Psi)
		=\tau^{[N]}(\eta)\det(I_N+\eta(I_N-\eta\Psi)^{-1}(\Psi-S^2\Psi))\\
		&=\tau^{[N]}(\eta)\det(I_N+\eta(I_N-\eta\Psi)^{-1}\mathbf{d}\mathbf{d}^T)=\tau^{[N]}(\eta)(1+\eta \mathbf{d}^T(I_N-\eta\Psi)^{-1}\mathbf{d}),
	\end{align*}
	which yields
	\begin{align}\label{e2.21}
		\frac{S^2\tau^{[N]}(\eta)}{\tau^{[N]}(\eta)}=1+\eta \mathbf{d}^T(I_N-\eta\Psi)^{-1}\mathbf{d}.
	\end{align}
	In addition, by Jacobi's formula, we have
	\begin{align}\label{e3.10}
		\frac{\ddd }{\ddd t}\ln\tau^{[N]}(\eta)=\Tr[-\eta(I_N-\eta\Psi)^{-1}\frac{\d }{\d t}\Psi].
	\end{align}
In view of \eqref{defpsi}, it follows that 
    \begin{align}\label{e3.11}
        \frac{\d }{\d t}\Psi_{jk}=-\frac{\lb_j+\lb_k}{2}d_jd_k.
    \end{align}
    %and both $\Psi$ and $D$ are symmetric,
Substituting \eqref{e3.11} into \eqref{e3.10} gives us
    \begin{align}\label{qN}
        \frac{\ddd }{\ddd t}\ln\tau^{[N]}(\eta)=\frac{\eta}{2}\mathbf{d}^T[D(I_N-\eta\Psi)^{-1}+(I_N-\eta\Psi)^{-1}D]\mathbf{d}=\eta \mathbf{d}^TD(I_N-\eta\Psi)^{-1}\mathbf{d},
    \end{align}
where  $D:=\operatorname{diag}(\lb_1,\cdots,\lb_N)$. Taking $\eta=\i$ in \eqref{qN} and extracting the imaginary part yields
	\begin{equation}\label{e2.23}
    \begin{split}
        \frac{\d}{\d t}\Im\ln\tau^{[N]}(\i)&=\frac{1}{2\i}\left(\frac{\ddd }{\ddd t}\ln\tau^{[N]}(\i)-\overline{\frac{\ddd }{\ddd t}\ln\tau^{[N]}(\i)}\right)=\mathbf{d}^TD(I_N+\Psi^2)^{-1}\mathbf{d}.
    \end{split}
	\end{equation}
	On the other hand, it follows from \eqref{e2.14} that
	\begin{align}
		M_{12}(0;n,t)=\i^n\e^{2\i t}\sum_{j=1}^{N}\lb_jd_j(n,t)\hat\alpha_j(n,t)=\i^n\E^{2\i t}[\mathbf{d}^TD(I_N+\Psi^2)^{-1}\mathbf{d}](n,t),
	\end{align}
	which combined with \eqref{eq:Nsoliton} and \eqref{e2.23} induces 
    \begin{align}\label{eqtau}
		q^{[N]}_{n}(t)=\i^{n+1}\E^{2\i t}\frac{\d}{\d t}\Im\ln\tau^{[N]}(\i;n+1,t).
	\end{align}
Similarly,  using \eqref{e2.14}, \eqref{def cNnt},  \eqref{erecoN} and \eqref{e2.21}, it follows that
	\begin{align}\label{ectau}
		\left(c^{[N]}_n(t)\right)^{-1}&=1+\mathbf{d}(n,t)^T\hat\beta(n,t)=1-\mathbf{d}(n,t)^T(I+\Psi(n,t)^2)^{-1}\Psi(n,t) \mathbf{d}(n,t) \nonumber \\
		&=\frac{1}{2}\left(\frac{\tau^{[N]}(-\i;n+2,t)}{\tau^{[N]}(-\i;n,t)}+\frac{\tau^{[N]}(\i;n+2,t)}{\tau^{[N]}(\i;n,t)}\right)=\re\left(\frac{\tau^{[N]}(\i;n+2,t)}{\tau^{[N]}(\i;n,t)}\right).
	\end{align}
%which is equivalent to
%\begin{align}\label{ectau}
%        c^{[N]}_n(t)=\left(\re\left(\frac{\tau^{[N]}(\i;n+2,t)}{\tau^{[N]}(\i;n,t)}\right)\right)^{-1}.
%   \end{align}

We are now ready to prove Theorem \ref{thm1}.

    \subsection{Proof of Theorem \ref{thm1}}
	%We are now ready to prove Theorem \ref{thm1}.
By \eqref{def:dj}, it is readily seen that $$d_j(n+2,t)=\lb_j^{-1} d_j(n,t).$$ This inspires us to consider the tau-function in \eqref{tau} for different parity of $n$. 
For the tau-function associated with $\Psi(2n,t)$,  define the linear operators
    $$A^{N,e}_{n,t}: \ell^2(\{l\}_{l=n}^{\infty})\to \mathbb{C}^N, \quad B^{N,e}_{n,t}: \mathbb{C}^N\to \ell^2(\{l\}_{l=n}^{\infty})$$
    by
	\begin{align}\label{def:evenop}
	(A^{N,e}_{n,t}[x])_j=\sum_{l=n}^{\infty}d_j(2l,t)x(l),\quad B^{N,e}_{n,t}[u](l)=-\i\sum_{j=1}^{N}d_j(2l,t)u_j,
	\end{align}
    for $j=1,\cdots,N$ and $l \in \mathbb{Z}\geq n$. Thus, by \eqref{defpsi}, one has
	\begin{align*}
		-\i\Psi(2n,t)=A^{N,e}_{n,t}\circ B^{N,e}_{n,t}.
	\end{align*}
Using the fact that $\det(I_N+ A^{N,e}_{n,t}\circ B^{N,e}_{n,t})=\det(\identity_{\ell^2(\{l\}_{l=n}^{\infty})}+ B^{N,e}_{n,t}\circ A^{N,e}_{n,t})$, we obtain from \eqref{eqtau}, \eqref{tau} and the above formula that 
	\begin{align*}
		q^{[N]}_{2n-1}(t)&=(-1)^{n}\E^{2\i t}\frac{\d}{\d t}\Im\ln\tau^{[N]}(\i;2n,t)\\
        &=(-1)^{n}\E^{2\i t}\frac{\d}{\d t}\Im\ln\det(I_N+A^{N,e}_{n,t}\circ B^{N,e}_{n,t})\\
		&=(-1)^{n}\E^{2\i t}\frac{\d}{\d t}\Im\ln\det(\identity_{\ell^2(\{l\}_{l=n}^{\infty})}+ \mathcal{K}^{N,e}_{n,t}),
	\end{align*}
where $\mathcal{K}^{N,e}_{n,t}:=B^{N,e}_{n,t}\circ A^{N,e}_{n,t}$. Note that $\mathcal{K}^{N,e}_{n,t}$ is a discrete integral operator acting on $\ell^2(\{l\}_{l=n}^{\infty})$ with kernel $\sum_{k=1}^{N}d_k(l+l',t)^2$, that is, 
	\begin{align}
		\mathcal{K}^{N,e}_{n,t}[x](l):=-\i\sum_{l'=n}^{\infty}\sum_{k=1}^{N}d_k(l+l',t)^2x(l'). 
	\end{align}
	%of the kernel $\sum_{k=1}^{N}d_k(l+l',t)^2$.
Let $N\to\infty$, by equations \eqref{lbva} and \eqref{e2.2}, it follows that
	\begin{align}\label{e2.28}
		\lim_{N\to\infty}\sum_{k=1}^{N}d_k(l+l',t)^2=\frac{1}{2\pi\i}\int_{\Sigma_1}r(s)(-\i s)^{-l-l'-2}e^{\i t(s+s^{-1})}\ddd s,
	\end{align}
which implies that sequence of operators $\mathcal{K}^{N,e}_{n,t}$ converges to the operator $\mathcal{K}^e_{n,t}$ on $ \ell^2(\{l\}_{l=n}^\infty)$ whose kernel is given by the right-hand side of \eqref{e2.28}. Consequently, 
	\begin{align}
		\lim_{N\to\infty}\det(\identity_{\ell^2(\{l\}_{l=n}^{\infty})}+ \mathcal{K}^{N,e}_{n,t})=\det(\identity_{\ell^2(\{l\}_{l=n}^{\infty})}+ \mathcal{K}^e_{n,t}),
	\end{align}
	where $\mathcal{K}^e_{n,t}=B^e_{n,t}\circ A^e_{n,t}$ with
	\begin{align}
		&A^e_{n,t}[f](s)=\sqrt{r(s)}\sum_{l=n}^{\infty}(-\i s)^{-l-1}\E^{\frac{\i t}{2}(s+s^{-1})}f(l),\\
		&B^e_{n,t}[g](l)=-\frac{1}{2\pi}\int_{\Sigma_1}\sqrt{r(s)}(-\i s)^{-l-1}\E^{\frac{\i t}{2}(s+s^{-1})}g(s)\ddd s.
	\end{align}
Again using $\det(\identity_{\ell^2(\{l\}_{l=n}^{\infty})}+B^e_{n,t}\circ A^e_{n,t})=\det(\identity_{L^2(\Sigma_1)}+A^e_{n,t}\circ B^e_{n,t})$, we arrive at \eqref{e1.4} for even $n$.
%the Fredholm determinant representation for the soliton-gas solution stated in Theorem \ref{thm1} where $\mathcal{K}_{2n,t}:=A^e_{n,t}\circ B^e_{n,t}$.

The same argument applies for the tau-function associated with $\Psi(2n-1,t)$. Replacing $A^{N,e}_{n,t}$ and $B^{N,e}_{n,t}$ in \eqref{def:evenop} by $A^{N,o}_{n,t}$ and $B^{N,o}_{n,t}$ with 
\begin{align*}
    (A^{N,o}_{n,t}[x])_j=\sum_{l=n}^{\infty}d_j(2l-1,t)x(l),\quad B^{N,o}_{n,t}[u](l)=-\i\sum_{j=1}^{N}d_j(2l-1,t)u_j.
\end{align*}
Taking the limit $N \to \infty$ yields \eqref{e1.4} for odd $n$, 
where
$\mathcal{K}_{2n-1,t}:=A^o_{n,t}\circ B^o_{n,t}$ with
\begin{align*}
    &A^o_{n,t}[f](s)=\sqrt{r(s)}\sum_{l=n}^{\infty}(-\i s)^{-l-\frac{1}{2}}\E^{\frac{\i t}{2}(s+s^{-1})}f(l),\\
		&B^o_{n,t}[g](l)=-\frac{1}{2\pi}\int_{\Sigma_1}\sqrt{r(s)}(-\i s)^{-l-\frac{1}{2}}\E^{\frac{\i t}{2}(s+s^{-1})}g(s)\ddd s.
\end{align*}

The formula \eqref{fcn} for $c_n(t)=\prod_{k=n}^\infty(1+\abs{q_n}^2)$  follows directly  by expressing $c_n^{[N]}(t)$ in \eqref{def cNnt} as a  Fredholm determinant and setting $N \to\infty$.

This completes the proof of Theorem \ref{thm1}. \qed

\section{Large-$n$ asymptotic analysis of the RH problem for $Z(\lb; n,0)$}\label{sec 3}
From the reconstruction formula \eqref{rec}, we have 
\begin{align}
		q_n(0)= Z_{12}(0;n+1,0),
	\end{align}
where $Z$ solves RH problem \ref{RHP0}. This inspires us to analyze the large-$n$ behavior of $Z(\lb;n+1,0)$.

If $t=0$, one has 
\begin{align*}
	\phi(\lb;n,0)=n\ln\lb,
\end{align*}
which we denote by $\phi(\lb)$ in this section for brevity.  
% In this section, we consider the asymptotic behavior of $q_n(0)$ as $n\to\pm\infty$.
% The reconstruction formula \eqref{rec} inspires us to analyze the large-$n$ behaviors of $Z(\lb;n+1,0)$.
% Recall the phase function for $t=0$,
% \begin{align*}
% 	\phi(\lb;n,0)=n\ln\lb.
% \end{align*}
% For brevity, we shall write simply $\phi(\lb)$ in place of $\phi(\lb;n,0)$ in this section.
% Moreover, the properties of $\phi(\lb)$ imply that in RH problem \ref{RHP0},
Thus, the jump matrix $J$ of $Z$ decays exponentially to the identity matrix as $n\to+\infty$ in this case, so does $Z(\lb)$ for large positive $n$.
%It then follows \hl{that} as $n\to+\infty$,
%\begin{align}
%    q_n(0)=\mathcal{O}(\E^{-cn}),
%\end{align}
%for some positive constant $c$.
%go to the proof
In what follows, we analyze the RH problem \ref{RHP0} for $Z(\lb; n,0)$  as $n\to-\infty$ and begin with the introduction of some auxiliary functions, which particularly include the so-called $g$-function \cite{gfunction, G2001,GT2002}  to control the exponentially growing off-diagonal factors in the jump matrix.  
% Firstly, we introduce the so-called $g$-function to 
% control the exponentially growing off-diagonal factors in the jump matrix .

\subsection{The auxiliary functions}\label{aux1}
To define the $g$-function, we first introduce
\begin{align}\label{eRn}
	R(\lb):=\sqrt{(\lb-\ii\eta_1)(\lb-\ii\eta_2)(\lb-\ii\eta_1^{-1})(\lb-\ii\eta_2^{-1})},
\end{align}
which determines a genus-1 Riemann surface $\mathcal{R}$ with homology basis $\{\mathfrak{a},\mathfrak{b}\}$ as shown in Figure \ref{f4}.
Note that the branch cuts of $R(\lb)$ are $\Sigma_1\cup\Sigma_2$, we choose the principal sheet such that, as $\lb\to\infty$, $R(\lb)=-\lb^2+\oo(\lb)$, and $R(\lb)$ satisfies the symmetry relation
\begin{align}\label{esR}
	R(\lb)=-\lb^2\overline{R(\bar\lambda^{-1})}.
\end{align}
Thus, we have $R(0)=1$ and $R(\lb)>0$ for $\lb\in\ii(-\infty,\eta_2^{-1})\cup\ii(\eta_2,+\infty)$.
In addition, $R_\pm(\lb)\in\ii\mathbb{R}^\pm$ for $\lb\in\Sigma_1$ while $R_\pm(\lb)\in\ii\mathbb{R}^\mp$ for $\lb\in\Sigma_2$.

\begin{figure}
	\centering
	\tikzset{every picture/.style={line width=0.75pt}}
	\begin{tikzpicture}
		\draw [thin] (0,0) -- (0,6);
		\draw [thin] (4,2) -- (4,8);
		\draw [thin] (4,2) -- (0,0);
		\draw [thin] (0,6) -- (4,8);
		\draw (2,2) -- (2,3);
		\draw (2,4) -- (2,6);
		\filldraw (2,6) circle (1pt) node [left] {\small $\ii\eta_2$};
		\filldraw (2,4) circle (1pt) node [left] {\small $\ii\eta_1$};
		\filldraw (2,3) circle (1pt) node [left] {\small $\ii\eta_1^{-1}$};
		\filldraw (2,2) circle (1pt) node [left] {\small $\ii\eta_2^{-1}$};
		\filldraw (7,6) circle (1pt);
		\filldraw (7,4) circle (1pt);
		\filldraw (7,3) circle (1pt);
		\filldraw (7,2) circle (1pt);
		\draw [blue] (2,5) ellipse (1 cm and 1.5 cm);
		\draw [blue,->] (1,5) -- (1,5.01);
		\draw [orange] (2,4.2) arc (150:210:1.4cm);
		\draw [orange, ->] (1.815,3.5) -- (1.815,3.51);
		\draw [thin] (5,0) -- (5,6);
		\draw [thin] (9,2) -- (9,8);
		\draw [thin] (9,2) -- (5,0);
		\draw [thin] (5,6) -- (9,8);
		\draw (7,2) -- (7,3);
		\draw (7,4) -- (7,6);
		\draw [orange] (7,2.8) arc (-30:30:1.4cm);
		\draw [dashed,orange] (2,2.8) -- (7,2.8);
		\draw [dashed,orange] (2,4.2) -- (7,4.2);
		\draw [dashed] (2,2) -- (7,2);
		\draw [dashed] (2,4) -- (7,4);
		\draw [dashed] (2,3) -- (7,3);
		\draw [dashed] (2,6) -- (7,6);
		\node at (3,0.5) {\text{Sheet I}};
		\node at (8,0.5) {\text{Sheet II}};
		\node at (2.6,6.5) [blue] {\text{$\mathfrak{b}$}};
		\node at (0.5,3.25) [orange] {\text{$\mathfrak{a}$}};
		\draw [lightgray,dashed,->] (0.7,3.25)--(1.7,3.25);
	\end{tikzpicture}
	\caption{\footnotesize  The genus-1 Riemann surface $\mathcal{R}$  with  homology basis $\{\mathfrak{a}, \mathfrak{b}\}$ associated with $R(\lambda)$ in \eqref{eRn}, where Sheet I is the principal sheet, and the curves $\Sigma_1$ and $\Sigma_2$ are the branch cuts.}\label{f4}
\end{figure}
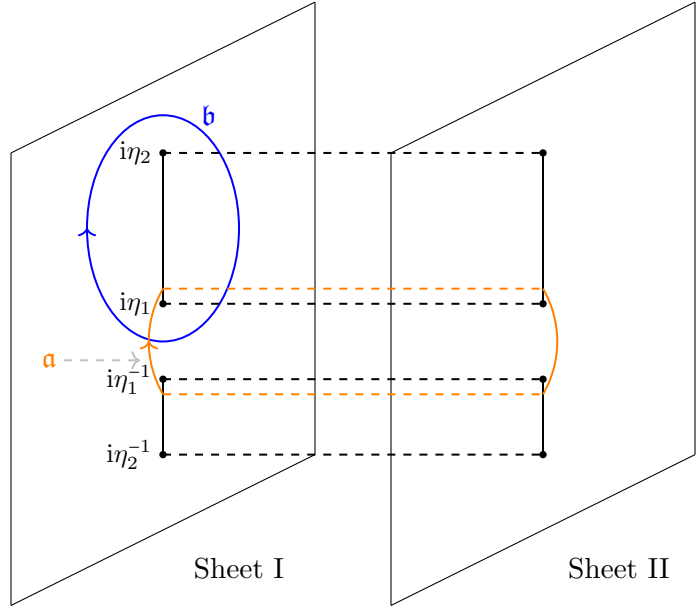

Next we define 
\begin{align}\label{evarphin}
	g(\lb):=-\frac{1}{2}\int_{\ii\eta_2}^{\lb}\frac{s-s^{-1}-\ii\kappa_1}{R(s)}\ddd  s,
\end{align}
 where 
\begin{align}
\kappa_1=2\left(\frac{2\Pi(l_1^2,k)}{K(k)}-1\right),\quad l_j=\frac{\eta_j-1}{\eta_j+1},\quad k=\frac{l_1}{l_2},\quad j=1,2,\label{def k}
\end{align}
with $K$ and $\Pi$ being the elliptic integrals defined in \eqref{ei1} and  \eqref{ei2}, respectively.
% and the integral path for $g(\lb)$ does not cross the interval $(\i\eta_1^{-1},\i\eta_1)$.
% Denote the interval
% \begin{align}
%    \Sigma= (\ii\eta_2^{-1},\ii\eta_2),\label{def Sigma}
% \end{align}
% oriented upwards.
% This phase function satisfies the following properties.
\begin{prop}\label{png}
The $g$-function defined in \eqref{evarphin} satisfies the following properties.    
    \begin{itemize}
        \item[\rm{(a)}] Let $\ddd g$ be the differential of $g(\lb)$, it then follows that
\begin{align*}
    \int_{\ii\eta_1^{-1}}^{\ii\eta_1}\ddd g=0.
\end{align*}
\item[\rm{(b)}]  $g(\lb)-\frac{\ln \lb}{2}$ is analytic in  $\C \setminus\Sigma$, where
\begin{align}
   \Sigma : = \ii(\eta_2^{-1}, \eta_2)
   \label{def Sigma}
\end{align}
with upward orientation and the branch cut of the logarithmic function is taken along the negative imaginary axis. Moreover, one has
\begin{align}\label{eq:gjump}
	g_+(\lb)=
	\begin{cases}
		-g_-(\lb),&\lb\in\Sigma_1\cup\Sigma_2,\\
		g_-(\lb)-\ii\Omega,&\lb\in\ii(\eta_1^{-1},\eta_1),
	\end{cases}
\end{align}
%\begin{align*}
%	\begin{aligned}
%		& g_+(\lb) + g_-(\lb) = 0, & \lb\in\Sigma_1\cup\Sigma_2,  \\
%	& g_+ (\lb) - g_-(\lb) = -\i\Omega, & \lb \in \ii(\eta_1^{-1},\eta_1),
%	\end{aligned} 
%\end{align*}
where
\begin{align}\label{eOmega}
	\Omega=\ii\oint_{\mathfrak{b}}\ddd g\in\mathbb{R}.
\end{align}

\item[\rm{(c)}] $g(\lb)$ satisfies the symmetry relation
\begin{align}\label{e3.6}
	g(\bar\lb^{-1})=-\overline{g(\lb)}.
\end{align}
\item[\rm{(d)}] 
$g(\lb)$ admits the asymptotic properties:
\begin{align}
    g(\lb)&=\frac{\ln\lb}{2}+\oo(1),\qquad \lambda\to \infty,\label{e4.10}\\
	g(\lambda)&=\frac{(\eta_2+\eta_2^{-1}-\kappa_1)(-\ii\lb-\eta_2)^{1/2}}{\sqrt{(\eta_2-\eta_1)(\eta_2-\eta_1^{-1})(\eta_2-\eta_2^{-1})}}(1+\oo(\lambda-\ii\eta_2)), \qquad \lambda\to\ii\eta_2, \label{asy g eta2}
\end{align}
and as $\lambda\to\ii\eta_1$,
\begin{align}
    g(\lambda)&=\mp\frac{\ii\Omega}{2}+\frac{(\eta_1+\eta_1^{-1}-\kappa_1)(\eta_1+\ii\lb)^{1/2}}{\sqrt{(\eta_2-\eta_1)(\eta_1-\eta_1^{-1})(\eta_1-\eta_2^{-1})}}(1+\oo(\lambda-\ii\eta_1)), \qquad \mp\re \lb>0.
    \label{asy g eta21}
\end{align}

       \end{itemize}
\end{prop}
\begin{proof}
By direct computations, we have
\begin{align}\label{int 1/R(s)}
 \int_{\ii\eta_1^{-1}}^{\ii\eta_1} \frac{1}{R(s)}\ddd  s &=-\ii \int_{\eta_1^{-1}}^{\eta_1} \frac{1}{\sqrt{(s-\eta_1)(s-\eta_1^{-1})(s-\eta_2)(s-\eta_2^{-1})}}\ddd  s\nonumber\\
    &=-\frac{\ii}{2} \int_{-l_1}^{l_1}\sqrt{\frac{(1-l_1^2)(1-l_2^2)}{(l_1^2-x^2)(l_2^2-x^2)}}\ddd x=\frac{-\ii\sqrt{(1-l_1^2)(1-l_2^2)}}{l_2}K(k),
\end{align}
where $K$ is defined in \eqref{ei1},
    and 
    \begin{align}
     \int_{\ii\eta_1^{-1}}^{\ii\eta_1} \frac{s-s^{-1}}{R(s)}\ddd  s&= \int_{\eta_1^{-1}}^{\eta_1} \frac{s+s^{-1}}{\sqrt{(s-\eta_1)(s-\eta_1^{-1})(s-\eta_2)(s-\eta_2^{-1})}}\ddd  s\nonumber\\
    &= \int_{-l_1}^{l_1}\sqrt{\frac{(1-l_1^2)(1-l_2^2)}{(l_1^2-x^2)(l_2^2-x^2)}}\cdot\frac{(1+x^2)\ddd x}{1-x^2} \nonumber \\
    & =\frac{2\sqrt{(1-l_1^2)(1-l_2^2)}}{l_2}(2\Pi(l_1^2,k)-K(k)).\label{int s-1/s/R(s)}
\end{align}
Thus, the choice of $\kappa_1$ in \eqref{def k} ensures item  (a). Item (b) follows immediately the definition of $g(\lb)$,  \eqref{int 1/R(s)} and \eqref{int s-1/s/R(s)}. The symmetry relation of $g$ in item (c) follows directly from that of $R$ given in \eqref{esR}.
For item (d), it is readily seen from the definition of $g$ in \eqref{evarphin} that $g(\lb)-\frac{\ln\lb}{2}$ is bounded as $\lb\to\infty$.
%As a consequence, by the analyticity of $g(\lb)-\frac{\ln\lb}{2}$, it is easy to verify that $g(\lb)$ is analytic in the neighborhood of $\Sigma$. 
The remaining statements, namely \eqref{asy g eta2} and \eqref{asy g eta21}, are verified as follows. From the jump condition for $g(\lb)$ given in \eqref{eq:gjump}, it follows that $g(\lb)(\lb-\i\eta_2)^{-1/2}$ is holomorphic at $\lb=\ii\eta_2$, where the branch cut of $(\lb-\i\eta_2)^{1/2}$ is taken along $(\ii\eta_2, -\ii \infty)$. A direct computation using L’Hôpital’s rule yields
\begin{align*}
    \lim_{\lb\to\ii\eta_2}g(\lb)(-\ii\lb -\eta_2)^{-1/2}=\frac{\eta_2+\eta_2^{-1}-\kappa_1}{\sqrt{(\eta_2-\eta_1)(\eta_2-\eta_1^{-1})(\eta_2-\eta_2^{-1})}},
\end{align*}
which is equivalent to \eqref{asy g eta2}.
%\begin{align*}
%	g(\lambda)&=\frac{(\eta_2+\eta_2^{-1}-\kappa_1)(-\ii\lb-\eta_2)^{1/2}}{\sqrt{(\eta_2-\eta_1)(\eta_2-\eta_1^{-1})(\eta_2-\eta_2^{-1})}}(1+\oo(\lambda-\ii\eta_2)).
%\end{align*}
A similar analysis leads to \eqref{asy g eta21}. 

This completes the proof of Proposition \ref{png}.
%Besides, this also prove \eqref{e4.10} in $(d)$.
%We have the symmetry for $g(\lb)$
%These complete the proof.
%for $\lambda\to\ii\eta_1$ and $\pm\re\lambda>0$, we have the similar asymptotic properties:
%\begin{align*}
%	g(\lambda)&=\mp\frac{\ii\Omega}{2}+\frac{(\eta_1+\eta_1^{-1}-\kappa_1)(\eta_1+\ii\lb)^{1/2}}{\sqrt{(\eta_2-\eta_1)(\eta_1-\eta_1^{-1})(\eta_1-\eta_2^{-1})}}(1+\oo(\lambda-\ii\eta_1)).
%\end{align*}
\end{proof}
The signature table of $\Re [g]$ is illustrated in Figure \ref{fig reg t=0} and we also set
\begin{align}\label{def:ginfty}
	g^{(\infty)}:=\lim_{\lb\to\infty}\left(g(\lb)-\frac{\ln\lb}{2}\right).
\end{align} 
for later use. 

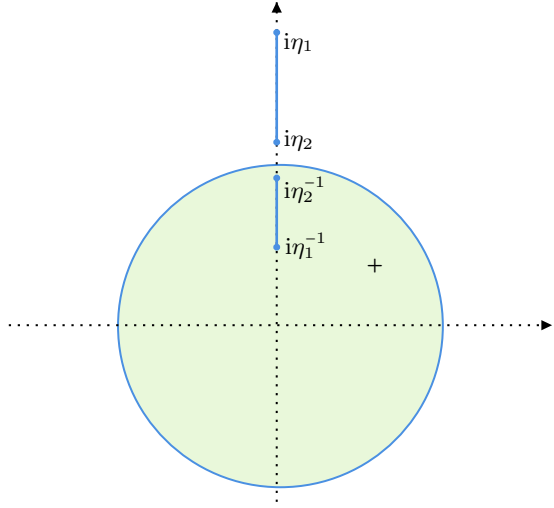
\begin{figure}

\centering
\tikzset{every picture/.style={line width=0.75pt}} %set default line width to 0.75pt

\begin{tikzpicture}[x=0.75pt,y=0.75pt,yscale=-0.6,xscale=0.6]
%Shape: Circle [id:dp7230156522170901] 
\draw  [color={rgb, 255:red, 74; green, 144; blue, 226 }  ,draw opacity=1 ][fill={rgb, 255:red, 184; green, 233; blue, 134 }  ,fill opacity=0.3 ] (185.99,318.72) .. controls (185.99,244.1) and (246.48,183.61) .. (321.1,183.61) .. controls (395.71,183.61) and (456.2,244.1) .. (456.2,318.72) .. controls (456.2,393.33) and (395.71,453.82) .. (321.1,453.82) .. controls (246.48,453.82) and (185.99,393.33) .. (185.99,318.72) -- cycle ;
%Straight Lines [id:da9493975688480626] 
\draw  [dash pattern={on 0.84pt off 2.51pt}]  (318,466.08) -- (318,45) ;
\draw [shift={(318,47)}, rotate = 90] [fill={rgb, 255:red, 0; green, 0; blue, 0 }  ][line width=0.08]  [draw opacity=0] (8.93,-4.29) -- (0,0) -- (8.93,4.29) -- cycle    ;
%Straight Lines [id:da8901688652077531] 
\draw [color={rgb, 255:red, 74; green, 144; blue, 226 }  ,draw opacity=1 ][line width=1]    (318,72.58) -- (318,164.58) ;
\draw [shift={(318,164.58)}, rotate = 89.69] [color={rgb, 255:red, 74; green, 144; blue, 226 }  ,draw opacity=1 ][fill={rgb, 255:red, 74; green, 144; blue, 226 }  ,fill opacity=1 ][line width=1]      (0, 0) circle [x radius= 1.74, y radius= 1.74]   ;
\draw [shift={(318,72.58)}, rotate = 89.69] [color={rgb, 255:red, 74; green, 144; blue, 226 }  ,draw opacity=1 ][fill={rgb, 255:red, 74; green, 144; blue, 226 }  ,fill opacity=1 ][line width=1]      (0, 0) circle [x radius= 1.74, y radius= 1.74]   ;
%Straight Lines [id:da8405759072744595] 
\draw [color={rgb, 255:red, 0; green, 0; blue, 0 }  ,draw opacity=1 ] [dash pattern={on 0.84pt off 2.51pt}]  (95.1,318) -- (544.1,318) ;
\draw [shift={(547.1,318)}, rotate = 180] [fill={rgb, 255:red, 0; green, 0; blue, 0 }  ,fill opacity=1 ][line width=0.08]  [draw opacity=0] (8.93,-4.29) -- (0,0) -- (8.93,4.29) -- cycle    ;
%Straight Lines [id:da7246130170784083] 
\draw [color={rgb, 255:red, 74; green, 144; blue, 226 }  ,draw opacity=1 ][line width=1]    (318,194.58) -- (318,252.58) ;
\draw [shift={(318,252.58)}, rotate = 90] [color={rgb, 255:red, 74; green, 144; blue, 226 }  ,draw opacity=1 ][fill={rgb, 255:red, 74; green, 144; blue, 226 }  ,fill opacity=1 ][line width=1]      (0, 0) circle [x radius= 1.74, y radius= 1.74]   ;
\draw [shift={(318,194.58)}, rotate = 90] [color={rgb, 255:red, 74; green, 144; blue, 226 }  ,draw opacity=1 ][fill={rgb, 255:red, 74; green, 144; blue, 226 }  ,fill opacity=1 ][line width=1]      (0, 0) circle [x radius= 1.74, y radius= 1.74]   ;

% Text Node
\draw (322.42,238.19) node [anchor=north west][inner sep=0.75pt]  [font=\footnotesize]  {$\mathrm{i} \eta _{1}^{-1}$};
% Text Node
\draw (321.3,152.14) node [anchor=north west][inner sep=0.75pt]  [font=\footnotesize]  {$\mathrm{i} \eta _{2}$};
% Text Node
\draw (321.52,191.13) node [anchor=north west][inner sep=0.75pt]  [font=\footnotesize]  {$\mathrm{i} \eta _{2}^{-1}$};
% Text Node
\draw (321.2,70.38) node [anchor=north west][inner sep=0.75pt]  [font=\footnotesize]  {$\mathrm{i} \eta _{1}$};
% Text Node
\draw (390.03,259.97) node [anchor=north west][inner sep=0.75pt]    {$+$};
\end{tikzpicture}
%	\centering	\tikzset{every picture/.style={line width=0.75pt}}
%	\resizebox{0.35\textwidth}{!}{
%	\begin{tikzpicture}
%		\draw [thin] (0,0) circle (2cm);
%		\filldraw [thin,lightgray] (0,0) circle (4cm);
%		\draw [->] (0,2) -- (0,2.7);
%		\draw [->] (0,16/3) -- (0,6.4);
%		\draw (0,2) -- (0,3);
%		\draw  (0,16/3) -- (0,8);
%		\filldraw (0,0) circle (1pt) node [left] {0};
%		\filldraw (0,8) circle (0.5pt) node [right] {$\ii\eta_2$};
%		\filldraw (0,16/3) circle (0.5pt) node [right] {$\ii\eta_1$};
%		\filldraw (0,3) circle (0.5pt) node [right] {$\ii\eta_1^{-1}$};
%		\filldraw (0,2) circle (0.5pt) node [right] {$\ii\eta_2^{-1}$};
%		\node at (-0.3,6) {$\Sigma_1$};
%		\node at (-0.3,2.5) {$\Sigma_2$};
%	\end{tikzpicture}
%}
	\caption{\footnotesize  The signature table of $\Re [g]$. $\Re [g]>0$ in the green region,  while $\Re[g]<0$ in the white region, and $\Re[g]=0$ on the blue curve.}
	\label{fig reg t=0}
\end{figure}

% For $t=0$ and $n\to-\infty$,  the jump matrix $J(\lb)$ in RH problem \ref{RHP0} does not decay exponentially to $I$, and thus we need to open the lenses.
% To open these lenses, we have introduced the phase function $g(\lb)$ with the signature table shown in Figure \ref{fig reg t=0}, and 
Besides the  $g$-function defined in \eqref{evarphin}, we further introduce an auxiliary function
\begin{align}\label{symdeltan}
	\delta(\lb):=&\exp\left(\frac{R(\lb)}{2\pi\ii}\left(\int_{\ii\eta_1}^{\ii\eta_2}\frac{\ln r(s)}{R_+(s)}\frac{\ddd  s}{s-\lb}-\int_{\ii\eta_2^{-1}}^{\ii\eta_1^{-1}}\frac{\ln r(\bar s^{-1})}{R_+(s)}\frac{\ddd  s}{s-\lb}-\int_{\ii\eta_1^{-1}}^{\ii\eta_1}\frac{\ii\Delta}{R(s)}\frac{\ddd  s}{s-\lb}\right)\right),
\end{align}
where 
\begin{align}\label{delta1}
	\Delta:&=2\ii\int_{\ii\eta_1}^{\ii\eta_2}\frac{\ln r(s)}{R_+(s)}\ddd  s\Bigg/\int_{\ii\eta_1^{-1}}^{\ii\eta_1}\frac{\ddd  s}{R(s)}=\frac{-2l_2\int_{\eta_1}^{\eta_2}\frac{\ln r(\ii s)}{|R_+(s)|}\ddd  s}{\sqrt{(1-l_1^2)(1-l_2^2)}K(k)}
\end{align}
is a real constant to guarantee the boundedness of $\delta(\lb)$ as $\lb\to\infty$. It is readily seen from its definition that 
$\delta(\lb)$ is analytic for $\lb\in \ccc\setminus\Sigma$ and admits the jump relations
\begin{align}\label{delta0}
	\delta_+(\lb) =\begin{cases}
		\delta_-(\lb)^{-1}r(\lb),&\lb\in\Sigma_1,\\
		\delta_-(\lb)^{-1}r(\bar\lb^{-1})^{-1},&\lb\in\Sigma_2,\\
		\delta_-(\lb)\E^{-\ii\Delta},&\lb\in\ii(\eta_1^{-1},\eta_1).
	\end{cases}
\end{align}
% where $\Delta$ is real constant to guarantee the boundedness of $\delta(\lb)$ as $\lb\to\infty$,
% \begin{align*}
% 	\Delta&=2\ii\int_{\ii\eta_1}^{\ii\eta_2}\frac{\ln r(s)}{R_+(s)}\ddd  s\Bigg/\int_{\ii\eta_1^{-1}}^{\ii\eta_1}\frac{\ddd  s}{R(s)}=\frac{-2l_2\int_{\eta_1}^{\eta_2}\frac{\ln r(\ii s)}{|R_+(s)|}\ddd  s}{\sqrt{(1-l_1^2)(1-l_2^2)}K(k)}.
% \end{align*}
Substituting $\lb\to\bar\lb^{-1}$ into \eqref{symdeltan}, one immediately verifies the symmetry relation 
\begin{align}\label{edeltasym}
	\delta(\bar\lb^{-1})=\overline{\delta(\lb)^{-1}}.
\end{align}
Moreover, one has
\begin{align}
&\delta(\infty):=\lim_{\lb\to\infty}	\delta(\lb)\nonumber\\
&=\exp\left\{\frac{1}{2\pi\ii}\left[\int_{\ii\eta_{1}}^{\ii \eta_{2}}  \frac{s\ln r(s) }{R_+(s)} \d s-
\int_{\ii\eta_{2}^{-1}}^{\ii\eta_{1}^{-1}}  \frac{s\ln r(\bar{s}^{-1}) }{R_+(s)} \d s-\int_{\ii\eta_{1}^{-1}} ^{\ii\eta_{1}} \frac{\ii \Delta s}{R(s)} \d s\right]
\right\}\in\mathbb{R}. \label{deltainfty0}
\end{align}

\subsection{Lenses opening}
With the functions $g$ and $\delta$ defined in \eqref{evarphin} and  \eqref{symdeltan}, we now define
\begin{align}\label{T1}
	T(\lb)=T(\lb; n)=\delta(\infty)^{\sigma_3}\E^{n g^{(\infty)}\sigma_3}Z(\lb)\E^{-n(g(\lb)-\frac{\ln\lb}{2})\sigma_3}\delta(\lb)^{-\sigma_3},
\end{align}
where $g^{(\infty)}$ is given in \eqref{def:ginfty}. It is then readily seen from RH problem \ref{RHP0}, Proposition \ref{png}, \eqref{delta0} and \eqref{deltainfty0} that $T$ solves the following RH problem. 
% With the aid of Proposition \ref{png} for $g(\lb)$ and the properties \eqref{delta0} and \eqref{deltainfty0} for $\delta(\lb)$,  we transform $Z(\lb)$ into the following new vector function 
% \begin{align}\label{T1}
% 	T(\lb)=T(\lb; n)=\E^{n g^{(\infty)}\sigma_3}\delta(\infty)^{\sigma_3}Z(\lb)\E^{-n(g(\lb)-\frac{\ln\lb}{2})\sigma_3}\delta(\lb)^{-\sigma_3},
% \end{align}
% which in turn satisfies the following RH problem:
\begin{Rhp}\ \label{r3.1}
	\begin{itemize}
		\item $T(\lb)$ is analytic in $\ccc\setminus\Sigma$.
		\item For $\lb\in\Sigma$, $T$ satisfies the jump condition
			$T_+(\lb)=T_-(\lb)J^{(T)}(\lb),$
		where
		\begin{align}
			&J^{(T)}(\lb)=J^{(T)}(\lb;n)\nonumber\\
			&:=\begin{cases}
				\left(\begin{matrix}
					r(\lb)^{-1}\delta_-(\lb)^2\E^{2n g_-(\lb)}&0\\\ii&r(\lb)^{-1}\delta_+(\lb)^2\E^{2n g_+(\lb)}
				\end{matrix}\right),&\lb\in\Sigma_{1},\\
				\left(\begin{matrix}
					r(\bar\lb^{-1})^{-1}\delta_+(\lb)^{-2}\E^{-2n g_+(\lb)}&\ii \\0&r(\bar\lb^{-1})^{-1}\delta_-(\lb)^{-2}\E^{-2n g_-(\lb)}
				\end{matrix}\right),&\lb\in\Sigma_{2},\\
				\E^{\ii(n\Omega+\Delta)\sigma_3},&\lb\in\ii(\eta_1^{-1},\eta_1).
			\end{cases}
		\end{align}
		\item As $\lb\to\infty$, we have
		$	T(\lb)=I+\oo(\lb^{-1}).$
	\end{itemize}
\end{Rhp}
For $\lb\in\Sigma_1\cup\Sigma_2$, it is noticed  that $J^{(T)}$ admits the following factorizations:
\begin{align*}
	J^{(T)}(\lb)=\begin{cases}
		\left(\begin{matrix}
			1&-\ii r(\lb)^{-1}\delta_-(\lb)^2\E^{2n g_-(\lb)}\\0&1
		\end{matrix}\right)\left(\begin{matrix}
			0&\ii\\\ii&0
		\end{matrix}\right)\left(\begin{matrix}
			1&-\ii r(\lb)^{-1}\delta_+(\lb)^2\E^{2n g_+(\lb)}\\0&1
		\end{matrix}\right),&\lb\in\Sigma_{1},\\
		\left(\begin{matrix}
		1&0\\-\ii r(\bar\lb^{-1})^{-1} \delta_-(\lb)^{-2} \E^{-2n g_-(\lb)}&1
	\end{matrix}\right)\left(\begin{matrix}
	0&\ii\\\ii&0
\end{matrix}\right)\left(\begin{matrix}
1&0\\-\ii r(\bar\lb^{-1})^{-1} \delta_+(\lb)^{-2} \E^{-2n g_+(\lb)}&1
\end{matrix}\right),&\lb\in\Sigma_{2}.\\
	\end{cases}
\end{align*}
This invokes us to open lenses around the interval $\Sigma_1\cup\Sigma_2$. To proceed, let $\Omega_{1,\pm}$ be two regions surrounding $\Sigma_1$ and set
\begin{align*}
	\Omega_{2,\pm}=\overline{\Omega_{1,\mp}^{-1}};
\end{align*}
see Figure  \ref{f3} for an illustration. 
% Moreover, to open the lenses, we make the following assumption for the reflection coefficient $r(\lb)$.
% \begin{assumption}\label{assum1}
%     Given the positive constants $\eta_2>\eta_1>1$, $r(\lb)$  is an analytic function with respect to $\lb$ near the interval $ (\i\eta_1, \i\eta_2)$ and takes positive values on the closed interval $ [\i\eta_1, \i\eta_2]$. 
%     %Moreover, $r(\lambda)\in H^1\left((\i\eta_1,\i\eta_2)\right)$.
% \end{assumption}
% To split the lenses, we take the regions $\Omega_{1,\pm}$ as shown in Figure \ref{f3}, and the regions $\Omega_{2,\pm}$ satisfying that
% \begin{align*}
% 	\Omega_{2,\pm}=\overline{\Omega_{1,\mp}^{-1}},
% \end{align*}
By introducing a matrix-valued function
\begin{align}
	G(\lb)=G(\lb;n):=\begin{cases}
		\left(\begin{matrix}
			1&\pm \ii \E^{2n g(\lb)}\delta(\lb)^2r(\lb)^{-1}\\0&1
		\end{matrix}\right),&\lb\in\Omega_{1,\pm},\\
	\left(\begin{matrix}
		1&0\\\pm\ii \E^{-2n g(\lb)}\delta(\lb)^{-2}r(\bar\lb^{-1})^{-1}&1
	\end{matrix}\right),&\lb\in\Omega_{2,\pm},\\
I,&\text{otherwise},
	\end{cases}
\end{align}
we define
%Then, under Assumption \ref{assum1}, we  transform $T(\lb)$ into
\begin{align}\label{T2}
	Z^{(1)}(\lb)=Z^{(1)}(\lb;n)=T(\lb)G(\lb).
\end{align}
Recall that $r(\lb)$ is analytic and positive in a neighborhood of $\Sigma_1 \cup \Sigma_2$, we have that $Z$ satisfies the following RH problem.
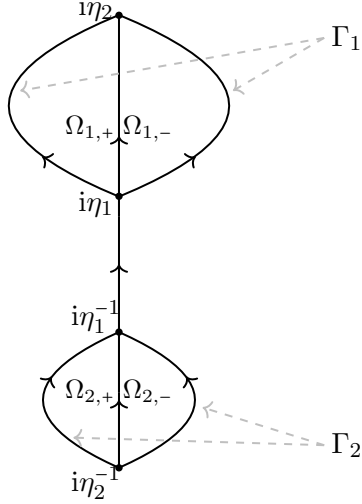
\begin{figure}
	\centering
	\tikzset{every picture/.style={line width=0.75pt}}
	\begin{tikzpicture}
		\draw [->] (0,0) -- (0,0.9);
		\draw [->] (0,1) -- (0,2.7);
		\draw [->] (0,10/3) -- (0,4.4);
		\draw [->] (-1,4.12) -- (-1.01,4.1265);
		\draw [->] (1,4.12) -- (1.01,4.1265);
		\draw [->] (-0.89,1.2) -- (-0.884,1.209);
		\draw [->] (0.89,1.2) -- (0.884,1.209);
		\draw (0,0) -- (0,1);
		\draw (0,1) -- (0,10/3);
		\draw (0,10/3) -- (0,6);
		\draw [domain=3.6:6] plot ({(\x-4.8)*(\x-4.8)-1.45},\x);
		\draw [domain=3.6:6] plot ({-(\x-4.8)*(\x-4.8)+1.45},\x);
		\draw [domain=0:2] plot ({(\x-1)*(\x-1)-1},{0.9*\x});
		\draw [domain=0:2] plot ({-(\x-1)*(\x-1)+1},{0.9*\x});
		\filldraw (0,0) circle (1pt);
		\filldraw (0,1.8) circle (1pt);
		\filldraw (0,3.6) circle (1pt);
		\filldraw (0,6) circle (1pt);
		\draw [->,dashed,lightgray] (2.7,5.7) -- (1.5,5);
		\draw [->,dashed,lightgray] (2.7,5.7) -- (-1.3,5);
		\draw [->,dashed,lightgray] (2.7,0.3) -- (1.1,0.8);
		\draw [->,dashed,lightgray] (2.7,0.3) -- (-0.6,0.4);
		\node at (-0.38,4.5) {\small $\Omega_{1,+}$};
		\node at (0.38,4.5) {\small $\Omega_{1,-}$};
		\node at (-0.38,1) {\small $\Omega_{2,+}$};
		\node at (0.38,1) {\small $\Omega_{2,-}$};
		\node at (3,5.7) {$\Gamma_1$};
		\node at (3,0.3) {$\Gamma_2$};
		\node at (-0.3,-0.2) {$\ii\eta_2^{-1}$};
		\node at (-0.3,2) {$\ii\eta_1^{-1}$};
		\node at (-0.3,3.5) {$\ii\eta_1$};
		\node at (-0.3,6.1) {$\ii\eta_2$};
	\end{tikzpicture}
	\caption{\footnotesize The regions $\Omega_{i,\pm}$, $i=1,2$, and the jump contour $\Sigma^{(1)}$ of $Z^{(1)}$.}\label{f3}
\end{figure}
\begin{Rhp}\label{r3}\
\begin{itemize}
	\item $Z^{(1)}(\lb)$ is analytic in $\ccc\setminus\Sigma^{(1)}$, where the jump contour $\Sigma^{(1)}$ is shown in Figure \ref{f3}.
	\item For $\lambda\in\Sigma^{(1)}$, $Z^{(1)}(\lb)$ satisfies the  jump condition $Z^{(1)}_+(\lb)=Z^{(1)}_-(\lb)J^{(1)}(\lb),$
	where
	\begin{align}\label{def:J1}
		J^{(1)}(\lb)=\begin{cases}
			\left(\begin{matrix}
				1&-\ii \E^{2n g(\lb)}\delta(\lb)^2r(\lb)^{-1}\\0&1
			\end{matrix}\right),&\lb\in\Gamma_1,\\
		\left(\begin{matrix}
			1&0\\-\ii \E^{-2n g(\lb)}\delta(\lb)^{-2}r(\bar\lb^{-1})^{-1}&1
		\end{matrix}\right),&\lb\in\Gamma_2,\\
	\left(\begin{matrix}
		0&\ii\\\ii&0
	\end{matrix}\right),&\lb\in\Sigma_1\cup\Sigma_2,\\
\E^{\ii(n\Omega+\Delta)\sigma_3},&\lb\in\ii(\eta_1^{-1},\eta_1).
		\end{cases}
	\end{align}
	\item As $\lb\to\infty$, we have  $Z^{(1)}(\lambda)=I+\oo(\lambda^{-1})$.
\end{itemize}
\end{Rhp}
From the signature table of Re$[g]$ in Figure \ref{fig reg t=0}, it is easily seen  that the jump matrix $J^{(1)}$ decays exponentially to $I$ on $\Gamma_{1}\cup\Gamma_{2}$ as $n\to-\infty$ except the endpoints. It inspires us to ignore the jump away from $\Sigma=\ii(\eta_2^{-1},\eta_2)$ and approximate $Z^{(1)}$ in global and local manners. More precisely, let
\begin{align}\label{def:Up}
    U(p):=\{\lb:|\lb-p|\leq \epsilon\}, \qquad p=\i\eta_{j}, \quad j=1,2,
\end{align}
where $\epsilon$ is a sufficiently small positive constant,
and 
\begin{align}\label{def:Up2}
    U(p):= \overline{U(\bar p^{-1})^{-1}},  \qquad p=\i\eta_{j}^{-1}, \quad j=1,2.
\end{align}
%Moreover, let $U:=\bigcup_{p=\ii\eta_1^{\pm1},\ii\eta_2^{\pm1}}U(p)$.
In what follows, we will construct  the global parametrix $Z^{(\infty)}(\lb)$ in Section \ref{1ng} and the local parametrix $Z^{(p)}(\lb)$ near $p=\ii \eta_1^{\pm 1}, \ii \eta_2^{\pm 1}$, in Section \ref{1nl} such that
\begin{align}\label{T3}
	Z^{(1)}(\lb)=\begin{cases}
	E(\lb)Z^{(\infty)}(\lb),&\lb\in  \mathbb{C}\setminus U,\\
	E(\lb)Z^{(p)}(\lb),&\lb\in  U(p),\quad p=\ii \eta_1^{\pm 1}, \ii \eta_2^{\pm 1},
	\end{cases}
\end{align}
where $$U:=\bigcup_{p=\ii\eta_1^{\pm1},\ii\eta_2^{\pm1}}U(p)$$ and $E(\lb)$ is an error function satisfying a small norm RH problem as shown in Section \ref{1ne}.

\subsection{Global parametrix}\label{1ng}
The global parametrix is obtained by ignoring the jump of $Z^{(1)}$ off $\Sigma$, which reads as follows.
% As $n\to-\infty$, the jump matrix $J^{(1)}$ decays exponentially to the identity off the contour $\Sigma$, and the global parametrix $Z^{(\infty)}(\lb)$ exactly solves the following RH problem.
\begin{Rhp}\label{r3.2}\
	\begin{itemize}
		\item $Z^{(\infty)}(\lb)=Z^{(\infty)}(\lb; n)$ is analytic in $\ccc\setminus\Sigma$.
        %with $\Sigma$ defined in \eqref{def Sigma}.
		\item For $\lb\in\Sigma$, $Z^{(\infty)}(\lambda)$ satisfies the jump conditions $Z_+^{(\infty)}(\lambda)=Z_-^{(\infty)}(\lambda) J^{(\infty)}(\lambda)$, where
		\begin{align*}
		    J^{(\infty)}(\lb)&=\begin{cases}
				\left(\begin{matrix}
					0&\ii\\\ii&0
				\end{matrix}\right),&\lb\in\Sigma_1\cup\Sigma_2,\\
			\E^{\ii(n\Omega+\Delta)\sigma_3},&\lb\in\ii(\eta_1^{-1},\eta_1).
			\end{cases}
		\end{align*}
		\item As $\lb\to\infty$,  we have $Z^{(\infty)}(\lb)=I+\oo(\lb^{-1})$.
	\end{itemize}
\end{Rhp}
%The above RH problem can be exactly solved.
% To construct the solution for RH problem \ref{r3.2}, firstly, we need to examine the Abel differential associated to the Riemann surface $\mathcal{R}$.
%In fact, by direct computation, it follows that
%\begin{align*}
 %  \int_{\ii \eta_1^{-1}}^{\ii \eta_1} \frac{\ddd \lb}{R(\lb)}=\frac{K(k)\sqrt{(1-l_1^2)(1-l_2^2)}}{\ii l_2},
%\end{align*}
%where $k$ is defined in \eqref{def k}. Consequently,
To solve the above RH problem, we start with the normalized Abel differential given by
%In view of \eqref{int 1/R(s)}, it is readily seen that the normalized Abel differential reads
\begin{align*}
	\omega(\lb)=\frac{\ii l_2}{2K(k)\sqrt{(1-l_1^2)(1-l_2^2)}}\frac{\ddd \lb}{R(\lb)},
\end{align*}
where $k$ is shown in \eqref{def k}, and define the Abel-Jacobi map
\begin{align}\label{defA}
	\mathcal{A}(\lb)=\int_{\ii\eta_2}^{\lb}\omega, \quad \lb \in \C \setminus\Sigma,
\end{align}
where the path of integration is on any simple arc from $\i \eta_2$ to $\lb$ which does not intersect $\Sigma$.
By direct calculations, we have 
%the values of $\mathcal{A}(\lb)$ at some special points:
\begin{align*}
	&\AAA_+(\ii\eta_1)=-\frac{\tau}{2},\qquad\AAA_+(\ii\eta_1^{-1})=-\frac{\tau}{2}-\frac{1}{2}, \\
	&\AAA_+\left(\frac{\ii(\eta_1+\eta_2)}{1+\eta_1\eta_2}\right)=\AAA(0)-\frac{\tau}{2},\qquad\AAA(0)=-\frac{1}{4}\left(1+\frac{F(\arcsin l_2,k)}{K(k)}\right),
\end{align*}
where
\begin{align*}
	\tau:=\oint_{\mathfrak{b}}\omega=\frac{\ii l_2}{K(k)\sqrt{(1-l_1^2)(1-l_2^2)}}\int_{\ii\eta_1}^{\ii\eta_2}\frac{\ddd  s}{R(s)}=\frac{\ii K(\sqrt{1-k^2})}{2K(k)}\in\ii\mathbb{R}_+
\end{align*}
is the $\mathfrak{b}$-period and $F$ is defined in \eqref{ei2}. In addition, it is straightforward to check from the definition of $\AAA$ that
%For any $\lb\in\ccc\setminus\Sigma$,  a direct calculation at $\bar\lb^{-1}$ shows that
\begin{align}
	\overline{\AAA(\bar\lb^{-1})}+\AAA(\lb)=-\frac{1}{2}, \label{sym AAA}
\end{align}
and 
%In addition, for $\lambda\in\Sigma$,  evaluating the defining integrals on each side of the contour gives the jump relations
\begin{align}\label{jumpA}
	\AAA_+(\lb)=\begin{cases}
		-\AAA_-(\lb),&\lb\in\Sigma_1,\\
		\AAA_-(\lb)-\tau,&\lb\in\ii(\eta_1^{-1},\eta_1),\\
		-\AAA_-(\lb)+1,&\lb\in\Sigma_2.
	\end{cases}
\end{align}
Next, we introduce a scalar function
\begin{align}\label{kappa}
	\kappa(\lb):=\left(\frac{\lb-\ii\eta_2}{\lb-\ii\eta_1}\right)^{\frac{1}{4}}\left(\frac{\lb-\ii/\eta_1}{\lb-\ii/\eta_2}\right)^{\frac{1}{4}}, \qquad \lb\in\ccc\setminus(\Sigma_1\cup\Sigma_2),
\end{align}
where the branch cut is chosen such that $\kappa(\lb) \to 1$ as $\lb \to \infty$.
Moreover, it is noted that $\kappa(\lb)-\kappa(\lb)^{-1}$ has a simple zero at $\lb=\ii\frac{\eta_1+\eta_2}{1+\eta_1\eta_2}$, and one has
\begin{align}\label{jumpkap}
	\kappa_+(\lb)=\ii\kappa_-(\lb),\qquad \lb\in\Sigma_1\cup\Sigma_2.
\end{align}

With the aid of the scalar functions $\AAA(\lb)$ and $\kappa(\lb)$, we define
%\begin{align}\label{def:Zinfty}
%  Z^{(\infty)}(\lb) := \begin{pmatrix}
%    Z^{(\infty)}_{11}(\lb) &  Z^{(\infty)}_{12}(\lb) \\
%    Z^{(\infty)}_{21}(\lb) &  Z^{(\infty)}_{22}(\lb)
%\end{pmatrix}
%\end{align}
%with 
\begin{subequations}\label{zinfty}
\begin{align}
	&Z^{(\infty)}_{11}(\lb)=\frac{(\kappa(\lb)+1/\kappa(\lb))\vartheta(\AAA(\lb)+\AAA(0)+\frac{1}{2}+\frac{n\Omega+\Delta}{2\pi},\tau)\vartheta(0,\tau)}{2\vartheta(\AAA(\lb)+\AAA(0)+\frac{1}{2},\tau)\vartheta(\frac{n\Omega+\Delta}{2\pi},\tau)},\\
	&Z^{(\infty)}_{12}(\lb)=\frac{(\kappa(\lb)-1/\kappa(\lb))\vartheta(-\AAA(\lb)+\AAA(0)+\frac{1}{2}+\frac{n\Omega+\Delta}{2\pi},\tau)\vartheta(0,\tau)}{2\vartheta(-\AAA(\lb)+\AAA(0)+\frac{1}{2},\tau)\vartheta(\frac{n\Omega+\Delta}{2\pi},\tau)}, \label{def:Zinfty12}\\
	&Z^{(\infty)}_{21}(\lb)=\frac{(\kappa(\lb)-1/\kappa(\lb))\vartheta(\AAA(\lb)-\AAA(0)-\frac{1}{2}+\frac{n\Omega+\Delta}{2\pi},\tau)\vartheta(0,\tau)}{2\vartheta(\AAA(\lb)-\AAA(0)-\frac{1}{2},\tau)\vartheta(\frac{n\Omega+\Delta}{2\pi},\tau)},\\
	&Z^{(\infty)}_{22}(\lb)=\frac{(\kappa(\lb)+1/\kappa(\lb))\vartheta(-\AAA(\lb)-\AAA(0)-\frac{1}{2}+\frac{n\Omega+\Delta}{2\pi},\tau)\vartheta(0,\tau)}{2\vartheta(-\AAA(\lb)-\AAA(0)-\frac{1}{2},\tau)\vartheta(\frac{n\Omega+\Delta}{2\pi},\tau)},
\end{align}
\end{subequations}
where $\vartheta(\lb,\tau)$ is the Riemann theta function defined by the Fourier series
\begin{align}\label{fou}
    \vartheta(\lb,\tau)=\sum_{l\in\mathbb{Z}}e^{2\pi\ii l\lb+\pi\ii l^2\tau},\quad \lb \in \C
\end{align}
associated with $\tau$. Note that $\vartheta(\lb,\tau)$ has a simple zero at $\lb=\frac{\tau+1}{2}$. From \eqref{defA}, \eqref{kappa} and \eqref{fou}, it follows that $Z^{(\infty)}(\lb)$ is analytic in $\C\setminus \Sigma$, and by \eqref{jumpA} and \eqref{jumpkap}, $Z^{(\infty)}(\lb)$ satisfies the jump condition of RH problem \ref{r3.2}. Since $\kappa \to 1$ as $\lb \to \infty$, this, together with the symmetry relation \eqref{sym AAA}, implies that $Z^{(\infty)}(\lb) \to I$. Thus, $Z^{(\infty)}(\lb)$ in \eqref{zinfty} indeed solves RH problem \ref{r3.2} and, in addition, satisfies the symmetry relation
\begin{align}\label{e3.20}
	Z^{(\infty)}(\lb)=\sigma_2 \overline{Z^{(\infty)}(0)^{-1} Z^{(\infty)} (\bar{\lb}^{-1})} \sigma_2.
\end{align}

\subsection{Local parametrices}\label{1nl}
Let $p=\ii\eta_j^{\pm1}$, $j=1,2$, and $U(p)$ defined in \eqref{def:Up} and \eqref{def:Up2}
be a small disc centered at $p$. The local paramertix in each $U(p)$ reads as follows. 
% In this part, we will give the local parametrices near the end points of $\Sigma_1$ and $\Sigma_2$, namely, for endpoints $p=\ii\eta_j^{\pm1}$, $j=1,2$. As $n\to-\infty$, we will show that  $Z^{(p)}(\lb)$ is described in terms of the Bessel function of index $0$ with a detailed introduction in Appendix \ref{app bessel}. In each $U(p)$, the jump conditions inspire us to construct $Z^{(p)}(\lb)$ satisfying the following RH problem.
\begin{Rhp}\  \label{r6}
\begin{itemize}
\item $Z^{(p)}(\lb)=Z^{(p)}(\lb; n)$ is analytic in $U(p)\setminus(\Gamma_1\cup\Gamma_2\cup\Sigma)$.
	\item For $\lb\in U(p)\cap(\Gamma_1\cup\Gamma_2\cup\Sigma)$, $Z^{(p)}(\lb)$ satisfies the jump condition
	\begin{align}
		Z^{(p)}_+(\lb)=Z^{(p)}_-(\lb)\times\begin{cases}
			\left(\begin{matrix}
				1&-\ii \E^{2n g(\lb)}\delta(\lb)^2r(\lb)^{-1}\\0&1
			\end{matrix}\right),&\lb\in\Gamma_1\cap U(p),\\
			\left(\begin{matrix}
				1&0\\-\ii \E^{-2n g(\lb)}\delta(\lb)^{-2}r(\bar\lb^{-1})^{-1}&1
			\end{matrix}\right),&\lb\in\Gamma_2\cap U(p),\\
			\left(\begin{matrix}
				0&\ii\\\ii&0
			\end{matrix}\right),&\lb\in(\Sigma_1\cup\Sigma_2)\cap U(p),\\
			\E^{\ii(n\Omega+\Delta)\sigma_3},&\lb\in\ii(\eta_1^{-1},\eta_1)\cap U(p).
		\end{cases}
	\end{align}
	\item As $n\to-\infty$, $Z^{(p)}(\lb)$ matches $Z^{(\infty)}(\lb)$ on the boundary $\partial U(p)$.
\end{itemize}
\end{Rhp}
The RH problem for $Z^{(p)}$ can be solved explicitly with the aid of the Bessel parametrix introduced in Appendix \ref{app bessel}. We give a sketch of the construction in what follows. 

%Here and the rest of our paper, we  denote the clockwise oriented boundary of $ U(p)$ as $\partial U(p)$.
% In the following, we need to establish the locally conformal maps with $\zeta(p)=0$,
% \begin{align*}
% 	\zeta:\ \lb\mapsto\zeta(\lb).
% \end{align*}

For $p=\ii\eta_2$, define
\begin{align}
\zeta=	\zeta(\lb):=g(\lb)^2.
\end{align}
By \eqref{asy g eta2}, it is easily seen that $\zeta(\lb)$ is analytic in $ U(\ii\eta_2)$ and $\zeta(\ii\eta_2)=0$.
%then in view of the asymptotic behaviors of $g(\lb)$ near $\ii\eta_2$ in \eqref{asy g eta2}, it follows that $\zeta(\lb)$ is analytic in $ U(\ii\eta_2)$ and naturally maps $\ii\eta_2$ to the origin point.
The local paramatrix around $\lb=\i \eta_2$ is given by
\begin{align}\label{1neta2}
	Z^{(\ii \eta_2)}(\lb)=A(\lb)\Psi_{\mathrm{Bes}}\left(\frac{n^2\zeta}{4}\right)\sigma_1\E^{n\sqrt{\zeta}\sigma_3}\left(\frac{\delta(\lb)\E^{-\frac{\pi\ii}{4}}}{\sqrt{r(\lb)}}
	\right)^{-\sigma_3},
\end{align}
where $\Psi_{\mathrm{Bes}}$ defined in \eqref{def:Bes} is the Bessel parametrix and
%is the solution of the RH problem related to the Bessel function of index 0 given in Appendix \ref{app bessel}, and 
\begin{equation*}
	A(\lb)=\frac{Z^{(\infty)}(\lb)}{\sqrt{2}}\left(\frac{\delta(\lb)\E^{-\frac{\pi\ii}{4}}}{\sqrt{r(\lb)}}
	\right)^{\sigma_3}\left(\begin{matrix}
		-\ii&1\\1&-\ii
	\end{matrix}\right)(-n\pi\sqrt{\zeta})^{\frac{\sigma_3}{2}}
\end{equation*}
with $Z^{\infty}$ and $\delta$ given in \eqref{zinfty} and \eqref{symdeltan}, respectively. 
Since $A(\lb)$ has no jump in $U(\i\eta_2)$ and admits at most $-\frac{1}{4}$-singularity at $\lb=\i\eta_2$, it follows that $A(\lb)$ is an analytic prefactor. 
In view of RH problem \ref{rhp:Bes} for $\Psi_{\mathrm{Bes}}$, one can check directly that $Z^{(\ii \eta_2)}$ in \eqref{1neta2} solves RH problem  \ref{r6} with $p=\ii \eta_2$.
%Using the global RH problem \ref{r3.2} and the Bessel parametrix in Appendix \ref{app bessel}, it is easily checked that $Z^{(\ii\eta_2)}(\lb)$ meets the requirements in RH problem \ref{r6}.

Similarly, for $\lb \in U(\ii \eta_1)$, we introduce
\begin{align}
	\zeta=\zeta(\lb):=\left(g(\lb)\pm\frac{\ii\Omega}{2}\right)^2,\quad \pm\re\lb>0,
\end{align}
and 
%Then, the local paramatrix around the endpoint $\lb=\i \eta_1$  can be constructed by
\begin{align}\label{def Zeta1}
	Z^{(\ii \eta_1)}(\lb):=A(\lb)\Psi_{\mathrm{Bes}}\left(\frac{n^2\zeta}{4}\right)\sigma_1\E^{(-n\sqrt{\zeta}\pm\frac{\ii\Omega}{2
		})\sigma_3}\left(\frac{\delta(\lb)\E^{-\frac{\pi\ii}{4}}}{\sqrt{r(\lb)}}\right)^{-\sigma_3},
\end{align}
where 
%$A(\lb)$ is analytic in $U(\i\eta_1)$ that
\begin{align*}
	A(\lb)=\frac{Z^{(\infty)}(\lb)}{\sqrt{2}}\E^{\mp\frac{\ii\Omega}{2}\sigma_3}\left(\frac{\delta(\lb)\E^{-\frac{\ii\pi}{4}}}{\sqrt{r(\lb)}}\right)^{\sigma_3}\left(\begin{matrix}
		-\ii&1\\1&-\ii
	\end{matrix}\right)(-n\pi\sqrt{\zeta})^{\frac{\sigma_3}{2}},
\end{align*}
solves RH problem  \ref{r6} with $p=\ii \eta_1$.
%Also,  using RH problem \ref{r3.2} and the Bessel parametrix in Appendix \ref{app bessel}, it follows that $Z^{(\ii\eta_1)}(\lb)$ satisfies the requirements in RH problem \ref{r6}.
Meanwhile, near $\ii\eta_j^{-1}$, $j=1,2$,  the local parametrix can be constructed through the symmetry relation
\begin{align}\label{1ni}
	Z^{(\ii \eta_j^{-1})}(\lb)=\sigma_2\overline{Z^{(\infty)}(0)^{-1}Z^{(\ii \eta_j)}(\bar\lb^{-1})}\sigma_2.
\end{align}
%Additionally, by the symmetry \eqref{e3.20}, since $Z^{(\ii\eta_j)}(\lb)$ have fulfilled the requirements, it is easy to confirm that $Z^{(\ii\eta_j^{-1})}(\lb)$ also fulfill them.
Finally, we note from \eqref{zinfty},
\eqref{1neta2}, \eqref{def Zeta1}, \eqref{1ni} and \eqref{asy bessel} that as $n\to-\infty$,
\begin{align}\label{asyZ2n}
    Z^{(p)}(\lb)Z^{(\infty)}(\lb)^{-1}=I+\oo(n^{-1}),\qquad \lb\in \partial U(p). 
\end{align}
%Therefore, we conclude that $Z^{(p)}(\lb)$ defined in \eqref{1neta2}, \eqref{def Zeta1} and \eqref{1ni} admits the local RH problem \ref{r6} for $	p=	\ii\eta_j^{\pm1}$, $j=1,2$.	

\subsection{The small-norm RH problem}\label{1ne}
%Previously,
In view of RH problems for $Z^{(1)}$, $Z^{(\infty)}$ and $Z^{(p)}$, 
it is readily seen that the error function (see \eqref{T3}) 
\begin{align}\label{defEc1}
					E(\lb)=\left\{\begin{array}{ll}
						Z^{(1)}(\lb)Z^{(\infty)}(\lb)^{-1}, &\lb\in \mathbb{C}\setminus U,\\
						Z^{(1)}(\lb)Z^{(p)}(\lb)^{-1},\quad &\lb\in   U(p),\; p=\i \eta_1^{\pm1},\i \eta_2^{\pm1},\\
					\end{array}\right.
				\end{align}
satisfies the following RH problem.
%The small-norm RH problem is uniquely solved, and consequently, the solution of RH problem \ref{r3} is well-defined via
%\begin{align}\label{T3}
%	Z^{(1)}(\lb)=\begin{cases}
%	E(\lb)Z^{(\infty)}(\lb),&\lb\in  \mathbb{C}\setminus U,\\
%	E(\lb)Z^{(p)}(\lb),&\lb\in  U(p),\ p=\ii \eta_1^{\pm 1}, \ii \eta_2^{\pm 1}.
%	\end{cases}
%\end{align}
\begin{Rhp}\label{rerrorn}\
	\begin{itemize}
		\item $E(\lambda)$ is analytic in $\ccc\setminus\Sigma^E$, where the contour $\Sigma^E$ is shown in Figure \ref{F5}.
		\item For $\lambda\in\Sigma^E$, $E(\lambda)$ satisfies the jump condition
		$E_+(\lambda)=E_-(\lambda)J^{(E)}(\lambda)$,
	where
	\begin{align}\label{def:JE}
		J^{(E)}(\lambda)=\begin{cases}
			Z^{(\infty)}(\lambda)J^{(1)}(\lambda)Z^{(\infty)}(\lambda)^{-1},&\lb\in\Gamma_j\setminus U,\ j=1,2,\\
			Z^{(p)}(\lambda)Z^{(\infty)}(\lb)^{-1},&\lb\in\partial U(p), \ p=\ii \eta_1^{\pm1}, \ii \eta_2^{\pm 1},
		\end{cases}
	\end{align}
    and $J^{(1)}$ is defined in \eqref{def:J1}. 
		\item As $\lambda\to\infty$, we have $E(\lambda)=I+\oo(\lambda^{-1})$.
	\end{itemize}
\end{Rhp}
Since the jump matrix $J^{(1)}$ of $Z^{(1)}$  tends to the identity matrix exponentially fast uniformly for $\lb\in(\Gamma_1\cup\Gamma_2)\setminus U$ as $n\to-\infty$, we have from   \eqref{def:JE} and \eqref{asyZ2n} that
\begin{align}\label{asymVEn1}
    J^{(E)}(\lb)=I+\oo(n^{-1}),\quad n\to-\infty.
\end{align}
By the small-norm RH problem theory \cite{RN10}, we conclude that
\begin{align}\label{eerror}
    E(\lb)=I+\oo(n^{-1}),\qquad n\to-\infty,
\end{align}
uniformly for $\lb \in \ccc\setminus\Sigma^E$. 

% we obtain from  \eqref{def:JE} and \eqref{asyZ2n} that
% \begin{align}\label{eVb1}
%     J^{(E)}(\lb)=I+\oo(n^{-1}),\qquad \lb\in\partial U.
% \end{align}
% In addition, seeing the definition of $J^{(1)}(\lb)$ in RH problem \ref{r3}, we obtain that $J^{(1)}(\lb)-I$ decays exponentially to zero as $n\to-\infty$, that is, for some constant $c$ and $\lb\in(\Gamma_1\cup\Gamma_2)\setminus U$,
% \begin{align}\label{eVb2}
%     J^{(1)}(\lb)=I+\oo(\E^{-cn}).
% \end{align}
% Therefore, since both $Z^{(\infty)}(\lb)$ and its inverse are bounded on $\Gamma_1\cup\Gamma_2$, in view of  \eqref{eVb2}, it is also easily seen that for $\lb\in\Gamma_j\setminus U$,
% \begin{align}\label{asymVEn1}
%     J^{(E)}(\lb)=I+\oo(n^{-1}),\quad n\to-\infty.
% \end{align}
% Using the asymptotic properties \eqref{eVb1} and \eqref{asymVEn1}, it follows from the small-norm RH problem theory \cite{RN10} that there exists a unique solution to RH problem   \ref{rerrorn} and 
% \begin{align}\label{eerror}
%     E(\lb)=I+\oo(n^{-1}),\quad n\to-\infty.
% \end{align}
%Finally, since RH problem \ref{rerrorn} is uniquely solved, the solution of RH problem \ref{r3} is well-defined and we obtain \eqref{T3}. 

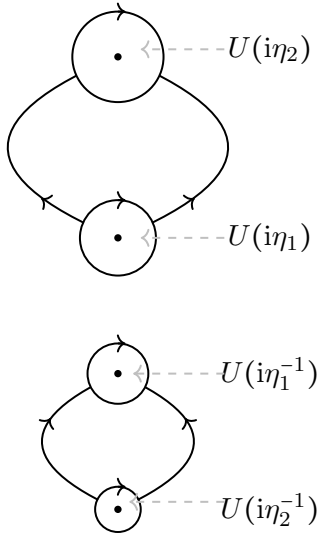
\begin{figure}
	\centering
	\tikzset{every picture/.style={line width=0.75pt}}
	\begin{tikzpicture}
		\draw [->] (-1,4.12) -- (-1.01,4.1265);
		\draw [->] (1,4.12) -- (1.01,4.1265);
		\draw [->] (-0.89,1.2) -- (-0.884,1.209);
		\draw [->] (0.89,1.2) -- (0.884,1.209);
		\draw [->] (0,6.6) -- (0.1,6.6);
		\draw [->] (0,4.1) -- (0.1,4.1);
		\draw [->] (0,2.2) -- (0.1,2.2);
		\draw [->] (0,0.3) -- (0.1,0.3);
		\draw [domain=3.8:5.75] plot ({(\x-4.8)*(\x-4.8)-1.45},\x);
		\draw [domain=3.8:5.75] plot ({-(\x-4.8)*(\x-4.8)+1.45},\x);
		\draw [domain=0.14:1.8] plot ({(\x-1)*(\x-1)-1},{0.9*\x});
		\draw [domain=0.14:1.8] plot ({-(\x-1)*(\x-1)+1},{0.9*\x});
		\draw (0,0) circle (0.3cm);
		\draw (0,1.8) circle (0.4cm);
		\draw (0,3.6) circle (0.5cm);
		\draw (0,6) circle (0.6cm);
		\filldraw (0,0) circle (1pt);
		\filldraw (0,1.8) circle (1pt);
		\filldraw (0,3.6) circle (1pt);
		\filldraw (0,6) circle (1pt);
		\node at (2,6.1) {$U(\ii\eta_2)$};
		\draw [->,dashed,lightgray] (1.4,6.1) -- (0.3,6.1);
		\node at (2,3.6) {$U(\ii\eta_1)$};
		\draw [->,dashed,lightgray] (1.4,3.6) -- (0.3,3.6);
		\node at (2,1.8) {$U(\ii\eta_1^{-1})$};
		\draw [->,dashed,lightgray] (1.4,1.8) -- (0.2,1.8);
		\node at (2,0) {$U(\ii\eta_2^{-1})$};
		\draw [->,dashed,lightgray] (1.4,0.1) -- (0.15,0.1);
	\end{tikzpicture}
	\caption{\footnotesize The jump contour $\Sigma^E$ of $E$.}\label{F5}
\end{figure}

\section{Large-$t$ asymptotic analysis of the RH problem for $Z(\lb;n+1,t)$}\label{sec 4}

To establish the large-$t$ asymptotics of $q_n$, one needs to analyze the RH problem for  $Z(\lb; n+1,t)$ as $t\to +\infty$, which is the main goal of this section. 
% In this section, we consider the asymptotic behavior of $q_n(t)$ as $t\to +\infty.$
% The  reconstruction formula \eqref{rec} inspires us to analyze the large-time behavior of $Z(\lb; n+1,t)$.
Note that, in this case, the phase function $\phi$ defined in \eqref{phi} in the jump of $Z(\lb; n+1,t)$ reads
\begin{align}\label{def phi}
\phi(\lb):=\phi(\lb; n+1,t)=-\ii t(\lb+\lb^{-1}-2)+(n+1)\ln\lb.
\end{align}
% For brevity, we shall write simply $\phi(\lb)$ in place of $\phi(\lb; n+1,t)$ in this section.
% The large-time behavior of $Z(\lb;n+1,t)$ is then determined by the signature table of Re$[\phi(\lb)]$.
A direct calculation shows that
\begin{align*}
	\Re [\phi(\lb)]=t\left( \I\lb(1-|\lb|^{-2})+\frac{n+1}{t}\ln|\lb|\right).
\end{align*}
If 
$
 \xi=\frac{n+1}{t}>-\frac{\eta_{1}-\eta_{1}^{-1}}{\ln\eta_{1}},
$
% we derive that when
%  \begin{align*}
% 	\xi>-\frac{\eta_{1}-\eta_{1}^{-1}}{\ln\eta_{1}},
% \end{align*}
we see from the signature table of $\Re [\phi(\lb)]$ illustrated in Figure \ref{fig phi} and \eqref{J} that the jump matrix $J$ tends to the identity matrix exponentially fast as $t\to+\infty$, which implies that $Z(\lb;n+1,t) \to I $ as $t\to+\infty$. As a consequence, we will only carry out a detailed large-$t$ asymptotic analysis of RH  problem \ref{RHP0} for $Z(\lb;n+1,t)$ with $\xi<-\frac{\eta_{1}-\eta_{1}^{-1}}{\ln\eta_{1}}$, or equivalently, in the regions $T_I, H_{I}, T_{II}$ and $H_{II}$ given in Definition \ref{def region}. The analysis will differ in different regions. We start with the introduction of some auxiliary functions used when $\xi \in H_I \cup T_{II} \cup H_{II}$.

%As in the previous section, we start with the introduction of some auxiliary functions
% focus on the case that 
% $\xi<-\frac{\eta_{1}-\eta_{1}^{-1}}{\ln\eta_{1}}$, we consider the asymptotic behavior of the gas RH  problem \ref{RHP0}. Firstly, we introduce the  so-called $g$-function mechanism \cite{gfunction, G2001,GT2002}.
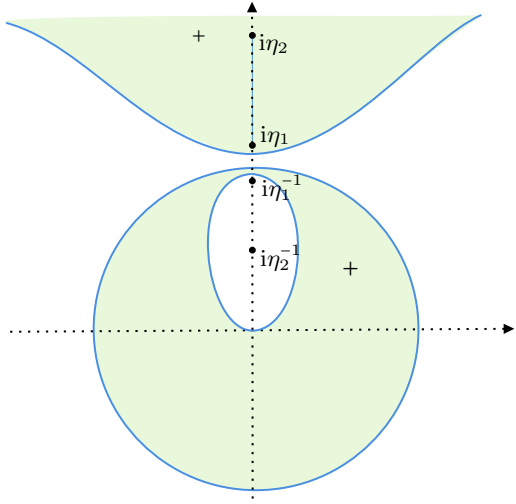
\begin{figure}		
\centering
	\tikzset{every picture/.style={line width=0.75pt}} %set default line width to 0.75pt	
	\begin{tikzpicture}[x=0.75pt,y=0.75pt,yscale=-0.6,xscale=0.6]
%Shape: Polygon Curved [id:ds4117880296419856] 
\draw  [color={rgb, 255:red, 0; green, 0; blue, 0 }  ,draw opacity=0 ][fill={rgb, 255:red, 184; green, 233; blue, 134 }  ,fill opacity=0.3 ] (421.67,114.01) .. controls (383.67,148.01) and (340.41,161.2) .. (317.4,161.36) .. controls (294.39,161.52) and (251.67,144.68) .. (211,112.01) .. controls (170.33,79.34) and (136.67,56.84) .. (113.67,51.34) .. controls (90.67,45.84) and (288.92,45.33) .. (317.35,45.38) .. controls (345.79,45.43) and (505.67,46.01) .. (509,44.68) .. controls (512.33,43.34) and (459.67,80.01) .. (421.67,114.01) -- cycle ;
%Shape: Circle [id:dp5693389500073517] 
\draw  [color={rgb, 255:red, 74; green, 144; blue, 226 }  ,draw opacity=1 ][fill={rgb, 255:red, 184; green, 233; blue, 134 }  ,fill opacity=0.3 ] (186.49,307.99) .. controls (186.49,233.37) and (246.98,172.89) .. (321.6,172.89) .. controls (396.21,172.89) and (456.7,233.37) .. (456.7,307.99) .. controls (456.7,382.61) and (396.21,443.1) .. (321.6,443.1) .. controls (246.98,443.1) and (186.49,382.61) .. (186.49,307.99) -- cycle ;
%Straight Lines [id:da18159634992021056] 
\draw [color={rgb, 255:red, 74; green, 144; blue, 226 }  ,draw opacity=0 ][fill={rgb, 255:red, 74; green, 144; blue, 226 }  ,fill opacity=1 ][line width=0.75]  [dash pattern={on 0.84pt off 2.51pt}]  (318.5,61.86) -- (318.5,154.02) ;
%Straight Lines [id:da26113900821192315] 
\draw [color={rgb, 255:red, 0; green, 0; blue, 0 }  ,draw opacity=1 ] [dash pattern={on 0.84pt off 2.51pt}]  (117,310.28) -- (533.5,307.8) ;
\draw [shift={(536.5,307.78)}, rotate = 179.66] [fill={rgb, 255:red, 0; green, 0; blue, 0 }  ,fill opacity=1 ][line width=0.08]  [draw opacity=0] (8.93,-4.29) -- (0,0) -- (8.93,4.29) -- cycle    ;
%Straight Lines [id:da5042170200967123] 
\draw [color={rgb, 255:red, 255; green, 255; blue, 255 }  ,draw opacity=0 ][fill={rgb, 255:red, 74; green, 144; blue, 226 }  ,fill opacity=0 ][line width=1.5]    (318.5,183.86) -- (318.5,241.86) ;
%Curve Lines [id:da0457622746347256] 
\draw [color={rgb, 255:red, 74; green, 144; blue, 226 }  ,draw opacity=1 ]   (113.67,51.34) .. controls (182.17,69.84) and (236.19,159.55) .. (317.4,161.36) ;
%Curve Lines [id:da6736321485044326] 
\draw [color={rgb, 255:red, 74; green, 144; blue, 226 }  ,draw opacity=1 ]   (509,44.68) .. controls (459,67.34) and (398.61,159.55) .. (317.4,161.36) ;
%Curve Lines [id:da6226795203061964] 
\draw [color={rgb, 255:red, 74; green, 144; blue, 226 }  ,draw opacity=1 ][fill={rgb, 255:red, 255; green, 255; blue, 255 }  ,fill opacity=1 ]   (317.67,178.01) .. controls (374.18,181.72) and (364.18,307.61) .. (319,309.39) ;
%Curve Lines [id:da06582762230718975] 
\draw [color={rgb, 255:red, 74; green, 144; blue, 226 }  ,draw opacity=1 ][fill={rgb, 255:red, 255; green, 255; blue, 255 }  ,fill opacity=1 ]   (317.67,178.01) .. controls (261.16,181.72) and (278.18,309.61) .. (319,309.39) ;
%Straight Lines [id:da6615319433178072] 
\draw  [dash pattern={on 0.84pt off 2.51pt}]  (318.75,450.15) -- (318.25,36.57) ;
\draw [shift={(318.25,33.57)}, rotate = 89.93] [fill={rgb, 255:red, 0; green, 0; blue, 0 }  ][line width=0.08]  [draw opacity=0] (8.93,-4.29) -- (0,0) -- (8.93,4.29) -- cycle    ;
\draw [shift={(318.5,241.86) }, rotate = 90] [color={rgb, 255:red, 0; green, 0; blue, 0}  ,draw opacity=1 ][fill={rgb, 255:red, 0; green, 0; blue, 0}  ,fill opacity=1 ][line width=1]      (0, 0) circle [x radius= 1.74, y radius= 1.74]   ;
\draw [shift={(318.5,183.86)}, rotate = 90] [color={rgb, 255:red, 0; green, 0; blue, 0}  ,draw opacity=1 ][fill={rgb, 255:red, 0; green, 0; blue, 0}  ,fill opacity=1 ][line width=1]      (0, 0) circle [x radius= 1.74, y radius= 1.74]   ;
\draw [shift={(318.5,61.86) }, rotate = 90] [color={rgb, 255:red, 0; green, 0; blue, 0}  ,draw opacity=1 ][fill={rgb, 255:red, 0; green, 0; blue, 0}  ,fill opacity=1 ][line width=1]      (0, 0) circle [x radius= 1.74, y radius= 1.74]   ;
\draw [shift={(318.5,154.02)}, rotate = 90] [color={rgb, 255:red, 0; green, 0; blue, 0}  ,draw opacity=1 ][fill={rgb, 255:red, 0; green, 0; blue, 0}  ,fill opacity=1 ][line width=1]      (0, 0) circle [x radius= 1.74, y radius= 1.74]   ;

% Text Node
\draw (323.6,176.29) node [anchor=north west][inner sep=0.75pt]  [font=\footnotesize]  {$\mathrm{i} \eta _{1}^{-1}$};
% Text Node
\draw (321.3,58.25) node [anchor=north west][inner sep=0.75pt]  [font=\footnotesize]  {$\mathrm{i} \eta _{2}$};
% Text Node
\draw (322.5,235) node [anchor=north west][inner sep=0.75pt]  [font=\footnotesize]  {$\mathrm{i} \eta _{2}^{-1}$};
% Text Node
\draw (322.5,135) node [anchor=north west][inner sep=0.75pt]  [font=\footnotesize]  {$\mathrm{i} \eta _{1}$};
% Text Node
\draw (265.37,54.41) node [anchor=north west][inner sep=0.75pt]  [font=\footnotesize]  {$+$};
% Text Node
\draw (390.53,249.24) node [anchor=north west][inner sep=0.75pt]    {$+$};
\end{tikzpicture}
	\caption{\footnotesize The signature table of $\Re[\phi]$ when $	\xi>-\frac{\eta_{1}-\eta_{1}^{-1}}{\ln\eta_{1}}$. $\Re [\phi]>0$ in the green region,  while $\Re[\phi]<0$ in the white region, and $\Re[\phi]=0$ on the blue curve. }\label{fig phi}
\end{figure}

\subsection{The auxiliary functions}\label{aux2}
% In region $T_I$, the off-diagonal entries of the jump matrix $J$ given in \eqref{J} are exponentially small as $t \to +\infty$;  hence no auxiliary functions are required. By contrast, in regions $H_I \cup T_{II} \cup H_{II}$ (see  Definition \ref{def region}), we introduce $g$-functions to control the exponentially growing entries in the jump matrix $J$.
%For $\xi \in T_I$, we denote
%\begin{align*}
%	g(\lb)=\frac{\phi(\lb)}{2t},
%\end{align*}
%while for $\xi \in H_I \cup T_{II} \cup H_{II}$ given in Definition \ref{def region}, we need a new phase function.
For $\lb\in \mathbb{C}\setminus (\i[\eta_1,\alpha(\xi) ]\cup \i [\alpha(\xi)^{-1},\eta_1^{-1}])$, we define
\begin{align}
R(\lb)=R(\lb;\xi):=\sqrt{(\lb-\i\eta_{1})(\lb-\i\alpha(\xi))(\lb-\i\eta_{1}^{-1})(\lb-\i\alpha(\xi)^{-1})},
\label{def R}
\end{align}
which determines a  genus-1 Riemann surface $\mathcal{R}$  with a homology basis $\{\mathfrak{a}, \mathfrak{b}\}$ as shown in Figure \ref{a&b}, 
% Note that the branch cuts of $R(\lb)$ are $\i(\eta_1,\alpha(\xi) )\cup \i (\alpha(\xi)^{-1},\eta_1^{-1})$.
and we choose the principal sheet so that $R(\lb)\to-\lb^2$ as $\lb\to\infty$. Here, 
\begin{align}\label{def alpha}
		\alpha(\xi)=\left\{\begin{array}{ll}
	\text{solution of \eqref{equ xialpha}},&\xi\in H_I,\\
		\eta_2,& \xi \in T_{II}\cup H_{II},\\	
	\end{array}\right.
\end{align}
where the equation for $\alpha(\xi)$ when $\xi\in H_I$ is given by 
\begin{align}\label{equ xialpha}
	%&\xi=\frac{-\frac{(\alpha(\xi) - 1)^2}{\alpha(\xi)}\left(\frac{(\alpha(\xi) + 1)^2}{\alpha(\xi)}+\frac{(\eta_1 - 1)^2}{\eta_1}\left(\frac{2\Pi(l_1^2,k)}{K(k(\xi))}-1\right)+\frac{(\eta_1 + 1)^2}{\eta_1}\left(\frac{E(k(\xi))}{K(k(\xi))}-1\right)\right)}{\alpha(\xi)+\alpha(\xi)^{-1}+2 - 4\frac{\Pi(l_1^2,k(\xi))}{K(k(\xi))}},\\
    &\xi=\frac{-\alpha(\xi)^2-\frac{1}{\alpha(\xi)^2}+\frac{(\alpha(\xi) +\frac{1}{\alpha(\xi)}+\eta_1 +\frac{1}{\eta_1})}{2}(\alpha(\xi) +\frac{1}{\alpha(\xi)}+2-\frac{4\Pi(l_1^2,k(\xi))}{K(k(\xi))})+\frac{\int_{\eta_1}^{\eta_1^{-1}}\frac{s^2+s^{-2}}{R(s)}	\ddd  s}{\int_{\eta_1}^{\eta_1^{-1}}\frac{1}{R(s)}	\ddd  s}}{\alpha(\xi)+\alpha(\xi)^{-1}+2 - 4\frac{\Pi(l_1^2,k(\xi))}{K(k(\xi))}}
\end{align}
with
\begin{align}
    &k(\xi)=\dfrac{l_1(\alpha(\xi)+1)}{\alpha(\xi)-1}.\label{equ xialpha2}
\end{align}
In Appendix \ref{app5}, we show that the equation \eqref{equ xialpha} admits a unique solution $\alpha(\xi)\in(\eta_1,\eta_2)$ for each $\xi\in (\xi_{\mathrm{crit}},-\frac{\eta_1-\eta_1^{-1}}{\ln\eta_1})$, where $\xi_{\mathrm{crit}}$ is defined in \eqref{def xi}.
% Consequently, for each $\xi\in (\xi_{\mathrm{crit}},-\frac{\eta_1-\eta_1^{-1}}{\ln\eta_1})$, there exists a unique $\alpha(\xi)\in(\eta_1,\eta_2)$ as a solution of . 
A direct calculation shows that
    \begin{align}\label{sym R}
     R(\lb)=-\lb^2\overline{R(\bar{\lb}^{-1})}.
    \end{align}

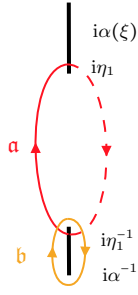
\begin{figure}[h]

\centering
\tikzset{every picture/.style={line width=0.75pt}} %set default line width to 0.75pt

\begin{tikzpicture}[x=0.75pt,y=0.75pt,yscale=-0.6,xscale=0.6]
	%uncomment if require: \path (0,530); %set diagram left start at 0, and has height of 530	
	%Straight Lines [id:da20069539428134475]
	\draw [line width=1.5]    (333.4,15.78) -- (333.58,75.35) ;
	%Straight Lines [id:da7511472562327318]
	\draw [line width=1.5]    (333.4,203.71) -- (333.64,244.47) ;
	%Curve Lines [id:da6927157975055704]
	\draw [color={rgb, 255:red, 254; green, 19; blue, 34 }  ,draw opacity=1 ]   (333.22,209.96) .. controls (299.88,205.96) and (293.88,72.63) .. (331.88,67.63) ;
	\draw [shift={(305.73,132.97)}, rotate = 89.37] [fill={rgb, 255:red, 254; green, 19; blue, 34 }  ,fill opacity=1 ][line width=0.08]  [draw opacity=0] (8.93,-4.29) -- (0,0) -- (8.93,4.29) -- cycle    ;
	%Curve Lines [id:da8333773003600226]
	\draw [color={rgb, 255:red, 254; green, 19; blue, 34 }  ,draw opacity=1 ] [dash pattern={on 4.5pt off 4.5pt}]  (331.88,67.63) .. controls (374.88,67.96) and (376.22,216.96) .. (333.22,209.96) ;
	\draw [shift={(364.8,144.82)}, rotate = 269.64] [fill={rgb, 255:red, 254; green, 19; blue, 34 }  ,fill opacity=1 ][line width=0.08]  [draw opacity=0] (8.93,-4.29) -- (0,0) -- (8.93,4.29) -- cycle    ;
	%Curve Lines [id:da45061017229372224]
	\draw [color={rgb, 255:red, 245; green, 166; blue, 35 }  ,draw opacity=1 ]   (333.55,252.5) .. controls (313.22,253.84) and (317.22,197.38) .. (332.55,197.17) ;
	\draw [shift={(320,220.4)}, rotate = 93.16] [fill={rgb, 255:red, 245; green, 166; blue, 35 }  ,fill opacity=1 ][line width=0.08]  [draw opacity=0] (8.93,-4.29) -- (0,0) -- (8.93,4.29) -- cycle    ;
	%Curve Lines [id:da21736018277154978]
	\draw [color={rgb, 255:red, 245; green, 166; blue, 35 }  ,draw opacity=1 ]   (332.55,197.17) .. controls (353.22,196.17) and (350.22,253.17) .. (333.55,252.5) ;
	\draw [shift={(346.8,229.33)}, rotate = 272.53] [fill={rgb, 255:red, 245; green, 166; blue, 35 }  ,fill opacity=1 ][line width=0.08]  [draw opacity=0] (8.93,-4.29) -- (0,0) -- (8.93,4.29) -- cycle    ;
	
	% Text Node
	\draw (359.42,233.2) node [anchor=north west][inner sep=0.75pt]  [font=\footnotesize,xscale=0.8,yscale=0.8]  {$\ii   \alpha^{-1}$};
	% Text Node
	\draw (350,62.02) node [anchor=north west][inner sep=0.75pt]  [font=\footnotesize,xscale=0.8,yscale=0.8]  {$\ii  \eta _{1}$};
	% Text Node
	\draw (351.47,30.82) node [anchor=north west][inner sep=0.75pt]  [font=\footnotesize,xscale=0.8,yscale=0.8]  {$\ii  \alpha (\xi)$};
	% Text Node
	\draw (358.02,204.13) node [anchor=north west][inner sep=0.75pt]  [font=\footnotesize,xscale=0.8,yscale=0.8]  {$\ii  \eta _{1}^{-1}$};
	% Text Node
	\draw (280.33,130.07) node [anchor=north west][inner sep=0.75pt]  [color={rgb, 255:red, 254; green, 19; blue, 34 }  ,opacity=1 ,xscale=0.8,yscale=0.8]  {$\mathfrak{a}$};
	% Text Node
	\draw (286.67,219.07) node [anchor=north west][inner sep=0.75pt]  [color={rgb, 255:red, 245; green, 166; blue, 35}  ,opacity=1 ,xscale=0.8,yscale=0.8]  {$\mathfrak{b}$};		
\end{tikzpicture}
\caption{\footnotesize  The basis of the  homology basis elements $\mathfrak{a}$ (red) and  $\mathfrak{b}$ (orange) of the Riemann surface $\mathcal{R}$ associted with \eqref{def R}. The dashed curves are on the lower sheet of $\mathcal{R}$. } \label{a&b}
	\end{figure}

%	The $\lb$-derivative of $\phi$ given in \eqref{def phi} reads
%	\begin{align}\label{def dphi}
%	\phi_\lb(\lb)=	\frac{n+1}{\lb}-\i t\left( 1-\frac{1}{\lb^2}\right) ,
%	\end{align}
%and it inspires us to  define
We next define 
	\begin{align}
	g(\lb)=g(\lb;\xi):=\frac{1}{2}\int_{\i \alpha(\xi)}^{\lb}\frac{\ii(s^2+s^{-2})+c_1(\xi)(s-s^{-1}) +c_0(\xi)}{R(s)}	\ddd  s,\label{dg}
	\end{align}
	%which is  a holomorphic function on $\mathcal{R}$, and 
where 
\begin{equation}\label{c1c0}
    \begin{aligned}
c_1(\xi)&:=\xi+\frac{\eta_{1}+\eta_{1}^{-1}+\alpha(\xi)+\alpha(\xi)^{-1}}{2}\in\mathbb{R},\qquad
    \\
	c_0(\xi) &:=-\ii\frac{\int_{\eta_1}^{\eta_1^{-1}}\frac{-(s^2+s^{-2})+c_1(\xi)(s+s^{-1}) }{R(s)}	\ddd  s}{\int_{\eta_1}^{\eta_1^{-1}}\frac{1}{R(s)}	\ddd  s}\in \ii\mathbb{R}.
	\end{aligned}
\end{equation}

Then, $g(\lb)$  admits the following proposition, whose proof  is similar to that of Proposition~\ref{png} and is therefore omitted.
%The proof  is similar to that of Proposition~\ref{png} and is therefore omitted.
\begin{prop} \label{propg2}
The function $g(\lb)$ defined in \eqref{dg} satisfies the following properties.
\begin{itemize}
        \item[\rm{(a)}]	Let $\ddd g$ be the differential of $g(\lb)$, it then follows that 
        \begin{align}\label{g a period}
		\oint_{  \mathfrak{a}}\ddd  g=0.
	\end{align}
    \item[\rm{(b)}] $ g(\lb)-\frac{\phi(\lb)}{2t}$ is analytic in $\C\setminus \i (\alpha(\xi)^{-1},\alpha(\xi)) $ and  satisfies the following conditions:
\begin{align}
	g_+(\lb)=
	\begin{cases}
		-g_-(\lb),&\lb\in\ii(\eta_{1},\alpha(\xi))\cup\ii(\alpha(\xi)^{-1},\eta_{1}^{-1}),\\
		g_-(\lb)-\ii\Omega,&\lb\in\ii(\eta_1^{-1},\eta_1),
	\end{cases}
\end{align}  
%\begin{align}
%	\begin{aligned}
%		& g_+(\lb) + g_-(\lb) = 0, & \lb \in \ii(\eta_{1},\alpha(\xi))\cup\ii(\alpha(\xi)^{-1},\eta_{1}^{-1}),  \\
%	& g_+ (\lb) - g_-(\lb) = -\i\Omega, & \lb \in \ii(\eta_{1}^{-1},\eta_{1}),
%	\end{aligned} \label{gconstraint}
%\end{align}
where 
\begin{align}
	\Omega=\i
		\oint_{ \mathfrak{b}}\ddd  g \in \R.\label{def Omegj}
\end{align}
\item[\rm{(c)}] $g(\lb)$ admits the symmetry relation
\begin{align}\label{sym g}
    g(\lb)=-\overline{g(\bar{\lb}^{-1})}.
\end{align}
\item[\rm{(d)}]  
As $\lb \rightarrow  \infty$, we have
\begin{align}
	 2tg(\lb)-\phi(\lb) =\oo(1).
	\end{align}
    
%$g(\lb)$ admits the asymptotic properties:
%\begin{align}
%&2tg(\lb)-\phi(\lb) =\oo(1), \qquad \lambda\to \infty,\label{e5.10}\\
%	&g(\lambda)=\frac{(\eta_2+\eta_2^{-1}-\kappa_1)(-\ii\lb-\eta_2)^{1/2}}{\sqrt{(\eta_2-\eta_1)(\eta_2-\eta_1^{-1})(\eta_2-\eta_2^{-1})}}(1+\oo(\lambda-\ii\eta_2)), \qquad \lambda\to\ii\eta_2, \label{asy g2 eta2}
%\end{align}

\end{itemize}

\end{prop}
\begin{figure}
	\centering
	\subfigure[]{		
		\tikzset{every picture/.style={line width=0.75pt}} %set default line width to 0.75pt        		
		\begin{tikzpicture}[x=0.75pt,y=0.75pt,yscale=-0.58,xscale=0.58]
			%uncomment if require: \path (0,530); %set diagram left start at 0, and has height of 530
%Shape: Polygon Curved [id:ds4117880296419856] 
\draw  [color={rgb, 255:red, 0; green, 0; blue, 0 }  ,draw opacity=0 ][fill={rgb, 255:red, 184; green, 233; blue, 134 }  ,fill opacity=0.3 ] (422.77,106.67) .. controls (384.77,140.67) and (341.51,153.86) .. (318.5,154.02) .. controls (295.49,154.18) and (252.77,137.34) .. (212.1,104.67) .. controls (171.43,72) and (137.77,49.5) .. (114.77,44) .. controls (91.77,38.5) and (290.02,37.99) .. (318.45,38.04) .. controls (346.89,38.09) and (506.77,38.67) .. (510.1,37.34) .. controls (513.43,36) and (460.77,72.67) .. (422.77,106.67) -- cycle ;
%Shape: Circle [id:dp5693389500073517] 
\draw  [color={rgb, 255:red, 74; green, 144; blue, 226 }  ,draw opacity=1 ][fill={rgb, 255:red, 184; green, 233; blue, 134 }  ,fill opacity=0.3 ] (186.49,307.99) .. controls (186.49,233.37) and (246.98,172.89) .. (321.6,172.89) .. controls (396.21,172.89) and (456.7,233.37) .. (456.7,307.99) .. controls (456.7,382.61) and (396.21,443.1) .. (321.6,443.1) .. controls (246.98,443.1) and (186.49,382.61) .. (186.49,307.99) -- cycle ;
%Straight Lines [id:da18159634992021056] 
\draw [color={rgb, 255:red, 74; green, 144; blue, 226 }  ,draw opacity=0 ][fill={rgb, 255:red, 74; green, 144; blue, 226 }  ,fill opacity=1 ][line width=0.75]  [dash pattern={on 0.84pt off 2.51pt}]  (318.5,61.86) -- (318.5,154.02) ;
%Straight Lines [id:da26113900821192315] 
\draw [color={rgb, 255:red, 0; green, 0; blue, 0 }  ,draw opacity=1 ] [dash pattern={on 0.84pt off 2.51pt}]  (117,310.28) -- (533.5,307.8) ;
\draw [shift={(536.5,307.78)}, rotate = 179.66] [fill={rgb, 255:red, 0; green, 0; blue, 0 }  ,fill opacity=1 ][line width=0.08]  [draw opacity=0] (8.93,-4.29) -- (0,0) -- (8.93,4.29) -- cycle    ;
%Straight Lines [id:da5042170200967123] 
\draw [color={rgb, 255:red, 255; green, 255; blue, 255 }  ,draw opacity=0 ][fill={rgb, 255:red, 74; green, 144; blue, 226 }  ,fill opacity=0 ][line width=1.5]    (318.5,183.86) -- (318.5,241.86) ;
%Curve Lines [id:da0457622746347256] 
\draw [color={rgb, 255:red, 74; green, 144; blue, 226 }  ,draw opacity=1 ]   (114.77,44) .. controls (183.27,62.5) and (237.29,152.21) .. (318.5,154.02) ;
%Curve Lines [id:da6736321485044326] 
\draw [color={rgb, 255:red, 74; green, 144; blue, 226 }  ,draw opacity=1 ]   (510.1,37.34) .. controls (460.1,60) and (399.71,152.21) .. (318.5,154.02) ;
%Curve Lines [id:da6226795203061964] 
\draw [color={rgb, 255:red, 74; green, 144; blue, 226 }  ,draw opacity=1 ][fill={rgb, 255:red, 255; green, 255; blue, 255 }  ,fill opacity=1 ]   (318.5,183.86) .. controls (375.01,187.57) and (364.18,307.61) .. (319,309.39) ;
%Curve Lines [id:da06582762230718975] 
\draw [color={rgb, 255:red, 74; green, 144; blue, 226 }  ,draw opacity=1 ][fill={rgb, 255:red, 255; green, 255; blue, 255 }  ,fill opacity=1 ]   (318.5,183.86) .. controls (261.99,187.57) and (278.18,309.61) .. (319,309.39) ;
%Straight Lines [id:da6615319433178072] 
\draw  [dash pattern={on 0.84pt off 2.51pt}]  (318.75,450.15) -- (318.34,31.68) ;
\draw [shift={(318.33,28.68)}, rotate = 89.94] [fill={rgb, 255:red, 0; green, 0; blue, 0 }  ][line width=0.08]  [draw opacity=0] (8.93,-4.29) -- (0,0) -- (8.93,4.29) -- cycle    ;
\draw [shift={(318.5,241.86) }, rotate = 90] [color={rgb, 255:red, 0; green, 0; blue, 0}  ,draw opacity=1 ][fill={rgb, 255:red, 0; green, 0; blue, 0}  ,fill opacity=1 ][line width=1]      (0, 0) circle [x radius= 1.74, y radius= 1.74]   ;
\draw [shift={(318.5,183.86)}, rotate = 90] [color={rgb, 255:red, 0; green, 0; blue, 0}  ,draw opacity=1 ][fill={rgb, 255:red, 0; green, 0; blue, 0}  ,fill opacity=1 ][line width=1]      (0, 0) circle [x radius= 1.74, y radius= 1.74]   ;
\draw [shift={(318.5,61.86) }, rotate = 90] [color={rgb, 255:red, 0; green, 0; blue, 0}  ,draw opacity=1 ][fill={rgb, 255:red, 0; green, 0; blue, 0}  ,fill opacity=1 ][line width=1]      (0, 0) circle [x radius= 1.74, y radius= 1.74]   ;
\draw [shift={(318.5,154.02)}, rotate = 90] [color={rgb, 255:red, 0; green, 0; blue, 0}  ,draw opacity=1 ][fill={rgb, 255:red, 0; green, 0; blue, 0}  ,fill opacity=1 ][line width=1]      (0, 0) circle [x radius= 1.74, y radius= 1.74]   ;
% Text Node
\draw (322,180) node [anchor=north west][inner sep=0.75pt]  [font=\footnotesize,xscale=0.8,yscale=0.8] {$\mathrm{i} \eta _{1}^{-1}$};
% Text Node
\draw (322,55) node [anchor=north west][inner sep=0.75pt] [font=\footnotesize,xscale=0.8,yscale=0.8] {$\mathrm{i} \eta _{2}$};
% Text Node
\draw (322,232) node [anchor=north west][inner sep=0.75pt][font=\footnotesize,xscale=0.8,yscale=0.8] {$\mathrm{i} \eta _{2}^{-1}$};
% Text Node
\draw (322,132) node [anchor=north west][inner sep=0.75pt]  [font=\footnotesize,xscale=0.8,yscale=0.8] {$\mathrm{i} \eta _{1}$};
% Text Node
\draw (265.37,54.41) node [anchor=north west][inner sep=0.75pt] [font=\footnotesize,xscale=0.8,yscale=0.8] {$+$};
% Text Node
\draw (390.53,249.24) node [anchor=north west][inner sep=0.75pt]    {$+$};
	\end{tikzpicture}}
\subfigure[]{
	
	\tikzset{every picture/.style={line width=0.75pt}} %set default line width to 0.75pt
	
	\begin{tikzpicture}[x=0.75pt,y=0.75pt,yscale=-0.58,xscale=0.58]

%Shape: Polygon Curved [id:ds5380416869389384] 
\draw  [color={rgb, 255:red, 0; green, 0; blue, 0 }  ,draw opacity=0 ][fill={rgb, 255:red, 184; green, 233; blue, 134 }  ,fill opacity=0.3 ] (457.5,82.27) .. controls (407,101.77) and (380.21,100.26) .. (341.5,100.27) .. controls (302.79,100.27) and (280,103.27) .. (228.5,82.77) .. controls (177,62.27) and (170.5,43.57) .. (146.5,42.77) .. controls (122.5,41.97) and (312.06,39.05) .. (340.5,39.11) .. controls (368.94,39.16) and (557,41.77) .. (539,42.27) .. controls (521,42.77) and (508,62.77) .. (457.5,82.27) -- cycle ;
%Shape: Circle [id:dp11356733424717913] 
\draw  [color={rgb, 255:red, 74; green, 144; blue, 226 }  ,draw opacity=1 ][fill={rgb, 255:red, 184; green, 233; blue, 134 }  ,fill opacity=0.3 ] (209.49,288.74) .. controls (209.49,214.12) and (269.98,153.63) .. (344.6,153.63) .. controls (419.21,153.63) and (479.7,214.12) .. (479.7,288.74) .. controls (479.7,363.35) and (419.21,423.84) .. (344.6,423.84) .. controls (269.98,423.84) and (209.49,363.35) .. (209.49,288.74) -- cycle ;
%Curve Lines [id:da8146597464263361] 
\draw [color={rgb, 255:red, 74; green, 144; blue, 226 }  ,draw opacity=1 ][fill={rgb, 255:red, 255; green, 255; blue, 255 }  ,fill opacity=1 ]   (341,188.77) .. controls (397.51,192.48) and (387.18,288.36) .. (342,290.14) ;
%Curve Lines [id:da78429901848814] 
\draw [color={rgb, 255:red, 74; green, 144; blue, 226 }  ,draw opacity=1 ][fill={rgb, 255:red, 255; green, 255; blue, 255 }  ,fill opacity=1 ]   (341,188.77) .. controls (284.49,192.48) and (301.18,290.36) .. (342,290.14) ;
%Straight Lines [id:da32702902515751475] 
\draw  [dash pattern={on 0.84pt off 2.51pt}]  (341.5,436.11) -- (341.5,15.61) ;
\draw [shift={(341.5,12.61)}, rotate = 90] [fill={rgb, 255:red, 0; green, 0; blue, 0 }  ][line width=0.08]  [draw opacity=0] (8.93,-4.29) -- (0,0) -- (8.93,4.29) -- cycle    ;
%Straight Lines [id:da8826944206520135] 
%\draw [line width=1.5]    (341.5,42.61) -- (341.5,134.77) ;
%Straight Lines [id:da5671982930286904] 
\draw [color={rgb, 255:red, 0; green, 0; blue, 0 }  ,draw opacity=1 ] [dash pattern={on 0.84pt off 2.51pt}]  (129,289.27) -- (567.6,288.75) ;
\draw [shift={(570.6,288.74)}, rotate = 179.93] [fill={rgb, 255:red, 0; green, 0; blue, 0 }  ,fill opacity=1 ][line width=0.08]  [draw opacity=0] (8.93,-4.29) -- (0,0) -- (8.93,4.29) -- cycle    ;
%Straight Lines [id:da7486990463825082] 
%\draw [line width=1.5]    (341.5,164.61) -- (341.5,222.61) ;
%Curve Lines [id:da9351088711883422] 
\draw [color={rgb, 255:red, 74; green, 144; blue, 226 }  ,draw opacity=1 ]   (146.5,42.77) .. controls (215.5,62.77) and (204.79,98.27) .. (341.5,100.27) ;
%Curve Lines [id:da6400643526915221] 
\draw [color={rgb, 255:red, 74; green, 144; blue, 226 }  ,draw opacity=1 ]   (539,42.27) .. controls (474,63.27) and (478.21,98.27) .. (341.5,100.27) ;
%Straight Lines [id:da1564671730226621] 
\draw [color={rgb, 255:red, 74; green, 144; blue, 226 }  ,draw opacity=1 ][line width=1.5]    (341.5,100.27) -- (341.5,134.77) ;
%Straight Lines [id:da4726106857326563] 
\draw [color={rgb, 255:red, 74; green, 144; blue, 226 }  ,draw opacity=1 ][line width=1.5]    (341.5,164.61) -- (341.5,188.77) ;
\draw [shift={(341.5,164.61) }, rotate = 90] [color={rgb, 255:red, 74; green, 144; blue, 226 }  ,draw opacity=1 ][fill={rgb, 255:red, 74; green, 144; blue, 226 }  ,fill opacity=1 ][line width=1]      (0, 0) circle [x radius= 1.74, y radius= 1.74]   ;
\draw [shift={(341.5,188.77)}, rotate = 90] [color={rgb, 255:red, 74; green, 144; blue, 226 }  ,draw opacity=1 ][fill={rgb, 255:red, 74; green, 144; blue, 226 }  ,fill opacity=1 ][line width=1]      (0, 0) circle [x radius= 1.74, y radius= 1.74]   ;
\draw [shift={ (341.5,100.27) }, rotate = 90] [color={rgb, 255:red, 74; green, 144; blue, 226 }  ,draw opacity=1 ][fill={rgb, 255:red, 74; green, 144; blue, 226 }  ,fill opacity=1 ][line width=1]      (0, 0) circle [x radius= 1.74, y radius= 1.74]   ;
\draw [shift={(341.5,134.77)}, rotate = 90] [color={rgb, 255:red, 74; green, 144; blue, 226 }  ,draw opacity=1 ][fill={rgb, 255:red, 74; green, 144; blue, 226 }  ,fill opacity=1 ][line width=1]      (0, 0) circle [x radius= 1.74, y radius= 1.74]   ;
\draw [shift={(341.5,222.61) }, rotate = 90] [color={rgb, 255:red, 0; green, 0; blue, 0}  ,draw opacity=1 ][fill={rgb, 255:red, 0; green, 0; blue, 0}  ,fill opacity=1 ][line width=1]      (0, 0) circle [x radius= 1.74, y radius= 1.74]   ;
\draw [shift={ (341.5,42.61) }, rotate = 90] [color={rgb, 255:red, 0; green, 0; blue, 0}  ,draw opacity=1 ][fill={rgb, 255:red, 0; green, 0; blue, 0}  ,fill opacity=1 ][line width=1]      (0, 0) circle [x radius= 1.74, y radius= 1.74]   ;
% Text Node
\draw (345.92,209.21) node [anchor=north west][inner sep=0.75pt]  [font=\footnotesize,xscale=0.8,yscale=0.8]  {$\mathrm{i} \eta _{2}^{-1}$};
% Text Node
\draw (344.8,125.67) node [anchor=north west][inner sep=0.75pt]  [font=\footnotesize,xscale=0.8,yscale=0.8]  {$\mathrm{i} \eta _{1}$};
% Text Node
\draw (345.02,161.15) node [anchor=north west][inner sep=0.75pt]  [font=\footnotesize,xscale=0.8,yscale=0.8]  {$\mathrm{i} \eta _{1}^{-1}$};
% Text Node
\draw (344.7,40.4) node [anchor=north west][inner sep=0.75pt]  [font=\footnotesize,xscale=0.8,yscale=0.8]  {$\mathrm{i} \eta _{2}$};
% Text Node
\draw (419.53,59.99) node [anchor=north west][inner sep=0.75pt]  [font=\footnotesize,xscale=0.8,yscale=0.8]   {$+$};
% Text Node
\draw (413.53,229.99) node [anchor=north west][inner sep=0.75pt]   [font=\footnotesize,xscale=0.8,yscale=0.8]  {$+$};
% Text Node
\draw (344,77.82) node [anchor=north west][inner sep=0.75pt]  [font=\footnotesize,xscale=0.8,yscale=0.8]  {$\mathrm{i} \alpha ( \xi )$};
% Text Node
\draw (345.5,184.32) node [anchor=north west][inner sep=0.75pt]  [font=\footnotesize,xscale=0.8,yscale=0.8]  {$\mathrm{i} \alpha ( \xi )^{-1}$};
`\end{tikzpicture}}
\subfigure[]{\tikzset{every picture/.style={line width=0.75pt}} %set default line width to 0.75pt        

\begin{tikzpicture}[x=0.75pt,y=0.75pt,yscale=-0.58,xscale=0.58]
%uncomment if require: \path (0,530); %set diagram left start at 0, and has height of 530
%Shape: Polygon Curved [id:ds9539047484717239] 
\draw  [color={rgb, 255:red, 0; green, 0; blue, 0 }  ,draw opacity=0 ][fill={rgb, 255:red, 184; green, 233; blue, 134 }  ,fill opacity=0.3 ] (439.85,61.2) .. controls (379.05,87.2) and (337.08,91.28) .. (314.07,91.44) .. controls (291.05,91.6) and (255.45,88.4) .. (198.25,65.2) .. controls (141.05,42) and (158.35,47.63) .. (135.35,42.13) .. controls (112.35,36.63) and (284.92,38.08) .. (313.35,38.13) .. controls (341.79,38.18) and (487.83,41.73) .. (499.85,38.13) .. controls (511.88,34.53) and (500.65,35.2) .. (439.85,61.2) -- cycle ;
%Shape: Circle [id:dp6790323290259638] 
\draw  [color={rgb, 255:red, 74; green, 144; blue, 226 }  ,draw opacity=1 ][fill={rgb, 255:red, 184; green, 233; blue, 134 }  ,fill opacity=0.3 ] (182.49,300.74) .. controls (182.49,226.12) and (242.98,165.63) .. (317.6,165.63) .. controls (392.21,165.63) and (452.7,226.12) .. (452.7,300.74) .. controls (452.7,375.35) and (392.21,435.84) .. (317.6,435.84) .. controls (242.98,435.84) and (182.49,375.35) .. (182.49,300.74) -- cycle ;
%Curve Lines [id:da47327686883825903] 
\draw [color={rgb, 255:red, 74; green, 144; blue, 226 }  ,draw opacity=1 ][fill={rgb, 255:red, 255; green, 255; blue, 255 }  ,fill opacity=1 ]   (313.33,214.61) .. controls (369.84,218.32) and (360.18,300.36) .. (315,302.14) ;
%Curve Lines [id:da7598178491438858] 
\draw [color={rgb, 255:red, 74; green, 144; blue, 226 }  ,draw opacity=1 ][fill={rgb, 255:red, 255; green, 255; blue, 255 }  ,fill opacity=1 ]   (313.33,214.61) .. controls (256.82,218.32) and (274.18,302.36) .. (315,302.14) ;
\draw [shift={(314.25,214.61)}, rotate = 176.24] [color={rgb, 255:red, 74; green, 144; blue, 226 }  ,draw opacity=1 ][fill={rgb, 255:red, 74; green, 144; blue, 226 }  ,fill opacity=1 ][line width=1]      (0, 0) circle [x radius= 2.01, y radius= 2.01]   ;
%Straight Lines [id:da9673787709575596] 
\draw  [dash pattern={on 0.84pt off 2.51pt}]  (314.25,442.61) -- (314.25,29.03) ;
\draw [shift={(314,25.03)}, rotate = 89.93] [fill={rgb, 255:red, 0; green, 0; blue, 0 }  ][line width=0.08]  [draw opacity=0] (8.93,-4.29) -- (0,0) -- (8.93,4.29) -- cycle    ;
%Straight Lines [id:da9224560079224464] 
\draw [color={rgb, 255:red, 74; green, 144; blue, 226 }  ,draw opacity=1 ][fill={rgb, 255:red, 74; green, 144; blue, 226 }  ,fill opacity=1 ][line width=1.5]    (314.25,54.61) -- (314.25,146.77) ;
\draw [shift={(314.25,146.77)}, rotate = 90] [color={rgb, 255:red, 74; green, 144; blue, 226 }  ,draw opacity=1 ][fill={rgb, 255:red, 74; green, 144; blue, 226 }  ,fill opacity=1 ][line width=1.5]      (0, 0) circle [x radius= 1.74, y radius= 1.74]   ;
\draw [shift={(314.25,54.61)}, rotate = 90] [color={rgb, 255:red, 74; green, 144; blue, 226 }  ,draw opacity=1 ][fill={rgb, 255:red, 74; green, 144; blue, 226 }  ,fill opacity=1 ][line width=1.5]      (0, 0) circle [x radius= 1.74, y radius= 1.74]   ;
%Straight Lines [id:da20992960698969332] 
\draw [color={rgb, 255:red, 0; green, 0; blue, 0 }  ,draw opacity=1 ] [dash pattern={on 0.84pt off 2.51pt}]  (113,303.03) -- (529.5,300.55) ;
\draw [shift={(532.5,300.53)}, rotate = 179.66] [fill={rgb, 255:red, 0; green, 0; blue, 0 }  ,fill opacity=1 ][line width=0.08]  [draw opacity=0] (8.93,-4.29) -- (0,0) -- (8.93,4.29) -- cycle    ;
%Straight Lines [id:da1396514151051409] 
\draw [color={rgb, 255:red, 74; green, 144; blue, 226 }  ,draw opacity=1 ][line width=1.5]    (314.25,176.32) -- (314.25,234.32) ;
\draw [shift={(314.25,234.32)}, rotate = 90] [color={rgb, 255:red, 74; green, 144; blue, 226 }  ,draw opacity=1 ][fill={rgb, 255:red, 74; green, 144; blue, 226 }  ,fill opacity=1 ][line width=1.5]      (0, 0) circle [x radius= 1.74, y radius= 1.74]   ;
\draw [shift={(314.25,176.32)}, rotate = 90] [color={rgb, 255:red, 74; green, 144; blue, 226 }  ,draw opacity=1 ][fill={rgb, 255:red, 74; green, 144; blue, 226 }  ,fill opacity=1 ][line width=1.5]      (0, 0) circle [x radius= 1.74, y radius= 1.74]   ;
%Curve Lines [id:da713975623589845] 
\draw [color={rgb, 255:red, 74; green, 144; blue, 226 }  ,draw opacity=1 ]   (135.35,42.13) .. controls (203.85,60.63) and (232.85,89.63) .. (314.07,91.44) ;
\draw [shift={(314.25,91.44)}, rotate = 1.28] [color={rgb, 255:red, 74; green, 144; blue, 226 }  ,draw opacity=1 ][fill={rgb, 255:red, 74; green, 144; blue, 226 }  ,fill opacity=1 ][line width=1]      (0, 0) circle [x radius= 2.01, y radius= 2.01]   ;
%Shape: Boxed Bezier Curve [id:dp8377138690321616] 
\draw [color={rgb, 255:red, 74; green, 144; blue, 226 }  ,draw opacity=1 ]   (492.78,42.13) .. controls (424.28,60.63) and (395.28,89.63) .. (314.07,91.44) ;
% Text Node
\draw (320,162) node [anchor=north west][inner sep=0.75pt]  [font=\footnotesize,xscale=0.8,yscale=0.8]   {$\mathrm{i} \eta _{1}^{-1}$};
% Text Node
\draw (320.3,39) node [anchor=north west][inner sep=0.75pt]  [font=\footnotesize,xscale=0.8,yscale=0.8]   {$\mathrm{i} \eta _{2}$};
% Text Node
\draw (320,226.01) node [anchor=north west][inner sep=0.75pt]  [font=\footnotesize,xscale=0.8,yscale=0.8]   {$\mathrm{i} \eta _{2}^{-1}$};
% Text Node
\draw (320,122.01) node [anchor=north west][inner sep=0.75pt]  [font=\footnotesize,xscale=0.8,yscale=0.8]   {$\mathrm{i} \eta _{1}$};
% Text Node
\draw (288.37,35.16) node [anchor=north west][inner sep=0.75pt]  [font=\footnotesize,xscale=0.8,yscale=0.8]   {$+$};
% Text Node
\draw (413.53,229.99) node [anchor=north west][inner sep=0.75pt]  [font=\footnotesize,xscale=0.8,yscale=0.8]   {$+$};
% Text Node
\draw (320,75) node [anchor=north west][inner sep=0.75pt]  [font=\footnotesize,xscale=0.8,yscale=0.8]   {$\mathrm{i} \lambda _{0}(\xi)$};
% Text Node
\draw (320,190) node [anchor=north west][inner sep=0.75pt]  [font=\footnotesize,xscale=0.8,yscale=0.8]   {$\mathrm{i} \lambda _{0}(\xi)^{-1}$};
\end{tikzpicture}
}
\subfigure[]{
	
	\tikzset{every picture/.style={line width=0.75pt}} %set default line width to 0.75pt
	
	\begin{tikzpicture}[x=0.75pt,y=0.75pt,yscale=-0.58,xscale=0.58]
		%uncomment if require: \path (0,431); %set diagram left start at 0, and has height of 431

%Shape: Triangle [id:dp029180997821849353] 
\draw  [color={rgb, 255:red, 0; green, 0; blue, 0 }  ,draw opacity=0 ][fill={rgb, 255:red, 184; green, 233; blue, 134 }  ,fill opacity=0.3 ] (341.05,42.61) -- (497.5,28.53) -- (189,28.53) -- cycle ;
%Shape: Circle [id:dp3566890221121126] 
\draw  [color={rgb, 255:red, 74; green, 144; blue, 226 }  ,draw opacity=1 ][fill={rgb, 255:red, 184; green, 233; blue, 134 }  ,fill opacity=0.3 ] (209.49,288.74) .. controls (209.49,214.12) and (269.98,153.63) .. (344.6,153.63) .. controls (419.21,153.63) and (479.7,214.12) .. (479.7,288.74) .. controls (479.7,363.35) and (419.21,423.84) .. (344.6,423.84) .. controls (269.98,423.84) and (209.49,363.35) .. (209.49,288.74) -- cycle ;
%Curve Lines [id:da7894264196687798] 
\draw [color={rgb, 255:red, 74; green, 144; blue, 226 }  ,draw opacity=1 ][fill={rgb, 255:red, 255; green, 255; blue, 255 }  ,fill opacity=1 ]   (341.5,222) .. controls (398.01,226.31) and (387.18,288.36) .. (342,290.14) ;
%Curve Lines [id:da7478330350306062] 
\draw [color={rgb, 255:red, 74; green, 144; blue, 226 }  ,draw opacity=1 ][fill={rgb, 255:red, 255; green, 255; blue, 255 }  ,fill opacity=1 ]   (341.5,222) .. controls (284.99,226.31) and (301.18,290.36) .. (342,290.14) ;
%Straight Lines [id:da7124188854487186] 
\draw  [dash pattern={on 0.84pt off 2.51pt}]  (341.5,436.11) -- (341,17) ;
\draw [shift={(341,14)}, rotate = 89.93] [fill={rgb, 255:red, 0; green, 0; blue, 0 }  ][line width=0.08]  [draw opacity=0] (8.93,-4.29) -- (0,0) -- (8.93,4.29) -- cycle    ;
%Straight Lines [id:da7209329240138381] 
\draw [color={rgb, 255:red, 74; green, 144; blue, 226 }  ,draw opacity=1 ][fill={rgb, 255:red, 74; green, 144; blue, 226 }  ,fill opacity=1 ][line width=1.5]    (341.5,42.61) -- (341.5,134.77) ;
\draw [shift={(341.5,42.61)}, rotate = 90] [color={rgb, 255:red, 74; green, 144; blue, 226 }  ,draw opacity=1 ][fill={rgb, 255:red, 74; green, 144; blue, 226 }  ,fill opacity=1 ][line width=1.5]      (0, 0) circle [x radius= 1.74, y radius= 1.74]   ;
\draw [shift={(341.5,134.77) }, rotate = 90] [color={rgb, 255:red, 74; green, 144; blue, 226 }  ,draw opacity=1 ][fill={rgb, 255:red, 74; green, 144; blue, 226 }  ,fill opacity=1 ][line width=1.5]      (0, 0) circle [x radius= 1.74, y radius= 1.74]   ;
%Straight Lines [id:da8115689461936187] 
\draw [color={rgb, 255:red, 0; green, 0; blue, 0 }  ,draw opacity=1 ] [dash pattern={on 0.84pt off 2.51pt}]  (140,291.03) -- (556.5,288.55) ;
\draw [shift={(559.5,288.53)}, rotate = 179.66] [fill={rgb, 255:red, 0; green, 0; blue, 0 }  ,fill opacity=1 ][line width=0.08]  [draw opacity=0] (8.93,-4.29) -- (0,0) -- (8.93,4.29) -- cycle    ;
%Straight Lines [id:da9143949287682055] 
\draw [color={rgb, 255:red, 74; green, 144; blue, 226 }  ,draw opacity=1 ][line width=1.5]    (341.5,164.61) -- (341.5,222.61) ;
\draw [shift={(341.5,164.61)}, rotate = 90] [color={rgb, 255:red, 74; green, 144; blue, 226 }  ,draw opacity=1 ][fill={rgb, 255:red, 74; green, 144; blue, 226 }  ,fill opacity=1 ][line width=1.5]      (0, 0) circle [x radius= 1.74, y radius= 1.74]   ;
\draw [shift={(341.5,222.61)}, rotate = 90] [color={rgb, 255:red, 74; green, 144; blue, 226 }  ,draw opacity=1 ][fill={rgb, 255:red, 74; green, 144; blue, 226 }  ,fill opacity=1 ][line width=1.5]      (0, 0) circle [x radius= 1.74, y radius= 1.74]   ;
%Curve Lines [id:da5207616838297551] 
\draw [color={rgb, 255:red, 74; green, 144; blue, 226 }  ,draw opacity=1 ][fill={rgb, 255:red, 184; green, 233; blue, 134 }  ,fill opacity=0.3 ]   (189,29.03) .. controls (254.5,41.53) and (290.5,40.03) .. (341.5,42.61) ;
%Curve Lines [id:da11263960020857489] 
\draw [color={rgb, 255:red, 74; green, 144; blue, 226 }  ,draw opacity=1 ][fill={rgb, 255:red, 184; green, 233; blue, 134 }  ,fill opacity=0.3 ]   (497.5,28.53) .. controls (415.5,40.03) and (404.5,41.03) .. (341.5,42.61) ;

% Text Node
\draw (343.92,205.71) node [anchor=north west][inner sep=0.75pt]  [font=\footnotesize,xscale=0.8,yscale=0.8]  {$\mathrm{i} \eta _{2}^{-1}$};
% Text Node
\draw (344.3,119.67) node [anchor=north west][inner sep=0.75pt]  [font=\footnotesize,xscale=0.8,yscale=0.8]  {$\mathrm{i} \eta _{1}$};
% Text Node
\draw (344.52,160) node [anchor=north west][inner sep=0.75pt]  [font=\footnotesize,xscale=0.8,yscale=0.8]  {$\mathrm{i} \eta _{1}^{-1}$};
% Text Node
\draw (343.5,46.01) node [anchor=north west][inner sep=0.75pt]  [font=\footnotesize,xscale=0.8,yscale=0.8]  {$\mathrm{i} \eta _{2}$};
% Text Node
\draw (301.03,25.99) node [anchor=north west][inner sep=0.75pt]  [font=\footnotesize,xscale=0.8,yscale=0.8]  {$+$};
% Text Node
\draw (413.53,229.99) node [anchor=north west][inner sep=0.75pt]  [font=\footnotesize,xscale=0.8,yscale=0.8]   {$+$};
\end{tikzpicture}}
\subfigure[]{
	
	\tikzset{every picture/.style={line width=0.75pt}} %set default line width to 0.75pt
	
	\begin{tikzpicture}[x=0.75pt,y=0.75pt,yscale=-0.59,xscale=0.59]
		%uncomment if require: \path (0,431); %set diagram left start at 0, and has height of 431
%Shape: Triangle [id:dp029180997821849353] 
\draw  [color={rgb, 255:red, 0; green, 0; blue, 0 }  ,draw opacity=0 ][fill={rgb, 255:red, 184; green, 233; blue, 134 }  ,fill opacity=0.3 ] (341.05,37.11) -- (497.5,23.03) -- (189,23.03) -- cycle ;
%Shape: Circle [id:dp3566890221121126] 
\draw  [color={rgb, 255:red, 74; green, 144; blue, 226 }  ,draw opacity=1 ][fill={rgb, 255:red, 184; green, 233; blue, 134 }  ,fill opacity=0.3 ] (209.49,288.74) .. controls (209.49,214.12) and (269.98,153.63) .. (344.6,153.63) .. controls (419.21,153.63) and (479.7,214.12) .. (479.7,288.74) .. controls (479.7,363.35) and (419.21,423.84) .. (344.6,423.84) .. controls (269.98,423.84) and (209.49,363.35) .. (209.49,288.74) -- cycle ;
%Curve Lines [id:da7894264196687798] 
\draw [color={rgb, 255:red, 74; green, 144; blue, 226 }  ,draw opacity=1 ][fill={rgb, 255:red, 255; green, 255; blue, 255 }  ,fill opacity=1 ]   (342,236.53) .. controls (398.51,240.24) and (387.18,288.36) .. (342,290.14) ;
%Curve Lines [id:da7478330350306062] 
\draw [color={rgb, 255:red, 74; green, 144; blue, 226 }  ,draw opacity=1 ][fill={rgb, 255:red, 255; green, 255; blue, 255 }  ,fill opacity=1 ]   (342,236.53) .. controls (285.49,240.24) and (301.18,290.36) .. (342,290.14) ;
%Straight Lines [id:da7124188854487186] 
\draw  [dash pattern={on 0.84pt off 2.51pt}]  (341.5,430.61) -- (341,17.03) ;
\draw [shift={(341,14.03)}, rotate = 89.93] [fill={rgb, 255:red, 0; green, 0; blue, 0 }  ][line width=0.08]  [draw opacity=0] (8.93,-4.29) -- (0,0) -- (8.93,4.29) -- cycle    ;
%Straight Lines [id:da7209329240138381] 
\draw [color={rgb, 255:red, 74; green, 144; blue, 226 }  ,draw opacity=1 ][fill={rgb, 255:red, 74; green, 144; blue, 226 }  ,fill opacity=1 ][line width=1.5]    (341.5,42.61) -- (341.5,134.77) ;
\draw [shift={(341.5,134.77) }, rotate = 90] [color={rgb, 255:red, 74; green, 144; blue, 226 }  ,draw opacity=1 ][fill={rgb, 255:red, 74; green, 144; blue, 226 }  ,fill opacity=1 ][line width=1.5]      (0, 0) circle [x radius= 1.74, y radius= 1.74]   ;
\draw [shift={(341.5,42.61) }, rotate = 90] [color={rgb, 255:red, 74; green, 144; blue, 226 }  ,draw opacity=1 ][fill={rgb, 255:red, 74; green, 144; blue, 226 }  ,fill opacity=1 ][line width=1.5]      (0, 0) circle [x radius= 1.74, y radius= 1.74]   ;
%Straight Lines [id:da8115689461936187] 
\draw [color={rgb, 255:red, 0; green, 0; blue, 0 }  ,draw opacity=1 ] [dash pattern={on 0.84pt off 2.51pt}]  (140,291.03) -- (556.5,288.55) ;
\draw [shift={(559.5,288.53)}, rotate = 179.66] [fill={rgb, 255:red, 0; green, 0; blue, 0 }  ,fill opacity=1 ][line width=0.08]  [draw opacity=0] (8.93,-4.29) -- (0,0) -- (8.93,4.29) -- cycle    ;
%Straight Lines [id:da9143949287682055] 
\draw [color={rgb, 255:red, 74; green, 144; blue, 226 }  ,draw opacity=1 ][line width=1.5]    (341.5,164.61) -- (341.5,222.61) ;
\draw [shift={(341.5,164.61) }, rotate = 90] [color={rgb, 255:red, 74; green, 144; blue, 226 }  ,draw opacity=1 ][fill={rgb, 255:red, 74; green, 144; blue, 226 }  ,fill opacity=1 ][line width=1.5]      (0, 0) circle [x radius= 1.74, y radius= 1.74]   ;
\draw [shift={(341.5,222.61) }, rotate = 90] [color={rgb, 255:red, 74; green, 144; blue, 226 }  ,draw opacity=1 ][fill={rgb, 255:red, 74; green, 144; blue, 226 }  ,fill opacity=1 ][line width=1.5]      (0, 0) circle [x radius= 1.74, y radius= 1.74]   ;
%Curve Lines [id:da5207616838297551] 
\draw [color={rgb, 255:red, 74; green, 144; blue, 226 }  ,draw opacity=1 ][fill={rgb, 255:red, 184; green, 233; blue, 134 }  ,fill opacity=0.3 ]   (189,23.53) .. controls (254.5,36.03) and (290.5,34.53) .. (341.5,37.11) ;
%Curve Lines [id:da11263960020857489] 
\draw [color={rgb, 255:red, 74; green, 144; blue, 226 }  ,draw opacity=1 ][fill={rgb, 255:red, 184; green, 233; blue, 134 }  ,fill opacity=0.3 ]   (497.5,23.03) .. controls (415.5,34.53) and (404.5,35.53) .. (341.5,37.11) ;

% Text Node
\draw (343.92,205.71) node [anchor=north west][inner sep=0.75pt]  [font=\footnotesize,xscale=0.8,yscale=0.8]  {$\mathrm{i} \eta _{2}^{-1}$};
% Text Node
\draw (344.3,119.67) node [anchor=north west][inner sep=0.75pt]  [font=\footnotesize,xscale=0.8,yscale=0.8]  {$\mathrm{i} \eta _{1}$};
% Text Node
\draw (344.52,163.65) node [anchor=north west][inner sep=0.75pt]  [font=\footnotesize,xscale=0.8,yscale=0.8]  {$\mathrm{i} \eta _{1}^{-1}$};
% Text Node
\draw (343.5,46.01) node [anchor=north west][inner sep=0.75pt]  [font=\footnotesize,xscale=0.8,yscale=0.8]  {$\mathrm{i} \eta _{2}$};
% Text Node
\draw (301.03,20.49) node [anchor=north west][inner sep=0.75pt]  [font=\footnotesize,xscale=0.8,yscale=0.8]  {$+$};
% Text Node
\draw (413.53,229.99) node [anchor=north west][inner sep=0.75pt]   [font=\footnotesize,xscale=0.8,yscale=0.8] {$+$};	
\end{tikzpicture}}
\caption{\footnotesize The signature table of Re$[g]$ for different $\xi$. $\Re [g]>0$ in the green region,  while $\Re[g]<0$ in the white region, and $\Re[g]=0$ on the blue curve.
 (a) $\xi=-\frac{\eta_{1}-\eta_{1}^{-1}}{\ln\eta_{1}}$; (b) $ 	\xi_{\mathrm{crit}}<\xi<-\frac{\eta_{1}-\eta_{1}^{-1}}{\ln\eta_{1}}$ and $\xi\in H_{I}$; (c) $\xi_{\mathrm{crit}}<\xi<-\frac{\eta_{1}-\eta_{1}^{-1}}{\ln\eta_{1}}$ and $\xi\in T_{II}$;  (d) $\xi=\xi_{\mathrm{crit}}$; (e) $\xi<\xi_{\mathrm{crit}}$, where $\xi_{\mathrm{crit}}$,  $\alpha(\xi)$ and $\lb_0(\xi)$ are defined in \eqref{def xi}, \eqref{def alpha} and \eqref{p0}, respectively.}
\label{fig reg}
% (a)  The critical velocity $\xi=-\frac{\eta_{1}-\eta_{1}^{-1}}{\ln\eta_{1}}$; (b)  The first genus-1 hyperelliptic wave region $H_I$: $ 	\xi_{\mathrm{crit}}<\xi<-\frac{\eta_{1}-\eta_{1}^{-1}}{\ln\eta_{1}}$; (c) The critical velocity $\xi=\xi_{\mathrm{crit}}$; (d) The second  genus-1 hyperelliptic wave region $H_{II}$: $	\xi<\xi_{\mathrm{crit}}$, where $\xi_{\mathrm{crit}}$ is defined in \eqref{def xi}.}
\end{figure}
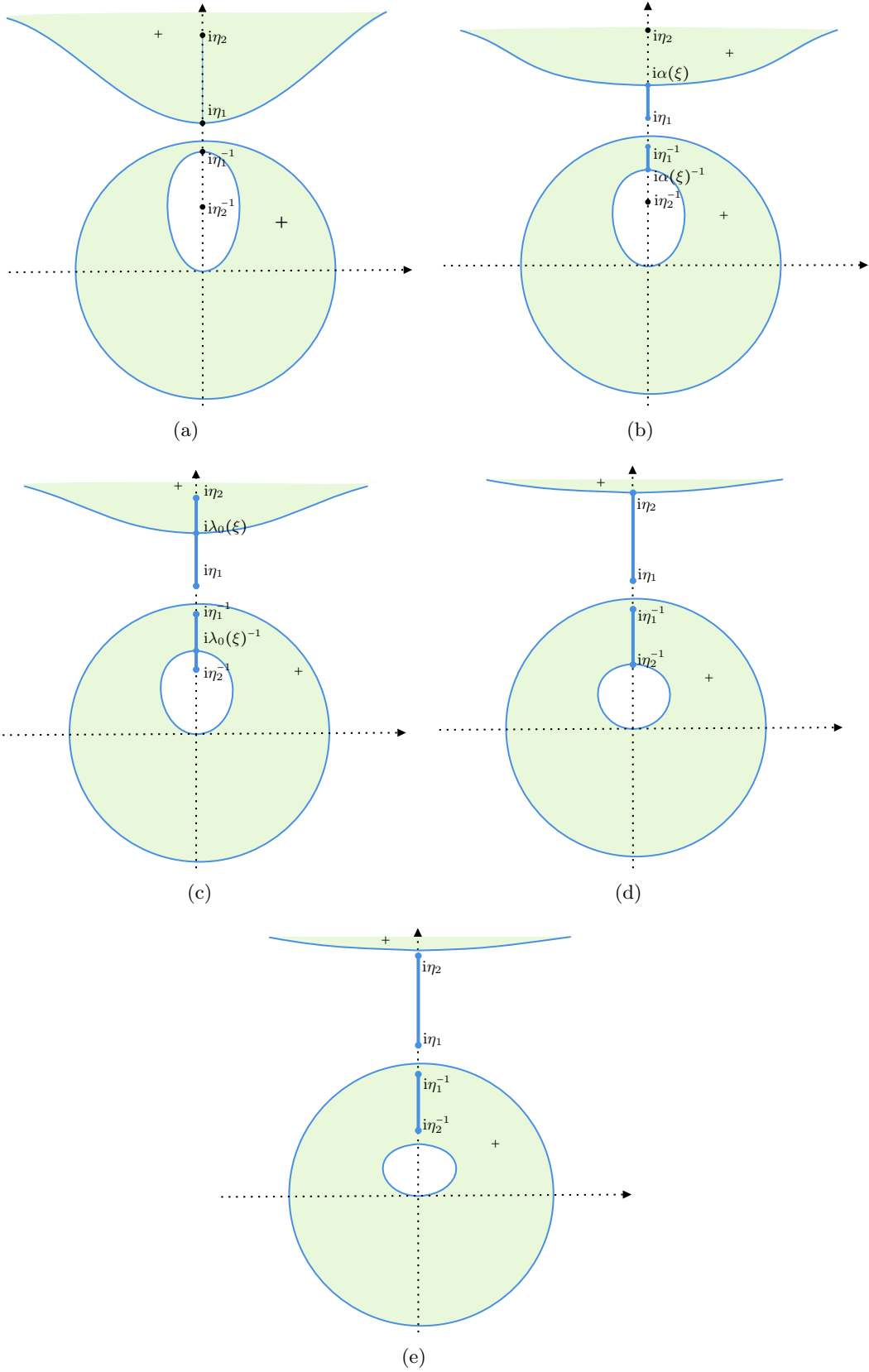

The signature table of Re$[g]$ for $\xi \le -\frac{\eta_{1}-\eta_{1}^{-1}}{\ln\eta_{1}}$ is shown in Figure \ref{fig reg}. It invokes us to make a classification of cases given in  Definition \ref{def region}, which is actually based on the intersection pattern between the critical curve $\re[g]=0$ and the cut $\Sigma_1\cup \Sigma_2$. In Picture (b) of Figure \ref{fig reg}, the critical curve intersects $\Sigma_1\cup \Sigma_2$ at two interior points $\i \alpha(\xi)^{\pm 1}$ away from the endpoints, which corresponds to the first genus-1 hyperelliptic wave region $H_I$. As $\xi$ varies, the intersection points  coalesce with the endpoints $\i \eta_1^{\pm 1}$ or $\i \eta_2^{\pm 1}$,
as illustrated in Pictures (a) and (d), respectively.  In these two cases, the  intersection points are simultaneously the endpoints of the cut and saddle points of the phase function, producing new phenomena locally. In particular, Picture (c)
depicts the configuration of the transition region
$T_{II}$ for the case $\xi>\xi_{\mathrm{crit}}$.
Picture (e) corresponds to the second genus-1 hyperelliptic wave region $H_{II}$, i.e., $\xi<\xi_{\mathrm{crit}}$, where the critical curve moves away from the cut and no intersection occurs.

% we can divide the different asymptotic regions defined in Definition \ref{def region}.  Based on the intersection pattern between the critical curve $\re[g]=0$ and the branch cuts $\Sigma_1\cup \Sigma_2$, we divide  $\xi \le  -\frac{\eta_{1}-\eta_{1}^{-1}}{\ln\eta_{1}}$ into four regions. In plot (b)($ 	\xi_{\mathrm{crit}}<\xi<-\frac{\eta_{1}-\eta_{1}^{-1}}{\ln\eta_{1}}$,  where $\xi_{\mathrm{crit}}$ is defined in \eqref{def xi}), the critical curve intersects the cut at two interior points $\i \alpha(\xi)$ and $\i \alpha(\xi)^{-1}$(distinct from the endpoints); here $\alpha(\xi)$ is modulated by $\xi$, and this regime corresponds to the region $H_I$. When $\alpha(\xi)=\eta_1$, the intersections coalesce with the endpoints $\i \eta_1$ and $\i \eta_1^{-1}$, yielding plot (a). In this case, the  endpoint $\i \eta_1^{\pm1}$ is simultaneously a cut endpoint and a saddle point of the phase function, producing new phenomena in local behavior.  When $\alpha(\xi)=\eta_2$, the intersections occur at the upper endpoints $\i \eta_2$ and $\i \eta_2^{-1}$, yielding plot (c); here the endpoint $\i \eta_2^{\pm 1}$ likewise serves as both a cut endpoint and a saddle point, again leading to new effects. As $\xi$ varies further, i.e. $\xi<\xi_{\mathrm{crit}}$, the critical curve moves away from the cut and no intersection occurs; the branch endpoints are then mere endpoints, and the local behavior near them differs from the previous cases. This configuration corresponds to the region $H_{II}$.

For later use, we also set
\begin{align}
	&g^{(\infty)}:=\lim_{\lb\to\infty}\left( g(\lb)-\frac{\phi(\lb)}{2t}\right).\label{def ginf}
\end{align}
As $\lb \to 0$, using the symmetry realtion \eqref{sym g}, it holds that
\begin{align}
   \lim_{\lb\to0}\left( g(\lb)-\frac{\phi(\lb)}{2t}\right)=-\overline{g^{(\infty)}}.\label{g0}
\end{align}
Moreover,  one can check that
\begin{align}
    \Im{[g^{(\infty)}]}=-\left(1+\frac{\pi\xi}{4}\right).\label{Imginf}
\end{align}

Besides the  $g$-function defined in \eqref{dg}, we further introduce an auxiliary function
\begin{align}
\delta(\lb):=\exp&\left\{\frac{R(\lb)}{2\pi\ii}\left[\int_{\ii\eta_{1}}^{\ii\alpha(\xi)}  \frac{\ln r(s) }{R_+(s)(s-\lb)} \d s-
\int_{\ii\alpha(\xi)^{-1}}^{\ii\eta_{1}^{-1}}  \frac{\ln r(\bar{s}^{-1}) }{R_+(s)(s-\lb)} \d s-\int_{\ii\eta_{1}^{-1}} ^{\ii\eta_{1}} \frac{\ii \Delta}{R(s)(s-\lb)} \d s\right]
\right\},\label{def:f}
\end{align}
where the logarithm takes the principal branch, and by \eqref{int 1/R(s)}, 
\begin{align}\label{def Delta}
	\Delta:=\frac{2\left(\int^{\eta_{1}}_{\alpha(\xi)}  \frac{\ln r(s) }{|R_+(s)|} \d s \right) }{\int_{\eta_{1}^{-1}}^{\eta_{1}}\frac{\d s}{R(s)}}=\frac{l(\alpha(\xi))\left(\int^{\eta_{1}}_{\alpha(\xi)}  \frac{\ln r(s) }{|R_+(s)|} \d s \right) }{\sqrt{(1-l_1^2)(1-l(\alpha(\xi))^2)}(2\Pi(l_1^2,k(\xi))-K(k(\xi)))},
\end{align}
is a real constant to guarantee the boundedness of $\delta(\lb)$ as $\lb\to\infty$.  Here, $k(\xi)$ is given in \eqref{equ xialpha2} and 
	\begin{align}\label{def lk}
	    	l(x)=\frac{x-1}{x+1}.
	\end{align}
It is readily seen from its definition that 
\begin{align}\label{deleta01}
	\delta_+(\lb) =\begin{cases}
		\delta_-(\lb)^{-1}r(\lb),&\lb\in  \ii(\eta_{1},\alpha(\xi)),\\
		\delta_-(\lb)^{-1}r(\bar\lb^{-1})^{-1},&\lb\in \ii(\alpha(\xi)^{-1},\eta_{1}^{-1}),\\
		\delta_-(\lb)\E^{-\ii\Delta},&\lb\in\ii(\eta_1^{-1},\eta_1).
	\end{cases}
\end{align}
Substituting $\lb\to\bar\lb^{-1}$ into \eqref{def:f}, one immediately verifies the symmetry relation 
\begin{align}\label{edeltasym2}
	\delta(\bar\lb^{-1})=\overline{\delta(\lb)^{-1}}.
\end{align}
In addition, one has
\begin{align}
&\delta(\infty):=\lim_{\lb\to\infty}	\delta(\lb)\nonumber\\
&=\exp\left\{\frac{1}{2\pi\ii}\left[\int_{\ii\eta_{1}}^{\ii\alpha(\xi)}  \frac{s\ln r(s) }{R_+(s)} \d s-
\int_{\ii\alpha(\xi)^{-1}}^{\ii\eta_{1}^{-1}}  \frac{s\ln r(\bar{s}^{-1}) }{R_+(s)} \d s-\int_{\ii\eta_{1}^{-1}} ^{\ii\eta_{1}} \frac{\ii \Delta s}{R(s)} \d s\right]
\right\}\in\mathbb{R}. \label{deltainfty}
\end{align}

\subsection{Lenses opening}\label{ol}
If $\xi \in T_I$, the jump matrix of $Z(\lb;n+1,t)$ in \eqref{J} tends to the identity matrix exponentially fast except in small neighbourhoods of   $\lb=\ii\eta_{1}^{\pm 1}$ as $t\to+\infty$. On the other hand, if $\xi \in H_{I} \cup T_{II} \cup H_{II}$, it is not the case and we need to perform lenses opening. To proceed, we first define
% First,   for $\xi \in H_{I} \cup T_{II} \cup H_{II}$, with the functions $g$ and $\delta$ defined in \eqref{dg} and  \eqref{def:f}, we now define
\begin{align}\label{2T1}
	T(\lb)=\delta(\infty)^{\sigma_3}\E^{t g^{(\infty)}\sigma_3}Z(\lb)\E^{-t(g(\lb)-\frac{\phi(\lb)}{2})\sigma_3}\delta(\lb)^{-\sigma_3},
\end{align}
where $\delta(\infty)$ and $g^{(\infty)}$ are given in \eqref{deltainfty} and \eqref{def ginf}, respectively. It is then readily seen from RH problem \ref{RHP0}, Proposition \ref{propg2}, \eqref{deleta01} and \eqref{deltainfty} that $T$ solves the following RH problem.
\begin{Rhp}\ \label{2RHT}
	\begin{itemize}
		\item $T(\lb)$ is analytic in $\ccc\setminus\Sigma$.
		\item For $\lb\in\Sigma$, $T(\lb)$ satisfies the jump condition
			$T_+(\lb)=T_-(\lb)J^{(T)}(\lb),$
		where
		\begin{align}
			&J^{(T)}(\lb)
			:=\begin{cases}
				\left(\begin{matrix}
					r(\lb)^{-1}\delta_-(\lb)^2\E^{2t g_-(\lb)}&0\\\ii&r(\lb)^{-1}\delta_+(\lb)^2\E^{2t g_+(\lb)}
				\end{matrix}\right),&\lb\in\ii(\eta_{1},\alpha(\xi)),\\
				\left(\begin{matrix}
					r(\bar\lb^{-1})^{-1}\delta_+(\lb)^{-2}\E^{-2t g_+(\lb)}&\ii \\0&r(\bar\lb^{-1})^{-1}\delta_-(\lb)^{-2}\E^{-2t g_-(\lb)}
				\end{matrix}\right),&\lb\in\ii(\alpha(\xi)^{-1},\eta_{1}^{-1}),\\
                \begin{pmatrix}
				1&0\\
				\i r(\lb)\delta(\lb)^{-2}\e^{-2tg(\lb)}&1\\
			\end{pmatrix},&\lb\in\ii (\alpha(\xi),\eta_{2}),\\
			\begin{pmatrix}
				1&\i r(\bar{\lb}^{-1})\delta(\lb)^{2}\e^{2tg(\lb)}\\
				0&1\\
			\end{pmatrix},&\lb\in\ii (\eta_{2}^{-1},\alpha(\xi)^{-1}),\\
            \E^{\ii(t\Omega+\Delta)\sigma_3},&\lb\in\ii(\eta_1^{-1},\eta_1),
			\end{cases}
		\end{align}
        where $\Omega$ and $\Delta$ are defined in \eqref{def Omegj} and \eqref{def Delta}, respectively. 
		\item As $\lb\to\infty$, we have
		$ T(\lb)=I+\oo(\lb^{-1}).$
	\end{itemize}
\end{Rhp}
For $\lb\in\ii(\eta_{1},\alpha(\xi))\cup\ii(\alpha(\xi)^{-1},\eta_{1}^{-1})$, it is noticed  that $J^{(T)}$ admits the following factorizations:
\begin{align*}
	J^{(T)}(\lb)=\begin{cases}
		\left(\begin{matrix}
			1&-\ii r(\lb)^{-1}\delta_-(\lb)^2\E^{2t g_-(\lb)}\\0&1
		\end{matrix}\right)\left(\begin{matrix}
			0&\ii\\\ii&0
		\end{matrix}\right)\left(\begin{matrix}
			1&-\ii r(\lb)^{-1}\delta_+(\lb)^2\E^{2t g_+(\lb)}\\0&1
		\end{matrix}\right),&\lb\in\ii(\eta_{1},\alpha(\xi)),\\
		\left(\begin{matrix}
		1&0\\-\ii r(\bar\lb^{-1})^{-1} \delta_-(\lb)^{-2} \E^{-2t g_-(\lb)}&1
	\end{matrix}\right)\left(\begin{matrix}
	0&\ii\\\ii&0
\end{matrix}\right)\left(\begin{matrix}
1&0\\-\ii r(\bar\lb^{-1})^{-1} \delta_+(\lb)^{-2} \E^{-2t g_+(\lb)}&1
\end{matrix}\right),&\lb\in\ii(\alpha(\xi)^{-1},\eta_{1}^{-1}).\\
	\end{cases}
\end{align*}
This, together with Figure \ref{fig reg},
invokes us to open lenses around the interval $\ii(\eta_{1},\eta_0(\xi))\cup\ii(\eta_0(\xi)^{-1},\eta_{1}^{-1})$, where %$\eta_1<\eta_0(\xi)<\eta_2$ is defined as follows:
%is a pair of phase points of function $g$ defined in \eqref{dg}, i.e.,
\begin{align}
    \eta_0(\xi) :=\begin{cases}
        \alpha(\xi), \quad \xi \in  H_{I},\\
        \lb_0(\xi), \quad \xi \in T_{II},\\
        \eta_2, \quad \xi \in H_{II}.
    \end{cases}
\end{align}
Here, $\alpha(\xi)$ is defined in \eqref{def alpha} and $\i \lb_0(\xi)^{\pm 1}$ is a pair of saddle points of the function $g$  with $\eta_1<\lb_0(\xi)<\eta_2$. In fact, from the definition of function $g$ in \eqref{dg}, direct computations show that 
\begin{align}\label{p0}
    \lb_0(\xi) +\lb_0(\xi)^{- 1} = \frac{c_1(\xi) + \sqrt{c_1(\xi)^2+ 4(-\i c_0(\xi)+2) }}{2},
\end{align}
where $c_1(\xi)$ and $c_0(\xi)$ is defined in \eqref{c1c0}.

Let $\Omega_{1,\pm}$ be two regions surrounding $\ii(\eta_{1},\eta_0(\xi))$ and set
\begin{align*}
	\Omega_{2,\pm}=\overline{\Omega_{1,\mp}^{-1}};
\end{align*}
see Figure \ref{figopen1} for an illustration. 
By introducing a matrix-valued function
\begin{align}
	G(\lb)=\left\{\begin{array}{ll}
		\begin{pmatrix}
			1& \pm \ii r(\lb)^{-1} \mathrm{e}^{2tg(\lb)}\\0&1\\
		\end{pmatrix},& \lb\in\Omega_{1,\pm},\\
		\begin{pmatrix}
			1&0\\
			\pm \ii r(\bar{\lb}^{-1})^{-1}  \mathrm{e}^{-2tg(\lb)}&1\\
		\end{pmatrix},& \lb\in\Omega_{2,\pm},\\
        I,&\text{otherwise},
	\end{array}\right.
	\label{def G}
\end{align}
we define
\begin{align}
Z^{(1)}(\lb)=T(\lb)G(\lb).\label{trans1}
\end{align}
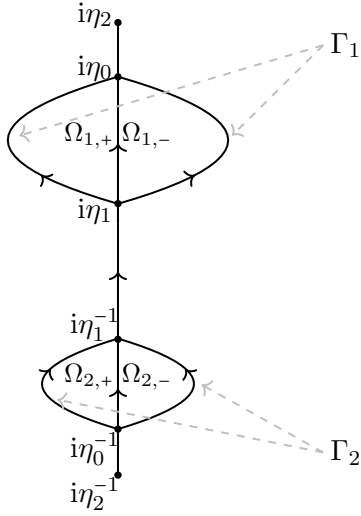
\begin{figure}[h]
	\centering
	\tikzset{every picture/.style={line width=0.75pt}}
	\begin{tikzpicture}
		\draw [->] (0,0) -- (0,1.15);
		\draw [->] (0,1.2) -- (0,2.7);
		\draw [->] (0,10/3) -- (0,4.4);
		\draw [->] (-1,3.97) -- (-1.01,3.98);
		\draw [->] (1,3.97) -- (1.01,3.975);
		\draw [->] (-0.89,1.4) -- (-0.884,1.408);
		\draw [->] (0.89,1.4) -- (0.884,1.408);
		\draw (0,0) -- (0,1);
		\draw (0,1) -- (0,10/3);
		\draw (0,10/3) -- (0,6);
		\draw [domain=3.6:6] plot ({(\x-4.8)*(\x-4.8)-1.45},{(\x-3.6)*0.7+3.6});
		\draw [domain=3.6:6] plot ({-(\x-4.8)*(\x-4.8)+1.45},{(\x-3.6)*0.7+3.6});
		\draw [domain=0:2] plot ({(\x-1)*(\x-1)-1},{0.6*\x+0.61});
		\draw [domain=0:2] plot ({-(\x-1)*(\x-1)+1},{0.6*\x+0.61});
		\filldraw (0,0) circle (1pt);
		\filldraw (0,0.61) circle (1pt);
		\filldraw (0,1.8) circle (1pt);
		\filldraw (0,3.6) circle (1pt);
		\filldraw (0,5.28) circle (1pt);
		\filldraw (0,6) circle (1pt);
		\draw [->,dashed,lightgray] (2.7,5.7) -- (1.5,4.5);
		\draw [->,dashed,lightgray] (2.7,5.7) -- (-1.3,4.5);
		\draw [->,dashed,lightgray] (2.7,0.3) -- (1.1,1.2);
		\draw [->,dashed,lightgray] (2.7,0.3) -- (-0.8,1);
		\node at (-0.38,4.5) {\small $\Omega_{1,+}$};
		\node at (0.38,4.5) {\small $\Omega_{1,-}$};
		\node at (-0.38,1.3) {\small $\Omega_{2,+}$};
		\node at (0.38,1.3) {\small $\Omega_{2,-}$};
		\node at (3,5.7) {$\Gamma_1$};
		\node at (3,0.3) {$\Gamma_2$};
		\node at (-0.3,-0.25) {$\ii\eta_2^{-1}$};
		\node at (-0.3,2) {$\ii\eta_1^{-1}$};
		\node at (-0.3,0.35) {$\ii\eta_0^{-1}$};
		\node at (-0.3,3.5) {$\ii\eta_1$};
		\node at (-0.3,6.1) {$\ii\eta_2$};
		\node at (-0.3,5.4) {$\ii\eta_0$};
	\end{tikzpicture}
	\caption{\footnotesize Lenses around the interval $\ii(\eta_{1},\eta_0(\xi))\cup\ii(\eta_0(\xi)^{-1},\eta_{1}^{-1})$ and the regions $\Omega_{j,\pm}$, $j=1,2$. }\label{figopen1}
\end{figure}
\begin{comment}
Thus,  for $\xi \in H_{I} \cup T_{II} \cup H_{II}$,  with the aid of Proposition \ref{propg2} for $g(\lb)$, properties \eqref{deltainfty} and \eqref{deleta01} for $\delta(\lb)$, and the definition \eqref{def G} of  $G(\lb)$, the first transformation of RH problem \ref{RHP0}  is defined by
\begin{align}
Z^{(1)}(\lb)=\delta(\infty)^{\sigma_{3}}\E^{tg^{(\infty)}\sigma_3}Z(\lb)\E^{- (tg(\lb)-\phi(\lb)/2)\sigma_3}\delta(\lb)^{-\sigma_{3}}G(\lb).\label{trans1}
\end{align}
\end{comment}
It is then readily seen that $Z^{(1)}$ satisfies the following RH problem.
\begin{Rhp} \hfill\label{Rhp 2}
\begin{itemize}
	\item  $Z^{(1)}(\lb)$ is analytic in $\mathbb{C}\setminus(\Gamma_{1}\cup\Gamma_{2}\cup\Sigma)$, where  $\Gamma_{1}$ and $\Gamma_{2}$ are shown in Figure \ref{figopen1}.
	\item For $\lb\in \Gamma_{1}\cup\Gamma_{2}\cup\Sigma$, $Z^{(1)}(\lb)$ satisfies the jump condition $Z^{(1)}_+(\lb)=Z^{(1)}_-(\lb)J^{(1)}(\lb)$, where for $\xi \in H_{I}\cup H_{II}$,
	\begin{align}\label{jump M1}
	J^{(1)}(\lb)=\left\{\begin{array}{ll}
		\begin{pmatrix}
			1&-\ii r(\lb)^{-1}\delta(\lb)^{2} \mathrm{e}^{2tg(\lb)}\\
			0&1\\
		\end{pmatrix},&\lb\in\Gamma_{1},\\
			\begin{pmatrix}
				1&0\\
				-\ii r(\bar{\lb}^{-1})^{-1}\delta(\lb)^{-2} \mathrm{e}^{-2tg(\lb)}&1\\
			\end{pmatrix},&\lb\in\Gamma_{2},\\
			\begin{pmatrix}
				0&\ii  \\
				\ii &0\\
			\end{pmatrix},\quad &\lb\in\ii(\eta_{1},\eta_0(\xi))\cup\ii(\eta_0(\xi)^{-1},\eta_{1}^{-1}),\\
			\begin{pmatrix}
				1&0\\
				\i r(\lb)\delta(\lb)^{-2}\e^{-2tg(\lb)}&1\\
			\end{pmatrix},&\lb\in\ii (\eta_0(\xi),\eta_{2}),\\
			\begin{pmatrix}
				1&\i r(\bar{\lb}^{-1})\delta(\lb)^{2}\e^{2tg(\lb)}\\
				0&1\\
			\end{pmatrix},&\lb\in\ii (\eta_{2}^{-1},\eta_0(\xi)^{-1}),\\
            \mathrm{e}^{\ii (t\Omega+\Delta)\sigma_{3}},\quad&\lb\in\ii(\eta_{1}^{-1},\eta_{1});
		\end{array}\right.
	\end{align}
    for $\xi \in T_{II}$,
    \begin{align}
    	J^{(1)}(\lb)=\left\{\begin{array}{ll}
		\begin{pmatrix}
			1&-\ii r(\lb)^{-1}\delta(\lb)^{2} \mathrm{e}^{2tg(\lb)}\\
			0&1\\
		\end{pmatrix},&\lb\in\Gamma_{1},\\
			\begin{pmatrix}
				1&0\\
				-\ii r(\bar{\lb}^{-1})^{-1}\delta(\lb)^{-2} \mathrm{e}^{-2tg(\lb)}&1\\
			\end{pmatrix},&\lb\in\Gamma_{2},\\
			\begin{pmatrix}
				0&\ii  \\
				\ii &0\\
			\end{pmatrix},\quad &\lb\in\ii(\eta_{1},\eta_0(\xi))\cup\ii(\eta_0(\xi)^{-1},\eta_{1}^{-1}),\\
			\left(\begin{matrix}
					r(\lb)^{-1}\delta_-(\lb)^2\E^{2t g_-(\lb)}&0\\\ii&r(\lb)^{-1}\delta_+(\lb)^2\E^{2t g_+(\lb)}
				\end{matrix}\right),&\lb\in\ii (\eta_0(\xi),\eta_{2}),\\
			\left(\begin{matrix}
					r(\bar\lb^{-1})^{-1}\delta_+(\lb)^{-2}\E^{-2t g_+(\lb)}&\ii \\0&r(\bar\lb^{-1})^{-1}\delta_-(\lb)^{-2}\E^{-2t g_-(\lb)}
				\end{matrix}\right),&\lb\in\ii (\eta_{2}^{-1},\eta_0(\xi)^{-1}),\\
            \mathrm{e}^{\ii (t\Omega+\Delta)\sigma_{3}},\quad&\lb\in\ii(\eta_{1}^{-1},\eta_{1}).
		\end{array}\right.
	\end{align}
	\item  As $ \lb\to\infty$, we have $
	Z^{(1)}(\lb)= I+\oo(\lb^{-1})$.
\end{itemize}
\end{Rhp}
It is noticed from the signature table of Re$[g]$ in Figure \ref{fig reg}  that the jump matrix $J^{(1)}$ on $\Gamma_{1}\cup\Gamma_{2}$  decays exponentially to $I$ as $t\to+\infty$ except around the endpoints. It inspires us to ignore the jump away from the endpoints and approximate $Z^{(1)}$ in global and local manners. In order to achieve this purpose, 
let
\begin{align}\label{def Up}
    U(p):=\{\lb:|\lb-p|\leq \epsilon\},\quad p=\ii\eta_{1},\ \ii\alpha(\xi),
\end{align}
where  $\epsilon>0$ is a sufficiently small constant,
and define 
\begin{align}\label{def Up1}
    U(p):= \overline{U(\bar p^{-1})^{-1}}, \quad p=\ii\eta_{1}^{-1},\ \ii\alpha(\xi)^{-1}.
\end{align}

% \begin{remark}
% From \eqref{p0}, we obtain the estimate $|\lambda_0(\xi)-\eta_2| \lesssim |\xi -\xi_{\mathrm{crit}}|$ as $\xi \to \xi_{\mathrm{crit}}$. From Definition \ref{def region}, for $\xi \in T_{II}$ we have $\lambda_0(\xi) \in U(\i\eta_2)$. Together with the definition of $\alpha(\xi)$ in \eqref{def alpha}, it suffices to analyze the local behavior of $\i\eta_1^{\pm 1}$ and $\i\alpha(\xi)^{\pm 1}$.

% \end{remark}
%Moreover, denote 
%\begin{align}\label{defun}
%    U:=\bigcup_{p=\ii\eta_1^{\pm1},\ii\alpha(\xi)^{\pm1}}U(p).
%\end{align}

In what follows, we will construct  the global parametrix $Z^{(\infty)}$ in Section \ref{sebsec glo} and the local parametrix $Z^{(p)}$ in each  $ U(p)$ in Section \ref{sebsec loc} such that
\begin{align}\label{def:err}
Z^{(1)}(\lb)=\left\{\begin{array}{ll}
	E(\lb)Z^{(\infty)}(\lb),\quad &\lb\in \mathbb{C}\setminus U,\\
	E(\lb)Z^{(p)}(\lb),\quad &\lb\in U(p),\\
\end{array}\right.
\end{align}
where 
\begin{align}\label{defun}
    U:=\bigcup_{p=\ii\eta_1^{\pm1},\ii\alpha(\xi)^{\pm1}}U(p),
\end{align}
and $E$ is an error function satisfying a small norm RH problem as shown in Section \ref{sebsec E}. Note that for $\xi\in T_{II}$, although the lenses start from $ \lambda_0(\xi)^{\pm 1}$, we have $\lambda_0(\xi) \to \eta_2$ for large positive $t$ and $\alpha(\xi)=\eta_2$ (see \eqref{def alpha}) in this case, it thus suffices to analyze the local behavior near $\lambda=\i\eta_1^{\pm 1}$ and $\lambda=\i\alpha(\xi)^{\pm 1}$.

% From \eqref{p0}, we obtain the estimate $|\lambda_0(\xi)-\eta_2| \lesssim |\xi -\xi_{\mathrm{crit}}|$ as $\xi \to \xi_{\mathrm{crit}}$. From Definition \ref{def region}, for $\xi \in T_{II}$ we have $\lambda_0(\xi) \in U(\i\eta_2)$. Together with the definition of $\alpha(\xi)$ in \eqref{def alpha}, it suffices to analyze the local behavior of $\i\eta_1^{\pm 1}$ and $\i\alpha(\xi)^{\pm 1}$.

To maintain consistency in the narrative, we define for $\xi \in T_{I}$,
\begin{align}\label{defZ1T1}
Z^{(1)}(\lb)=Z(\lb),
\end{align}
which can also be approximated in the way of \eqref{def:err} with $p=\ii \eta_{1}^{\pm 1}$ and $U:=\bigcup_{p=\ii\eta_1^{\pm1}}U(p)$.

\subsection{Global parametrix}\label{sebsec glo}
For $\xi \in H_{I} \cup T_{II} \cup H_{II}$,  the global parametrix $Z^{(\infty)}$  satisfies the following  RH problem.
\begin{Rhp} \hfill\label{Rhp glo}
\begin{itemize}
	\item $Z^{(\infty)}(\lb)$ is analytic in $\mathbb{C}\setminus\ii[\alpha(\xi)^{-1},\alpha(\xi)]$.
	\item For $\lb\in \ii(\alpha(\xi)^{-1},\alpha(\xi))$, $Z^{(\infty)}(\lb)$ satisfies the  jump condition \begin{align}
		Z^{(\infty)}_+(\lb)=
		Z^{(\infty)}_-(\lb)\left\{\begin{array}{ll}		
			\begin{pmatrix}
				0&\ii  \\
				\ii &0\\
			\end{pmatrix},\quad &\lb\in\ii(\eta_{1},\alpha(\xi))\cup\ii(\alpha(\xi)^{-1},\eta_{1}^{-1}),\\
			\mathrm{e}^{\ii (t\Omega+\Delta)\sigma_{3}},\quad&\lb\in\ii(\eta_{1}^{-1},\eta_{1}).
		\end{array}\right.
	\end{align}
	\item As $\lb\to\infty$, we have   $ Z^{(\infty)}(\lb)=I+\oo(\lb^{-1}).$
\end{itemize}
\end{Rhp}
Similar to RH problem \ref{r3.2},  one can solve the above RH problem with
\begin{subequations}\label{zinftyc2}
\begin{align}
	&Z^{(\infty)}_{11}(\lb)=\frac{(\kappa(\lb)+1/\kappa(\lb))\vartheta(\AAA(\lb)+\AAA(0)+\frac{1}{2}+\frac{t\Omega+\Delta}{2\pi},\tau)\vartheta(0,\tau)}{2\vartheta(\AAA(\lb)+\AAA(0)+\frac{1}{2},\tau)\vartheta(\frac{t\Omega+\Delta}{2\pi},\tau)},\\
	&Z^{(\infty)}_{12}(\lb)=\frac{(\kappa(\lb)-1/\kappa(\lb))\vartheta(-\AAA(\lb)+\AAA(0)+\frac{1}{2}+\frac{t\Omega+\Delta}{2\pi},\tau)\vartheta(0,\tau)}{2\vartheta(-\AAA(\lb)+\AAA(0)+\frac{1}{2},\tau)\vartheta(\frac{t\Omega+\Delta}{2\pi},\tau)}, \label{eq:Zinf12}\\
	&Z^{(\infty)}_{21}(\lb)=\frac{(\kappa(\lb)-1/\kappa(\lb))\vartheta(\AAA(\lb)-\AAA(0)-\frac{1}{2}+\frac{t\Omega+\Delta}{2\pi},\tau)\vartheta(0,\tau)}{2\vartheta(\AAA(\lb)-\AAA(0)-\frac{1}{2},\tau)\vartheta(\frac{t\Omega+\Delta}{2\pi},\tau)},\\
	&Z^{(\infty)}_{22}(\lb)=\frac{(\kappa(\lb)+1/\kappa(\lb))\vartheta(-\AAA(\lb)-\AAA(0)-\frac{1}{2}+\frac{t\Omega+\Delta}{2\pi},\tau)\vartheta(0,\tau)}{2\vartheta(-\AAA(\lb)-\AAA(0)-\frac{1}{2},\tau)\vartheta(\frac{t\Omega+\Delta}{2\pi},\tau)},
\end{align}
\end{subequations}
where
%The above RH problem is explicitly solvable.
%To present its explicit solution, we first  introduce the notations that will appear in the expression as follows. The analytic function $\kappa=\kappa(\lb)$, defined for  $\lb\in\mathbb{C}\setminus\ii(\eta_{1},\alpha(\xi))\cup\ii(\alpha(\xi)^{-1},\eta_{1}^{-1})$, is given by
\begin{align*}
	\kappa(\lb)= \left( \frac{\lb-\ii\alpha(\xi)}{\lb-\ii\eta_{1}}\right)^{\frac{1}{4}} \left(\frac{\lb-\ii\eta_{1}^{-1}}{\lb-\ii\alpha(\xi)^{-1}}\right)^{\frac{1}{4}},\quad \lb \in \C \setminus (\i(\eta_1, \alpha(\xi) )\cup \i (\eta_1^{-1},\alpha(\xi)^{-1})),
\end{align*}  
and the branch is chosen such that $  \kappa(\lb)=1+\oo(\lb^{-1})$ as $ \lb\rightarrow\infty$.
Moreover, by \eqref{int 1/R(s)}, the normalized Abel differential now becomes
\begin{align*}
	\omega(\lb)=\frac{\ii l(\alpha(\xi))}{2K(k(\xi))\sqrt{(1-l(\eta_1)^2)(1-(\alpha(\xi))^2)}}\frac{\ddd \lb}{R(\lb)},
\end{align*}
where $k(\xi)=l(\eta_1)/l(\alpha(\xi))$ with $l(x)$ defined in \eqref{def lk}, $K$ and $R$ are given in \eqref{ei1} and \eqref{def R} respectively. The corresponding Abel-Jacobi map is given by
\begin{align*}
    \AAA(\lb) =\int_{\i \alpha(\xi)}^\lb \omega, \quad \lb \in \C \setminus \i(\alpha(\xi)^{-1},\alpha(\xi)),
\end{align*}
where the path of integration is on any simple arc from $\i\alpha(\xi)$ to $\lb$ which does not intersect $\i(\alpha(\xi)^{-1},\alpha(\xi))$ with 
the $\mathfrak{b}$-period defined by
\begin{align}\label{def tau}
	\tau:=\int_{\mathfrak{b}}\omega=\frac{\ii K(\sqrt{1-k(\xi)^2})}{2K(k(\xi))}.
\end{align}

For $\xi \in T_{I}$, the global paramertix is given by 
\begin{align}\label{T1glo}
	    Z^{(\infty)}(\lb) = \left(\frac{\lb- \ii \eta_1}{\lb- \ii \eta_1^{-1}}\right)^{\tilde{m} \sigma_3},
	\end{align}
where $\tilde{m} \in \mathbb{N}$ is an integer to be determined later.

% we have the different global parametrix. Let $\tilde{m} \in \mathbb{N}$ be an integer to be determined later. Then for $\lb \in \C\setminus U$, 
% 	\begin{align}\label{T1glo}
% 	    Z^{(\infty)}(\lb) = \left(\frac{\lb- \ii \eta_1}{\lb- \ii \eta_1^{-1}}\right)^{\tilde{m} \sigma_3}.
% 	\end{align}
%  which can be induced by matching with the generalized Laguerre polynomials in $U(p),\ p= \ii \eta_{1}^{\pm 1}$. See more details in the following subsection.

\subsection{Local parametrices}\label{sebsec loc}	
In this section, we build local parametrix $Z^{(p)}$ in each $U(p)$. As aforementioned, the constructions differ for $\xi$ belonging to different regions, which will be treated separately. In what follows, we shall state the local RH problem in a unified manner, and it is understood that the jump condition occurs only when the corresponding set is not an empty set.  
			% In this subsection, we present the construction of the local parametrices near the critical points. 
   %          We show that, depending on the behaviors of the phase functions $\phi$ and $g$, the local RH problem $Z^{(p)}$ for each points $p$ can be  described in terms of
		 % local parametrices associated with different special functions. 
% For $\xi \in T_I$, the local parametrices near $p=\ii \eta_{1}^{\pm1}$ are expressed in terms of the generalized Laguerre polynomial of index $0$;   for $\xi \in H_I$, the local parametrices near $p=\i \eta_1^{\pm1}$ and $p= \i \alpha(\xi)^{\pm1}$ are described by the Bessel function of index $0$ and the Airy function, respectively;  for $\xi \in T_{II}$, the local parametrices near $p=\ii \eta_1^{\pm1}$ and $p=\ii \eta_2^{\pm1}$ are expressed through the Bessel function of index $0$ and the Painlev\'e XXXIV equation, respectively; while for $\xi \in H_{II}$, the local parametrices near $p=\ii \eta_1^{\pm1}$ and $p=\ii \eta_2^{\pm1}$ are described solely by the Bessel function of index $0$.
    
\paragraph{Local parametrices for $\xi \in T_{I}$}
 
%  Now we present the local parametrices for $\xi \in T_{I}$.
By  Definition \ref{def region},  the region $T_I$  is divided into a family of subregions labeled by a parameter $m\in \{0\}\cup \mathbb{N}$:
\begin{align*}
 T_I=\left\{(n+1,t):\  \bigcup_{m=0}^{\infty}T^{(m)}_I\right\}, \  T^{(m)}_I=\left\{	 (n+1,t):\ -\frac{2m+1}{\ln \eta_{1}}\frac{\ln t}{t}<\xi+\frac{\eta_{1}-\eta_{1}^{-1}}{\ln\eta_{1}}< - \frac{2m-1}{\ln \eta_{1}}\frac{\ln t}{t}\right\}. 
\end{align*}
Local parametrices are constructed separately within each subregion $T^{(m)}_I$. For a fixed $m$,   the construction proceeds as follows.

For	$p=\ii \eta_{1}^{\pm 1}$,  the local RH problem in each small neighborhood $U(p)$ defined in  \eqref{def Up} and \eqref{def Up1}, is formulated as follows.

			\begin{Rhp}\hfill\label{Rhp loc0}
				\begin{itemize}
					\item $Z^{(p)}(\lb)$ is analytic in $ U(p)\setminus\left( \Sigma_1 \cup \Sigma_2 \right)$.
					\item $Z^{(p)}(\lb)$  satisfies the jump condition $$Z^{(p)}_+(\lb)=Z^{(p)}_-(\lb)\left\{\begin{array}{ll}
					\begin{pmatrix}
						1&0\\
						\i r(\lb)\e^{-\phi(\lb)}&1\\
					\end{pmatrix},&\lb\in \Sigma_1 \cap U(p),\\
					\begin{pmatrix}
						1&\i r(\bar{\lb}^{-1})\e^{\phi(\lb)}\\
						0&1\\
					\end{pmatrix},&\lb\in \Sigma_2 \cap U(p),
					\end{array}\right.$$
                    where $\Phi$ is defined in \eqref{phi}. 
					\item As $t\to+\infty$, $Z^{(p)}(\lb)$ matches $Z^{(\infty)}(\lb)$ on the boundary $\partial U(p)$ of $U(p)$.
				\end{itemize}
			\end{Rhp}	
The above RH problem can be solved explicitly with the aid of the generalized Laguerre polynomial parametrix introduced  in Appendix \ref{Lague}.  We give a sketch of the construction in what follows. 

For $p=\i \eta_1^{\pm 1}$, we define an analytic function in $U(p)$ by
          \begin{align}\label{defz0}
%        \zeta:=\zeta(\lb) = \pm t\left(-\frac{\eta_1-\eta_1^{-1}}{ \ln \eta_1} \ln \lb -\ii (\lb+\lb^{-1}-2) -\ii \left(2-\frac{\pi}{2} \frac{\eta_1-\eta_1^{-1}}{ \ln \eta_1}  \right)           \right).\\
         \zeta=\zeta(\lb) := \mp t\left(\frac{\eta_1-\eta_1^{-1}}{ \ln \eta_1} \ln \lb +\ii (\lb+\lb^{-1}) -\frac{\pi\ii}{2} \frac{\eta_1-\eta_1^{-1}}{ \ln \eta_1}  \right).
          \end{align}
It is readily seen that 
% follows that as $\lb \to p$,
\begin{align} \label{zetaT1}
    \zeta= \mp t \left( \frac{\eta_1-\eta_1^{-1}}{ \ln \eta_1} p^{-1}  + \i (1-p^{-2})
    \right) (\lb-p) + \oo((\lb- p)^2), \quad \lambda \to p.
\end{align}
Thus, we have
		 \begin{align*}
		     \zeta(p)=0, \qquad \arg( \zeta'(p))=-\frac{\pi}{2}.
		 \end{align*}
The jump contour from the $\lb$-plane to the $\zeta$-plane under the mapping \eqref{defz0}  for  $\lb \in U(\ii \eta_1)$ is depicted in Figure \ref{scalingzk}.

			\begin{figure}
 \begin{center}
   \begin{tikzpicture}[scale=1]
 %   \draw[dotted,thick,teal](-3,0) circle (2);
       \node[shape=circle,fill=black, scale=0.13]  at (-3,0){0};
        \node[below] at (-3,0) { $\ii \eta_1$};
           \draw[thick] (-3,0) --(-3,2);
         \draw[thick,->] (-3,0) --(-3,1);
        \node  at (-2.3,0.86) { $\Sigma_{1}$};
        \node[]  at (-3,2.3) { $\lb$-plane};
   %     \draw [dotted,-latex,teal ]  (-1, 0)  to  [out=90,  in=0] (-3, 2);
  %   \draw[dotted,thick,teal](3,0) circle (2);
        \node[shape=circle,fill=black, scale=0.13]  at (2.5,0){0};
        \node[below] at (2.5,0) { $0$};
           \draw[thick] (2.5,0) --(4.5,0);
         \draw[thick, ->] (2.5,0) --(3.6,0);
      %  \node  at (5,0.86) { $\Sigma_{1}$};
        \node[]  at (3.3,2.3) { $\zeta$-plane};

 %   \node    at (-3, -3 )  {  $ (a)$};
 %   \node    at (4, -3 )  {  $ (b)$};
 \draw [-latex] (-0.5,0)--(0.5,0);
    \end{tikzpicture}
    \caption{\footnotesize{The jump contour of $Z^{(\ii \eta_1)}$ from the $\lb$-plane (left) to the $\zeta$-plane (right) under the mapping \eqref{defz0}. }}
      \label{scalingzk}
  \end{center}
\end{figure}
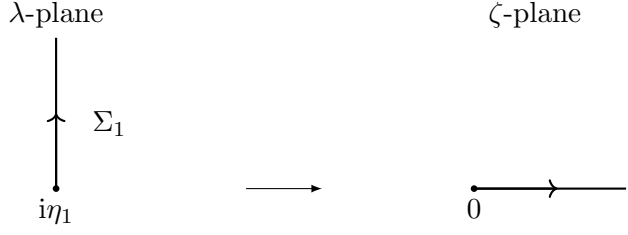

For $\lb \in U(\i \eta_1)$, we then define
%the following local parametrix  for the endpoints $p=\ii \eta_{1}^{\pm1}$, $\lb \in U(p)$,
	\begin{align}\label{2tt1}
	    Z^{(\i \eta_1)}(\lb)= A(\lb)
	L(\zeta(\lb))G^{(\i \eta_1)}(\lb)^{-1}, 
	\end{align}	
where $L$ depending on $m$, is  the generalized Laguerre polynomial parametrix given in \eqref{def:LagSol},
\begin{align}
    A(\lb)= Z^{(\infty)}(\lb) \zeta^{m\sigma_3} G^{(\i \eta_1)}(\lb)
    \end{align}
 with
    \begin{align}
    G^{(\i \eta_1)}(\lb) =
         \left(\ii r(\lb) \right)^{-\frac{\sigma_3}{2}} \E^{\frac{t}{2}\left( \left( \xi + \frac{\eta_1-\eta_1^{-1}}{\ln \eta_1} \right)\ln \lb + \ii  \left(2-\frac{\pi}{2} \frac{\eta_1-\eta_1^{-1}}{\ln \eta_1}\right)\right)\sigma_3}.
\end{align}

To avoid pole singularity near at $\lb=\i \eta_1$, we choose $\tilde{m}$ in \eqref{T1glo} to be $-m$, which also yields that $A$ is analytic in $ U(\i \eta_1)$. 
For $\lb \in U(\ii\eta_{1}^{-1})$, the local parametrix is simply defined through the symmetry relation 
				\begin{align}\label{locetaT12}
	Z^{(\i \eta_1^{-1})}(\lb)=\sigma_2\overline{Z^{(\infty)}(0)^{-1}Z^{(\i \eta_1)}(\bar\lb^{-1})}\sigma_2.
\end{align}

% In order not to have poles at $\lb=\i \eta_1^{\pm1}$, we henceforth choose $\tilde{m}$ in \eqref{T1glo} equals to $-m$. Then, in view of RH problem \ref{Rhp glo} for $Z^{(\infty)}$, a direct calculation shows that $A$ is analytic in $ U(p)$.
%Also,  using RH problem \ref{T1glo} and the generalized Laguegrre polynomial parametrix in Appendix \ref{Lague}, it follows that $Z^{(p)}(\lb)$ satisfies the requirements in RH problem \ref{Rhp loc0}.

From RH problem \ref{Lzeta} for $L$, one can check that $Z^{(p)}$ defined in \eqref{2tt1} and \eqref{locetaT12} indeed solve RH problem \ref{Rhp loc0} for $p=\i \eta_1^{\pm 1}$.  In particular, 
 we obtain from  \eqref{T1glo}, \eqref{zetaT1} and \eqref{Linf} that, as $t\to+\infty$,
		% we find that for $p=\ii \eta_1^{\pm 1}$, $Z^{(p)}(\lb)$ satisfies the asymptotic behavior for $\lb\in \partial U(p)$  as:
				\begin{align*}					
		Z^{(p)}(\lb)Z^{(\infty)}(\lb)^{-1}=\begin{cases}
	I+\oo(\min (\zeta^{-1},  \zeta^{2m-1}\E^{t(\xi+\frac{\eta_1-\eta_1^{-1}}{\ln \eta_1})\ln \eta_1}), \zeta^{-2m-1}\E^{-t(\xi+\frac{\eta_1-\eta_1^{-1}}{\ln \eta_1})\ln \eta_1})),  \ &m\ge 1,
    \\
		    I+\oo( \min( \zeta^{-1} \E^{-t(\xi +\frac{\eta_1-\eta_1^{-1}}{\ln \eta_1} ) \ln \eta_1}, \zeta^{-1} \E^{t(\xi +\frac{\eta_1-\eta_1^{-1}}{\ln \eta_1} ) \ln \eta_1} )), \ &m=0,
		\end{cases}
					\end{align*}
 for $\xi \in T_I^{(m)}$ and $\lb\in\partial U(p)$. Combining with \eqref{defz0}, we arrive at
				 \begin{align}\label{asy:loc y0}
		& Z^{(p)}(\lb)Z^{(\infty)}(\lb)^{-1}\nonumber 
        \\
        &=\begin{cases}I+\oo (\min( t^{-1}, \E^{(2m-1)\ln t +  t(\xi+\frac{\eta_1-\eta_1^{-1}}{\ln \eta_1})\ln \eta_1},\E^{-(2m+1)\ln t - t(\xi+\frac{\eta_1-\eta_1^{-1}}{\ln \eta_1})\ln \eta_1} )),\ &m \ge 1,\\
		I+\oo( \min( t^{-1} \E^{-t(\xi +\frac{\eta_1-\eta_1^{-1}}{\ln \eta_1} ) \ln \eta_1}, t^{-1} \E^{t(\xi +\frac{\eta_1-\eta_1^{-1}}{\ln \eta_1} ) \ln \eta_1} )), \ &m=0.
		\end{cases}
	    \end{align}
        % In view of RH problem \ref{Lzeta} for $L$, we conclude that $Z^{(p)}$ defined in \eqref{2tt1} solves RH problem \ref{Rhp loc0} with $p=	\ii\eta_1^{\pm1}$.	

\paragraph{Local parametrices for $\xi \in H_{I}$}

%Now we present the local parametrices for $\xi \in H_{I}$. 
For $p=\i \eta_1^{\pm1}, \i \alpha(\xi)^{\pm1}$, 
the local RH problem in each small neighborhood $U(p)$ defined in \eqref{def Up} and \eqref{def Up1}  reads as follows.		
				\begin{Rhp}\hfill\label{Rhp loc}
				\begin{itemize}
					\item $Z^{(p)}(\lb)$ is  analytic in $ U(p)\setminus\left(\Gamma_{1}\cup\Gamma_{2}\cup \Sigma\right)$.
					\item $Z^{(p)}(\lb)$  satisfies the jump condition $$Z^{(p)}_+(\lb)=Z^{(p)}_-(\lb)\left\{\begin{array}{ll}
						\begin{pmatrix}
						1&-\ii r(\lb)^{-1}\delta(\lb)^2 \mathrm{e}^{2tg(\lb)}\\
						0&1\\
					\end{pmatrix},&\lb\in\Gamma_{1}\cap U(p),\\
					\begin{pmatrix}
						1&0\\
						-\ii r(\bar{\lb}^{-1})^{-1}\delta(\lb)^{-2} \mathrm{e}^{-2tg(\lb)}&1\\
					\end{pmatrix},&\lb\in\Gamma_{2}\cap U(p),\\
					\begin{pmatrix}
						0&\ii  \\
						\ii &0\\
					\end{pmatrix},\quad &\lb\in(\ii(\eta_{1},\alpha(\xi))\cup\ii(\alpha(\xi)^{-1},\eta_{1}^{-1}))\cap U(p),\\
					\mathrm{e}^{\ii (t\Omega+\Delta)\sigma_{3}},\quad&\lb\in\ii(\eta_{1}^{-1},\eta_{1})\cap U(p),\\
					\begin{pmatrix}
						1&0\\
						\i r(\lb)\delta(\lb)^{-2}\e^{-2tg(\lb)}&1\\
					\end{pmatrix},&\lb\in\ii (\alpha(\xi),\eta_{2})\cap U(p),\\
					\begin{pmatrix}
						1&\i r(\bar{\lb}^{-1})\delta(\lb)^2\e^{2tg(\lb)}\\
						0&1\\
					\end{pmatrix},&\lb\in\ii (\eta_{2}^{-1},\alpha(\xi)^{-1})\cap U(p).
					\end{array}\right.$$
					\item As $t\to+\infty$, $Z^{(p)}(\lb)$ matches $Z^{(\infty)}(\lb)$ on the boundary $\partial U(p)$ of $U(p)$.
				\end{itemize}
			\end{Rhp}

            The above problem  for $p=\i \eta_1^{\pm1}$ can be solved explicitly with the aid of the Bessel parametrix introduced in Appendix \ref{app bessel}, while for $p=\i \alpha(\xi)^{\pm 1}$ it can be solved explicitly using the Airy parametrix introduced in Appendix \ref{app airy}. We give a sketch of the construction in what follows. 
            
For $p=\i \eta_1^{\pm1} $, it follows from the definition of $g(\lb)$ in \eqref{dg} that as $\lb\to p$ and $\pm \Re \lb>0$, 
			\begin{align}
					g(\lb)=g_\pm(p)+\sqrt{(\lb-p)\frac{ (-(\eta_1^2+\eta_1^{-2})+c_1(\xi)(\eta_1+\eta_1^{-1})-\ii c_0(\xi))^2}{(p-\ii\alpha(\xi))(p-\ii\alpha(\xi)^{-1})(p+p^{-1})}}+\oo\left((\lb-p)^{\frac{3}{2}}\right). 
				\end{align}
			%where the sign $\pm$ is contingent upon whether $\lb$ converges to $p$ from the $+$ or $-$ direction. It also holds 
            Note that $g_+(p)=-g_-(p)$ and
				 \begin{align}
				 		g_+(\ii\eta_{1}^{\pm1})=\frac{-
				 			\ii\Omega}{2},				
				\end{align}
			with $\Omega$ given in \eqref{def Omegj}.
				We then define a local conformal map in  $U(p)$ by
                \begin{align}\label{T2zeta1}
					\zeta=\zeta(\lb):=\frac{t^2}{4}(g(\lb)-g_\pm(p))^2,\quad \mp\Re\lb>0 
				\end{align}
				with $\zeta(p)=0$ and $$\zeta'(p) = \frac{t^2 (-(\eta_1^2+\eta_1^{-2})+c_1(\xi)(\eta_1+\eta_1^{-1})-\ii c_0(\xi))^2}{4(p-\ii\alpha(\xi))(p-\ii\alpha(\xi)^{-1})(p+p^{-1})}\in\ii\mathbb{R}^{\pm},\quad p=\ii\eta_1^{\pm1}.$$
                         
				For $\lb\in U(\ii\eta_{1})$,  the local parametrix thereby is given by
				\begin{align}\label{loc1}
					Z^{(\ii\eta_{1})}(\lb)=A(\lb)\Psi_{\Bes}(\zeta(\lb))\sigma_{1}\mathrm{e}^{2\zeta(\lb)^{\frac{1}{2}}\sigma_{3}}(\E^{\frac{\pi\ii}{4}}\delta(\lb)r(\lb)^{-\frac{1}{2}})^{-\sigma_{3}} \mathrm{e}^{\pm\frac{ \ii  t\Omega\sigma_{3}}{2}},\quad \zeta\in \C^{\pm},
				\end{align}
				where	
				\begin{align*}
				&	A(\lb)=Z^{(\infty)}(\lb)\mathrm{e}^{\mp\frac{ \ii  t\Omega\sigma_{3}}{2}}(\E^{\frac{\pi\ii}{4}}\delta(\lb)r(\lb)^{-\frac{1}{2}})^{\sigma_{3}}\frac{1}{\sqrt{2}}\begin{pmatrix}
						-\ii &1\\1&-\ii \\
					\end{pmatrix}(2\pi\zeta^{\frac{1}{2}})^{\frac{\sigma_{3}}{2}},
				\end{align*}
				and $\Psi_{\Bes}$ is the Bessel parametrix defined in \eqref{def:Bes}. In view of RH problem \ref{Rhp glo} for $Z^{(\infty)}$, a direct calculation shows that $A$ is analytic in $ U(\i \eta_1)$.
                %Also,  using RH problem \ref{Rhp glo} and the Bessel parametrix in Appendix \ref{app bessel}, it follows that $Z^{(\ii\eta_1)}(\lb)$ satisfies the requirements in RH problem \ref{Rhp loc}.
				For $\lb \in U(\ii\eta_{1}^{-1})$, the local parametrix is simply defined through the symmetry relation 
				\begin{align}\label{loceta12}
	Z^{(\i \eta_1^{-1})}(\lb)=\sigma_2\overline{Z^{(\infty)}(0)^{-1}Z^{(\i \eta_1)}(\bar\lb^{-1})}\sigma_2.
\end{align}

From RH problem \ref{rhp:Bes} for $\Psi_{\mathrm{Bes}}$,  one can check that $Z^{(p)}$ defined in \eqref{loc1} and  \eqref{loceta12} solves   RH problem \ref{Rhp loc} for $p=\ii\eta_1^{\pm1}$.	
%Additionally,  since $Z^{(\ii\eta_1)}(\lb)$ have fulfilled the requirements, it is easy to confirm that $Z^{(\ii\eta_1^{-1})}(\lb)$ also fulfill them in RH problem \ref{Rhp loc}.
Moreover, as $t\to +\infty$, we have  
\begin{align*}					
Z^{(p)}(\lb)Z^{(\infty)}(\lb)^{-1}=I+\oo(\zeta^{-\frac{1}{2}}),\qquad \lb\in \partial U(p).
\end{align*}
Combining with \eqref{T2zeta1}, this leads to
					 \begin{align}\label{eq:asyBoundary}
				Z^{(p)}(\lb)Z^{(\infty)}(\lb)^{-1}=I+\oo(t^{-1}),
                \qquad \lb\in \partial U(p).
					\end{align}

We next turn to the local parametrix for $p=\i \alpha(\xi)^{\pm 1}$. 
From the definition of $g(\lb)$ in \eqref{dg}, it follows that as $\lb\to p$,
			\begin{align}
			g(\lb)&=(\lb-p)^{\frac{3}{2}}\sqrt{\frac{p+p^{-1}}{(p-\ii\eta_1)(p-\ii\eta_1^{-1})}}\frac{2\xi+3(\alpha(\xi)+\alpha(\xi)^{-1})-\eta_1-\eta_1^{-1}}{4p}\times \left(1+\oo\left( \lb-p\right) \right).
			\end{align}
%					Therefore, we 
        Define  a local conformal map in $U(p)$, $p=\ii\alpha(\xi)^{\pm1}$, by 
        \begin{align}
						\zeta=\zeta(\lb):=\left(\pm\dfrac{3tg(\lb)}{2} \right)^{\frac{2}{3}}.\label{def zeta2}
					\end{align}
					A direct calculation shows that $	\zeta(p)=0$ and
					\begin{align*}
						\zeta'(p)= t^{\frac{2}{3}}\left[\frac{6\xi+9(\alpha(\xi)+\alpha(\xi)^{-1})-3(\eta_1+\eta_1^{-1})}{-8p}\left(\frac{p+p^{-1}}{(p-\ii\eta_1)(p-\ii\eta_1^{-1})}\right)^{\frac{1}{2}} \right]^{\frac{2}{3}}\in\ii\mathbb{R}^{\mp}.
                        %,\quad p=\ii\alpha(\xi)^{\pm1}.
					\end{align*}
		For $p=\ii\alpha(\xi)$, the local parametrix is then given by
					\begin{align}\label{1th1}
						Z^{(\ii\alpha(\xi))}(\lb)=A(\lb)\Psi_{\Ai}(\zeta(\lb)) \e^{\frac{2}{3} \zeta(\lb)^{\frac{3}{2}}\sigma_3}G^{(\ii\alpha(\xi))}(\zeta(\lb)),
					\end{align}
					where	
                      % \begin{align*}
					%	&G^{(\ii\alpha(\xi))}(\zeta)=
					%	\left\{	\begin{array}{ll}
					%		\sigma_{1}(\E^{\frac{\pi\ii}{4}}\delta(\lb)r(\lb)^{-\frac{1}{2}})^{-\sigma_{3}}, & p=\ii\alpha(\xi),\\
					%		(\E^{\frac{\pi\ii}{4}}\delta(\lb)r(\bar{\lb}^{-1})^{\frac{1}{2}})^{-\sigma_{3}}, & p=\ii\alpha(\xi)^{-1},
					%	\end{array}\right.\\
					%	&	A(\lb)=Z^{(\infty)}(\lb)G^{(p)}(\lb)^{-1}\frac{1}{\sqrt{2}}\begin{pmatrix}
				%	1	&	-\ii \\
				%		-	\ii &1\\
				%		\end{pmatrix}\zeta(\lb)^{\frac{\sigma_{3}}{4}},
				%	\end{align*} 
					  \begin{align*}
						&G^{(\ii\alpha(\xi))}(\zeta)=					
							\sigma_{1}(\E^{\frac{\pi\ii}{4}}\delta(\lb)r(\lb)^{-\frac{1}{2}})^{-\sigma_{3}},\\
						&	A(\lb)=Z^{(\infty)}(\lb)G^{(p)}(\lb)^{-1}\frac{1}{\sqrt{2}}\begin{pmatrix}
					1	&	-\ii \\
						-	\ii &1\\
						\end{pmatrix}\zeta(\lb)^{\frac{\sigma_{3}}{4}},
					\end{align*} 
					and $\Psi_{\Ai}$ is the Airy parametrix defined in \eqref{def:Airy}.  In view of RH problem \ref{Rhp glo} for $Z^{(\infty)}$, a direct calculation shows that $A$ is analytic in  $U(\i \alpha(\xi))$.  
                    %Also,  using RH problem \ref{Rhp glo} and the Airy parametrix in Appendix \ref{app airy}, it follows that $Z^{(p)}(\lb)$ satisfies the requirements in RH problem \ref{Rhp loc}.
                    And for $p=\ii\alpha(\xi)^{-1}$, the local parametrix is then given by
                    \begin{align}\label{loceta13}
	Z^{(\i \alpha(\xi)^{-1})}(\lb)=\sigma_2\overline{Z^{(\infty)}(0)^{-1}Z^{(\i \alpha(\xi))}(\bar\lb^{-1})}\sigma_2.
\end{align}                 
                    Consequently, for the endpoints
					$	p=	\ii\alpha(\xi)^{\pm1}$, we analogously obtain from  \eqref{asymlo} that $Z^{(p)}$ satisfies the asymptotic behavior for $\lb\in \partial U(p)$  as \begin{align*}					
						Z^{(p)}(\lb)Z^{(\infty)}(\lb)^{-1}=I+\oo(\zeta^{-\frac{3}{2}}),\quad t\rightarrow +\infty,
					\end{align*}
					or equivalently, by \eqref{def zeta2}, we have
					\begin{align}\label{asy:loc y2}
						Z^{(p)}(\lb)Z^{(\infty)}(\lb)^{-1}=I+\oo(t^{-1}), \qquad \lb\in \partial U(p).
					\end{align}
	This, together RH problem \ref{RHairy} for $\Psi_{\Ai}$, implies  that $Z^{(p)}$ defined in  \eqref{1th1} solves RH problem \ref{Rhp loc} with $	p=	\ii\alpha(\xi)^{\pm1}$.			

\paragraph{Local parametrices for $\xi \in H_{II}$}
In this case, all the local parametrices can be built with the aid of the Bessel parametrix. As the construction is similar to that of $Z^{(p)}$ with $p=\i \eta_1^{\pm 1}$ for $\xi \in H_{I}$, we omit the details here. Moreover, we note that the estimate 
\eqref{eq:asyBoundary} still holds in this case.

\paragraph{Local parametrices for $\xi \in T_{II}$} \label{SubsecP34loc}

%Now we present the local parametricesfor $\xi \in T_{II}$. 
For $p=\ii\eta_{1}^{\pm 1}$, $Z^{(p)}$ can be constructed by \eqref{loc1} and \eqref{loceta12} respectively. For  $p=\ii\alpha(\xi)^{\pm 1}=\ii\eta_2^{\pm}$, it needs a different construction.
%Moreover, in this region, there exists a pair of phase points $ \ii \lb_0^{\pm 1}$, $\lb_0>0$ of $g$ which satisfies that: When $\xi<\xi_{\mathrm{crit}}$,  $\re g(\ii \lb_0)=0$; When  $\xi>\xi_{\mathrm{crit}}$, $g_{\lb}(\ii \lb_0)=0$. As $\xi \to \xi_{\mathrm{crit}}$, $|\lb_0-\eta_2|\sim |\xi -\xi_{\mathrm{crit}} |$. Thus, when  $\xi>\xi_{\mathrm{crit}}$ in this region, we should open the lenses at the point $ \ii \lb_0^{\pm 1}$ instead of $\ii \alpha(\xi)^{\pm 1}$ in Subsection \ref{ol}.  
 For convenience,  we focus on the region $\left\{(n+1,t):0\le\xi-\xi_{\mathrm{crit}}\le Ct^{-2/3}\right\}$ in what follows, as the discussions for the other half region is similar.

For $p=\ii \eta_2^{\pm1}$, the local RH problem in each small neighborhood $U(p)$ defined in \eqref{def Up} and \eqref{def Up1}  reads as follows.		
				\begin{Rhp}\hfill\label{Rhp locp34}
				\begin{itemize}
					\item $Z^{(p)}(\lb)$ is analytic in $ U(p)\setminus\left(\Gamma_{1}\cup\Gamma_{2}\cup\ \Sigma\right)$.
					\item $Z^{(p)}(\lb)$   satisfies the jump condition $Z^{(p)}_+(\lb)=Z^{(p)}_-(\lb) J^{(p)}(\lb)$, where  {\small$$J^{(p)}(\lb)=\left\{\begin{array}{ll}
						\begin{pmatrix}
						1&-\ii r(\lb)^{-1}\delta(\lb)^2 \mathrm{e}^{2tg(\lb)}\\
						0&1\\
					\end{pmatrix},&\lb\in\Gamma_{1}\cap U(p),\\
					\begin{pmatrix}
						1&0\\
						-\ii r(\bar{\lb}^{-1})^{-1}\delta(\lb)^{-2} \mathrm{e}^{-2tg(\lb)}&1\\
					\end{pmatrix},&\lb\in\Gamma_{2}\cap U(p),\\
					\begin{pmatrix}
						0&\ii  \\
						\ii &0\\
					\end{pmatrix},\quad &\lb\in(\ii(\eta_{1},\lb_0(\xi))\cup\ii(\lb_0(\xi)^{-1},\eta_{1}^{-1}))\cap U(p),\\
					\mathrm{e}^{\ii (t\Omega+\Delta)\sigma_{3}},\quad&\lb\in\ii(\eta_{1}^{-1},\eta_{1})\cap U(p),\\
					\begin{pmatrix}
						r(\lb)^{-1}\delta_-(\lb)^{-2}\e^{2tg_-(\lb)}&0\\
						\ii &r(\lb)^{-1}\delta_+(\lb)^{-2}\e^{2tg_+(\lb)}\\
					\end{pmatrix},&\lb\in\ii (\lb_0(\xi),\eta_{2})\cap U(p),\\
					\begin{pmatrix}
					r(\bar{\lb}^{-1})^{-1}\delta_+(\lb)^{-2}\e^{-2tg_+(\lb)}&\ii \\
						0&	r(\bar{\lb}^{-1})^{-1}\delta_-(\lb)^{-2}\e^{-2tg_-(\lb)}\\
					\end{pmatrix},&\lb\in\ii (\eta_{2}^{-1},\lb_0(\xi)^{-1})\cap U(p),
					\end{array}\right.$$}
                    with $\lb_0(\xi)$ given in \eqref{p0}.
					\item As $t\to+\infty$, $Z^{(p)}(\lb)$ matches $Z^{(\infty)}(\lb)$ on the boundary $\partial U(p)$ of $U(p)$.
				\end{itemize}
			\end{Rhp}
The above RH problem  can be solved explicitly with the aid of the Painlev\'e XXXIV parametrix introduced in Appendix \ref{app4}.  We give a sketch of the construction in what follows. 	
				
From the definition of $g$ given in \eqref{dg}, it is readily seen that,  as $\lb \to \ii \eta_2$,
				\begin{align}\label{eq:gexpan}
				    g(\lb) = s_0 d_0^{\frac{1}{2}}  (\xi-\xi_{\mathrm{crit}}) \E^{-\frac{\pi \ii}{4}} (\lb-\ii\eta_2)^{\frac{1}{2}}+\frac{2}{3}d_0^{\frac{3}{2}}\E^{-\frac{3\pi\ii}{4}}(\lb-\ii\eta_2)^{\frac{3}{2}}+\oo((\xi-\xi_{\mathrm{crit}})(\lb-\ii\eta_2)^{\frac{3}{2}}),
				\end{align}
				where
				\begin{align}
				   & d_0 =\frac{2}{3} \left( \frac{P(\ii \eta_2)}{2\eta_2^2 \sqrt{(\eta_2-\eta_1)(\eta_2-\eta_1^{-1})(\eta_2-\eta_2^{-1})}}\right)^{\frac{2}{3}},\\
				   & s_0= \frac{2 d_0 \eta_2^2}{P(\ii \eta_2)} \left(\eta_2 +\eta_2^{-1} -2 \left(\frac{2\Pi(l_1^2,k)}{K(k)}-1 \right)\right), \\
				   &   P(\lb) = \ii \lb^3 + \left( c_1(\xi_{\mathrm{crit}}) -\eta_2\right) \lb^2 + \left[c_0(\xi_{\mathrm{crit}}) + \ii \eta_2 \left( c_1(\xi_{\mathrm{crit}}) -\eta_2\right)\right] \lb -\eta_2^{-1},
				\end{align}
                and where 
                \begin{align*}
                   l_1=\frac{\eta_1-1}{\eta_1+1},\qquad k=\frac{(\eta_1-1)(\eta_2+1)}{(\eta_1+1)(\eta_2-1)}, 
                \end{align*}
                 $c_0$ and $c_1$ are given in \eqref{c1c0}. 	We emphasize that $P(\ii \eta_2)\in \mathbb{R}$, as do $d_0$ and $s_0$.
%for $\lb \in \ii\mathbb{R}$, the function $P(\lb)$ is real. With the branch of $R(\lb)$ chosen in \eqref{def R}, it follows that $P(\ii \eta_2)<0$.
			As a consequence of \eqref{eq:gexpan},
%Therefore, we define  a local conformal map in the neighborhood $U(p)$ as 
\begin{align}
						\zeta=\zeta(\lb):=\left(\frac{3t}{2}g(\lb;\xi_{\mathrm{crit}})  \right)^{\frac{2}{3}}
                        %\quad p=\ii \eta_{2}^{\pm 1}.
                        \label{def zeta3}
					\end{align}
defines a local conformal map in $U(\ii \eta_2)$, which satisfies					
%For $p= \ii \eta_2$, by \eqref{dg}, we obtain that $\zeta$ is analytic in $U(\ii \eta_2)$, and
					\begin{align*}
					    \zeta(\i \eta_2)=0, \qquad \zeta'(\i \eta_2)= -\ii t^{\frac{2}{3}}d_0.
					\end{align*}
Under the mapping \eqref{def zeta3}, the jump contours of $Z^{(\ii \eta_2)}$ is mapped from the left picture in Figure \ref{scalingzk2} to the right picture  in Figure \ref{scalingzk2}.
					
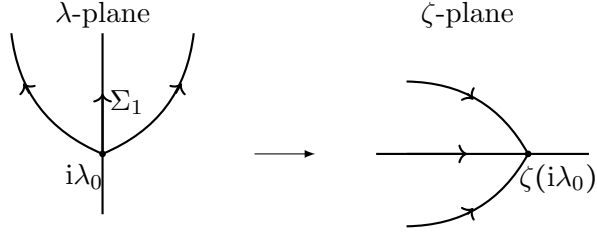
\begin{figure}[t]
 \begin{center}
   \begin{tikzpicture}[scale=0.8]
]
   % \draw[dotted,thick,teal](-3,0) circle (2);
        \node[shape=circle,fill=black, scale=0.13]  at (-3,0){0};
        \node[below] at (-3.3,0) { $\ii \lb_0$};
           \draw[thick] (-3,-1) --(-3,2);
         \draw[->,thick] (-3,0) --(-3,1);
        \node  at (-2.6,0.86) { $\Sigma_{1}$};
        \node[]  at (-3,2.3) { $\lb$-plane};
       \draw[thick]
    (-3,0) to[bend right=30] (-1.5,2);
 \draw[thick,->] (-1.8,1)--(-1.7,1.17);
    
  \draw[thick]
    (-3,0) to[bend left=30]  (-4.5,2);
 \draw[thick,->] (-4.2,1)--(-4.3,1.17);
 
         \draw[thick]
    (2,-1.2) to[bend right=30] (4,0);
    \draw[thick,->] (3,-1)--(3.1,-0.9);
  \draw[thick]
    (2,1.2) to[bend left=30]  (4,0);
\draw[thick,->] (3,1)--(3.1,0.9);

    \draw[thick] (4,0)--(5,0);

        \node[shape=circle,fill=black, scale=0.13]  at (4,0){0};
        \node[below] at (4.5,0) { $\zeta(\ii\lb_0)$};
           \draw[thick] (1.5,0) --(4,0);
         \draw[thick,->] (2,0) --(3,0);

        \node[]  at (3,2.3) { $\zeta$-plane};

 \draw [-latex] (-0.5,0)--(0.5,0);
\

    \end{tikzpicture}
    \caption{\footnotesize{ The jump contours of $Z^{(\ii \eta_2)}$ from the $\lb$-plane (left) to the $\zeta$-plane (right) under the mapping \eqref{def zeta3}.}}
      \label{scalingzk2}
  \end{center}
\end{figure}
					
To proceed, we define			
				\begin{align}
	\tilde{s}(\lb;\xi):=t g(\lb)\zeta(\lb)^{-\frac{1}{2}}-\frac{2}{3}\zeta(\lb),
\end{align}
which is analytic in  in $U(\i \eta_2)$ and set
		\begin{align}
		    \tilde{s}_0:= \tilde{s}(\ii\eta_2;\xi) = s_0 t^{\frac{2}{3}} (\xi-\xi_{\mathrm{crit}}).
		\end{align}
For $\lb \in U(\i \eta_2)$, the local parametrix is given by
					\begin{align}\label{1tt22}
						Z^{(\i \eta_2)}(\lb)=A(\lb) \sigma_1 \begin{pmatrix}
						 1&0\\
						 \ii a(\tilde{s}_0) &1
						\end{pmatrix}
						M^{\rm{P}_{34}}(\zeta(\lb))G^{(\i \eta_2)}(\zeta(\lb)),
					\end{align}
					where
					\begin{align*}
					    A(\lb)  = Z^{(\infty)}(\lb)  r(\lb)^{-\frac{\sigma_3}{2}} \delta(\lb)^{\sigma_3} \E^{tg(\lb) \sigma_3} \E^{-\frac{\pi \ii}{4} \sigma_3} \sigma_1 G_2(\zeta(\lb) )^{-1} \E^{\left(\frac{2}{3}\zeta^{3/2} +\tilde{s}_0 \zeta^{1/2}\right)\sigma_3} \frac{1}{\sqrt{2}} \begin{pmatrix}
					1	&	-\ii \\
						-	\ii &1\\
						\end{pmatrix} \zeta^{\frac{\sigma_3}{4}}\sigma_1,
					\end{align*}
                    $a(s)$ defined in \eqref{eq:Psi-infinity} is related to the Painlev\'{e} XXXIV equaiton \eqref{def:P34} through the transformation \eqref{def:u} and
					\begin{align*}
					    G^{(\i \eta_2)}(\zeta(\lb)) = G_2(\zeta(\lb)) \sigma_1  \E^{\frac{\pi \ii}{4}\sigma_3} \E^{-tg(\lb) \sigma_3}  \delta(\lb)^{-\sigma_3}r(\lb)^{\frac{\sigma_3}{2}}
                        \end{align*}
                        with
                        \begin{align*}
					     G_2(\zeta(\lb)) = \begin{cases}
					       \begin{pmatrix}
					1	&	0\\
						1	&1\\
						\end{pmatrix}, \qquad &\lb \in \tilde{\Omega}_1,\\
						\begin{pmatrix}
					1	&	0\\
						-1	&1\\
						\end{pmatrix}, \qquad &\lb \in \tilde{\Omega}_2.
				    \end{cases} 
                    %\quad \tilde{\Omega}_1, \tilde{\Omega}_2 \ \text{are shown  in Figure \ref{scalingzk3}. }
					\end{align*}
Here, the regions $\tilde{\Omega}_1$ and $ \tilde{\Omega}_2$ are shown in Figure \ref{scalingzk3},
                    $M^{\rm{P}_{34}}(\zeta)=M^{\rm{P}_{34}}(\zeta; 0, 0, \tilde{s}_0)$ is the Painlev\'e XXXIV parametrix given in Appendix \ref{app4}. Moreover, similar to the previous case, it can be  verified directly that $A$ is analytic in  $U(\i \eta_2)$.  	For  $\lb\in U(\ii\eta_{2}^{-1})$,   the local parametrix is given by the symmetry relation
				\begin{align}\label{1tt222}
	Z^{(\i \eta_2^{-1})}(\lb)=\sigma_2\overline{(Z^{(\infty)}(0))^{-1}Z^{(\i \eta_2)}(\bar\lb^{-1})}\sigma_2.
\end{align}

					\begin{figure}[h]
 \begin{center}
   \begin{tikzpicture}[scale=1]
]
             \draw[dotted, thick]
    (7,-1.2) to[bend right=30] (8.5,0);
 \draw[dotted,thick,->] (7.8,-0.9)--(8,-0.75);

  \draw[dotted,thick]
    (7,1.2) to[bend left=30]  (8.5,0);
     \draw[dotted,thick,->] (7.8,0.9)--(8,0.75);
     
    \draw[thick] (8.5,1.6) --(10,0);
    \draw[thick,->] (9.2,0.853)--(9.29,0.79);
    
    \draw[thick] (8.5,-1.6) --(10,0);
   \draw[thick,->] (9.2,-0.853)--(9.29,-0.79);
    
     \node  at (9,0.4) { $\tilde{\Omega}_1$};
      \node  at (9,-0.4) { $\tilde{\Omega}_2$};

        \node[shape=circle,fill=black, scale=0.13]  at (10,0){0};
        \node[below] at (10,0) { $0$};
         \draw[thick] (7,0) --(10,0);
         
           \node[shape=circle,fill=black, scale=0.13]  at (8.5,0){0};
        \node[below] at (7.8,0) { $\zeta(\ii\lb_0)$};

  \draw[thick, ->] (8.5,0)--(9,0);
  
    \end{tikzpicture}
    \caption{\footnotesize{  Regions $\tilde{\Omega}_1$ and $\tilde{\Omega}_2$ in $\zeta$-plane and the jump contours of $M^{\rm{P}_{34}}$ (solid lines).  }}
      \label{scalingzk3}
  \end{center}
\end{figure}
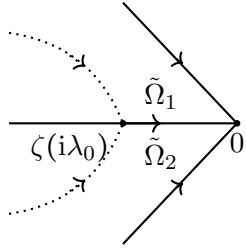

In view of RH problem \ref{RHp34} for $M^{\rm{P}_{34}}$, one can check that $Z^{(p)}$ defined in \eqref{1tt22} and  \eqref{1tt222} solves RH problem \ref{Rhp locp34} with $p=	\ii\eta_2^{\pm1}$.	In addition, we have, as $t\to +\infty$,                    % Moreover, for the endpoints $	p=	\ii \eta_{2}^{\pm1}$,  we analogously obtain that $Z^{(p)}$ satisfies the asymptotic behavior for $\lb\in \partial U(p)$  as
					\begin{align*}					
						Z^{(p)}(\lb)Z^{(\infty)}(\lb)^{-1}=I+\oo(\zeta^{-\frac{1}{2}}),\qquad \lb\in \partial U(p),
					\end{align*}
which, by \eqref{def zeta3}, also implies that
					\begin{align}\label{asy:loc y3}
						Z^{(p)}(\lb)Z^{(\infty)}(\lb)^{-1}=I+\oo(t^{-\frac{1}{3}}),\qquad \lb\in \partial U(p).
					\end{align}

%Finally, for $p=\i \eta_1^{\pm 1}, \i \eta_2^{\pm2}$, using \eqref{asy:loc y} and \eqref{asy:loc y3},  we obtain 
 %\begin{align}
%						Z^{(p)}(\lb)Z^{(\infty)}(\lb)^{-1}=I+\oo(t^{-\frac{1}{3}}).
%					\end{align}

					\subsection{The small-norm RH problem}% the potential $u(x)$}
				\label{sebsec E}
				
				%Consider the following ``remainder" Riemann--Hilbert  problem:
				%\section{RHP for Error}
				%In view of RH problems for $Z^{(1)}$, $Z^{(\infty)}$  and  $Z^{(p)}$, it is readily seen that the error function satisfies (see \eqref{def:err})
From  \eqref{def:err}, it follows that 
                \begin{align}\label{defEc2}
					E(\lb)=\left\{\begin{array}{ll}
						Z^{(1)}(\lb)Z^{(\infty)}(\lb)^{-1}, &\lb\in \mathbb{C}\setminus U,\\
						Z^{(1)}(\lb)Z^{(p)}(\lb)^{-1},\quad &\lb\in   U(p).\\
					\end{array}\right.
				\end{align}
				In what follows, we estimate the error function  $E$ for $\xi$ belonging to different regions.
				%The error function  $E(\lb)$ is analytic on $\mathbb{C}\setminus\Gamma^{(E)}$, denote the anti-clockwise oriented boundary of $ U(p)$ as $\partial U(p)$ for $p=\ii\eta_{1}^{\pm1},\ \ii\alpha(\xi)^{\pm1}, \ \ii \eta_{2}^{\pm 1}$, $\Gamma^{(E)}$  is given by:
			
				\paragraph{Estimate of $E$ for $\xi \in T_{I}$}
				On account of the constructions of  local parametrices for $\xi \in T_I$ above, it suffices to estimate the error function separately in each subregion $T_I^{(m)}$ with $m$ fixed. 	Combining RH problem \ref{Rhp glo} for $Z^{(\infty)}$ with \eqref{defZ1T1},  we obtain the following RH problem for $E$.
    %             Denote
				% \begin{align}
				%     	\Gamma^{(E)}: = \left((\Sigma_1\cup\Sigma_2)\setminus U \right)\cup \left(\bigcup_{p=\ii\eta_{1}^{\pm1}}\partial U(p) \right).
				% \end{align}

	\begin{Rhp} \hfill
					\begin{itemize}
						\item $E(\lb)$ is analytic in  $\mathbb{C}\setminus\Gamma^{(E)}$, where 
                        \begin{align}
				    	\Gamma^{(E)}: = \left((\Sigma_1\cup\Sigma_2)\setminus U \right)\cup \left(\bigcup_{p=\ii\eta_{1}^{\pm1}}\partial U(p) \right).
				\end{align}
						\item   $E(\lb)$ satisfies the jump condition $E_+(\lb)=E_-(\lb)J^{(E)}(\lb)$, where $$J^{(E)}(\lb)=\left\{\begin{array}{ll}
                            Z^{(\infty)}(\lb)J(\lb)Z^{(\infty)}(\lb)^{-1}, & \lb \in (\Sigma_1\cup\Sigma_2)\setminus U,\\
						Z^{(p)}(\lb)Z^{(\infty)}(\lb)^{-1},& \lb\in  \partial U(p), \; p=\i \eta_1^{\pm 1}.
						\end{array}\right.$$
						\item  As $\lb\to\infty$, we have $E(\lb)=I+\oo(\lb^{-1})$.
					\end{itemize}
				\end{Rhp}		
		
			%Denote the above  jump condition of $E(\lb)$ as $J^{(E)}(\lb)$.
			% From  \eqref{asy:loc y0}, the jump matrix on  $\Gamma^{(E)}$  decays to the identity matrix as $t\to+\infty$ as follows:
			% for 
          As $t\to +\infty$, we have the following estimate of the jump matrix  on  $\Gamma^{(E)}$:
			$$J^{(E)}(\lb)-I = \oo(\E^{-ct}),$$
	for $\lb\in(\Sigma_1\cup\Sigma_2)\setminus U$,
        where $c$ is a positive constant, and by \eqref{asy:loc y0},
			\begin{align}
			    & J^{(E)}(\lb)-I
                \nonumber 
                \\
                &=\begin{cases}
			       	\oo (\min( t^{-1}, \E^{(2m-1)\ln t +  t(\xi+\frac{\eta_1-\eta_1^{-1}}{\ln \eta_1})\ln \eta_1},\E^{-(2m+1)\ln t - t(\xi+\frac{\eta_1-\eta_1^{-1}}{\ln \eta_1})\ln \eta_1} )),\quad m\ge1,\\
			       	\oo (\min( \E^{-\ln t +  t(\xi+\frac{\eta_1-\eta_1^{-1}}{\ln \eta_1})\ln \eta_1},\E^{-\ln t - t(\xi+\frac{\eta_1-\eta_1^{-1}}{\ln \eta_1})\ln \eta_1} )),\quad m=0,
			    \end{cases}
			\end{align}
	for $\lb\in \bigcup_{p=\ii\eta_{1}^{\pm1}}\partial U(p)$. We then obtain from the standard small-norm RH problem argument \cite{RN10} that
			\begin{align}\label{asy ET1}
		E(\lb)=\begin{cases}
		I+	\oo (\min( t^{-1}, \E^{(2m-1)\ln t +  t(\xi+\frac{\eta_1-\eta_1^{-1}}{\ln \eta_1})\ln \eta_1}, \E^{-(2m+1)\ln t - t(\xi+\frac{\eta_1-\eta_1^{-1}}{\ln \eta_1})\ln \eta_1} )),\quad m\ge 1, \\
		I+\oo( \min(  \E^{-\ln t+t(\xi +\frac{\eta_1-\eta_1^{-1}}{\ln \eta_1} ) \ln \eta_1},  \E^{-\ln t -t(\xi +\frac{\eta_1-\eta_1^{-1}}{\ln \eta_1} ) \ln \eta_1} )), \quad m=0,
		\end{cases}
		\end{align}				
for large positive $t$.
% Therefore, by the standard small-norm RH problem argument \cite{RN10}, 	we have that as $t\to+\infty$,	
% 		\begin{align}\label{asy ET1}
% 		E(\lb)=\begin{cases}
% 		I+	\oo (\min( t^{-1}, \E^{(2m-1)\ln t +  t(\xi+\frac{\eta_1-\eta_1^{-1}}{\ln \eta_1})\ln \eta_1}, \E^{-(2m+1)\ln t - t(\xi+\frac{\eta_1-\eta_1^{-1}}{\ln \eta_1})\ln \eta_1} )),\quad m\ge 1, \\
% 		I+\oo( \min(  \E^{-\ln t+t(\xi +\frac{\eta_1-\eta_1^{-1}}{\ln \eta_1} ) \ln \eta_1},  \E^{-\ln t -t(\xi +\frac{\eta_1-\eta_1^{-1}}{\ln \eta_1} ) \ln \eta_1} )), \quad m=0.
% 		\end{cases}
% 		\end{align}				
				
\paragraph{Estimate of $E$ for $\xi \in H_{I} \cup H_{II}$}

In view of RH problems \ref{Rhp 2} and \ref{Rhp glo} for $Z^{(1)}$ and  $Z^{(\infty)}$, we obtain from \eqref{defEc2} the following RH problem for $E$.

\begin{figure}
				\centering
				\tikzset{every picture/.style={line width=0.75pt}}
				\begin{tikzpicture}
					\draw [->] (0,0.88) -- (0.03,0.88);
					\draw [->] (0,2.165) -- (0.03,2.165);
					\draw [->] (0,4.02) -- (0.03,4.02);
					\draw [->] (-0.43,5.5) -- (-0.429,5.502);
					\draw [->] (-1,3.97) -- (-1.01,3.975);
					\draw [->] (1,3.97) -- (1.01,3.975);
					\draw [->] (-0.89,1.4) -- (-0.884,1.408);
					\draw [->] (0.89,1.4) -- (0.884,1.408);
					\draw (0,0) -- (0,0.35);
					\draw (0,5.76) -- (0,6);
					\draw [domain=3.77:5.78] plot ({(\x-4.8)*(\x-4.8)-1.45},{(\x-3.6)*0.7+3.6});
					\draw [domain=3.77:5.78] plot ({-(\x-4.8)*(\x-4.8)+1.45},{(\x-3.6)*0.7+3.6});
					\draw [domain=0.15:1.81] plot ({(\x-1)*(\x-1)-1},{0.6*\x+0.61});
					\draw [domain=0.15:1.81] plot ({-(\x-1)*(\x-1)+1},{0.6*\x+0.61});
					\filldraw (0,0) circle (1pt);
					\filldraw (0,0.61) circle (1pt);
					\draw (0,0.61) circle (8pt);
					\filldraw (0,1.8) circle (1pt);
					\draw (0,1.8) circle (10pt);
					\filldraw (0,3.6) circle (1pt);
					\draw (0,3.6) circle (12pt);
					\filldraw (0,5.28) circle (1pt);
					\draw (0,5.28) circle (14pt);
					\filldraw (0,6) circle (1pt);
					\node at (-0.3,-0.2) {$\ii\eta_2^{-1}$};
					\node at (-0.6,2) {$\ii\eta_1^{-1}$};
					\node at (-0.85,0.5) {$\ii\alpha(\xi)^{-1}$};
					\node at (0.7,3.5) {$\ii\eta_1$};
					\node at (-0.3,6.1) {$\ii\eta_2$};
					\node at (1,5.2) {$\ii\alpha(\xi)$};
				\end{tikzpicture}
				\caption{ \footnotesize The jump contours $\Gamma^{(E)}$.} \label{curve:E}	
			\end{figure}
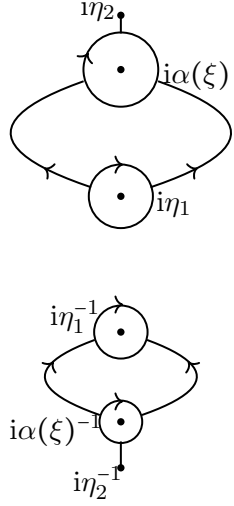
			\begin{Rhp} \label{E}\hfill
					\begin{itemize}
						\item $E(\lb)$ is  analytic in $\mathbb{C}\setminus\Gamma^{(E)}$,
                        where
                        \begin{align*}
\Gamma^{(E)}:=\left((\Gamma_{1}\cup\Gamma_{2}\cup\ii(\eta_2^{-1},\alpha(\xi)^{-1})\cup\ii(\eta_2,\alpha(\xi)))\setminus U\right)\cup\left(  \bigcup_{p=\ii\eta_{1}^{\pm1},\ \ii\alpha(\xi)^{\pm1}}\partial U(p)\right);
				\end{align*}
                see Figure \ref{curve:E}  for an illustration. 
	
						\item   $E(\lb)$ satisfies the jump condition $E_+(\lb)=E_-(\lb)J^{(E)}(\lb)$, where $$J^{(E)}(\lb)=\left\{\begin{array}{ll}
							\begin{pmatrix}
							1&-\ii r(\lb)^{-1}\delta(\lb)^{2} \mathrm{e}^{2tg(\lb)}\\
							0&1\\
			 \end{pmatrix},&\lb\in\Gamma_{1}\setminus U,\\
						\begin{pmatrix}
							1&0\\
							-\ii r(\bar{\lb}^{-1})^{-1}\delta(\lb)^{-2} \mathrm{e}^{-2tg(\lb)}&1\\
						\end{pmatrix},&\lb\in\Gamma_{2}\setminus U,\\				
						\begin{pmatrix}
							1&0\\
							\i r(\lb)\delta(\lb)^{-2}\e^{-2tg(\lb)}&1\\
						\end{pmatrix},&\lb\in\ii (\alpha(\xi),\eta_{2})\setminus U,\\
						\begin{pmatrix}
							1&\i r(\bar{\lb}^{-1})\delta(\lb)^{2}\e^{2tg(\lb)}\\
							0&1\\
						\end{pmatrix},&\lb\in\ii (\eta_{2}^{-1},\alpha(\xi)^{-1})\setminus U,\\
							Z^{(p)}(\lb)Z^{(\infty)}(\lb)^{-1},& \lb\in \partial U(p), \; p=\ii\eta_{1}^{\pm1},\; \ii\alpha(\xi)^{\pm1}.
						\end{array}\right.$$
						\item  As $\lb\to\infty$, we have $E(\lb)=I+\oo(\lb^{-1})$.
					\end{itemize}
				\end{Rhp}
				
				%	For $\xi \in H_{I} \cup H_{II}$,	denote the above  jump condition of $E(\lb)$ as $J^{(E)}(\lb)$.
				Thanks to the estimate \eqref{asy:loc y2}, as $t\to+\infty$, we have
				\begin{align}\label{err E}
					J^{(E)}(\lb)-I=\left\{\begin{array}{lll}
				\oo(\mathrm{e}^{-ct}), \quad& \lb\in(\Gamma_{1}\cup\Gamma_{2}\cup\ii(\eta_2^{-1},\alpha(\xi)^{-1})\cup\ii(\eta_2,\alpha(\xi)))\setminus U,\\ 			
						\oo(t^{-1}), \quad&\lb\in \bigcup_{p=\ii\eta_{1}^{\pm1},\ \ii\alpha(\xi)^{\pm1}} \partial U(p),
					\end{array}\right.
				\end{align}
				where $c$ is a positive constant.
We then conclude from the standard small-norm RH problem argument \cite{RN10} that $t\to+\infty$ that	
		\begin{align}\label{asy E}
		E(\lb)=I+	\oo(t^{-1}),
		\end{align}					
		for large positive $t$.
        
	\paragraph{Estimate of $E$ for $\xi \in  T_{II}$}
    In this case, $E$ still satisfies RH problem  \ref{E}
    but with
				\begin{align}
					\Gamma^{(E)}:=
					\left((\Gamma_{1}\cup\Gamma_{2})\setminus U\right)\cup\left(  \bigcup_{p=\ii\eta_{1}^{\pm1},\ \ii\eta_{2}^{\pm1}}\partial U(p)\right),
				\end{align}
    and $Z^{(p)}$ changed accordingly. 
	% In this case, the RH problem for $E$ yields RH problem \ref{E} with $\alpha(\xi)=\eta_2$ and 
  %     $Z^{(p)}$ established in Subsection \ref{SubsecP34loc}.
		By the estimates \eqref{eq:asyBoundary} and \eqref{asy:loc y3}, one has that as $t \to +\infty$,
				 \begin{align}\label{err ET2}
					J^{(E)}(\lb)-I=\left\{\begin{array}{lll}
				\oo(\mathrm{e}^{-ct}), \quad& \lb\in(\Gamma_{1}\cup\Gamma_{2})\setminus U,\\ 
						\oo(t^{-1}), \quad&\lb\in \bigcup_{p=\ii\eta_{1}^{\pm1}} \partial U(p),\\
							\oo(t^{-1/3}), \quad&\lb\in \bigcup_{p=\ii\eta_{2}^{\pm1}} \partial U(p),
					\end{array}\right.
			\end{align}
				where $c$ is a positive constant.
This leads to 			
		\begin{align}\label{asy ET2}
		E(\lb)=I+\oo(t^{-1/3}),
		\end{align}						
for large positive $t$.

\section{Proofs of the main results}\label{sec:proof}

\subsection{Proof of Theorem \ref{Thm 1}}
As $n\to+\infty$, the jump matrix of $Z(\lb;n,0)$ tends to the identity matrix exponentially fast, which implies that
\begin{equation}
    Z(\lb;n,0)=I+\mathcal{O}(\E^{-cn})
\end{equation}
for some positive constant $c$. This, together with \eqref{rec}, gives us \eqref{eq:qnasy}.

% then it follows that
% \begin{align*}
%     q_n(0)=\mathcal{O}(\E^{-cn}),
% \end{align*}
% for some positive constant $c$.
As $n\to -\infty$, by tracking back the transformations \eqref{T1}, \eqref{T2} and \eqref{T3}, we obtain that near $\lb=0$,
\begin{align}\label{eTrans}
Z(\lambda;n+1,0)=\E^{-(n+1)g^{(\infty)}\sigma_3}\delta(\infty)^{\sigma_3}E(\lambda)Z^{(\infty)}(\lambda)\delta(\lambda)^{-\sigma_3}\E^{(n+1)(g(\lb)-\frac{\ln\lb}{2})\sigma_3},
\end{align}
where 
%$g^{(\infty)}$, $\delta(\infty)$, $\delta$ and $g$ are defined in Subsection \ref{aux1}, and 
$E$ and $Z^{(\infty)}$ are given in \eqref{defEc1} and \eqref{zinfty}, respectively.
%where $g^{(\infty)}$, $\delta(\infty)$, $E$, $Z^{\infty}$, $\delta$ and $g$ are given in 
From the symmetry relations \eqref{e3.6} and \eqref{edeltasym}, it follows that
\begin{align*}
	g^{(\infty)}+\lim_{\lb\to0}\left(g(\lb)-\frac{\ln\lb}{2}\right)=-\frac{\pi\ii}{2},\qquad \delta(\infty)\delta(0)=1.
\end{align*}
A combination of \eqref{eTrans}, \eqref{rec} and the above formulas yields 
\begin{align}\label{eqn}
	q_n(0)=\ii^{n+1}\left(E(0;n+1)Z^{(\infty)}(0;n+1)\right)_{12}.
\end{align}
By \eqref{def:Zinfty12}, we have
\begin{align}\label{1nqn0}
	Z_{12}^{(\infty)}(0;n+1)=\frac{1}{2}\left(\sqrt{\frac{\eta_2}{\eta_1}}-\sqrt{\frac{\eta_1}{\eta_2}}\right)   \Jac\left(\frac{K(k)((n+1)\Omega+\Delta)}{\pi},k\right),
\end{align}
where $\Jac( z,k)$ is the subsidiary  Jacobi elliptic function with modulus $ k=\frac{(\eta_1-1)(\eta_2+1)}{(\eta_2-1)(\eta_1+1)}$ and
\begin{align}\label{ndjac}
\Jac(2K(k) z,k)=\frac{\vartheta(0,\tau)\vartheta(z+\frac{1}{2},\tau)}{\vartheta(\frac{1}{2},\tau)\vartheta(z,\tau)},
\end{align}
$K$ is the complete elliptic integral defined in  \eqref{ei1}, and  $\Omega$ and $\Delta$ are constants 
given by  \eqref{eOmega} and \eqref{delta1}, respectively.
%\begin{align*}
%   \Jac(\frac{K(k)((n+1)\Omega+\Delta)}{\pi},k)=  \frac{\vartheta(\frac{1}{2}+\frac{(n+1)\Omega+\Delta}{2\pi},\tau)\vartheta(0,\tau)}{\vartheta(\frac{1}{2},\tau)\vartheta(\frac{(n+1)\Omega+\Delta}{2\pi},\tau)}.
%\end{align*}
%It is noticed from \cite{dlmf} that
%\begin{align}\label{ndjac}
%    \vartheta_4(z,\tau)=\vartheta(z+\frac{1}{2},\tau),\qquad \Jac(2K(k) z,k)=\frac{\vartheta(0,\tau)\vartheta_4(z,\tau)}{\vartheta_4(0,\tau)\vartheta(z,\tau)},
%\end{align}

%Then we have
%\begin{align}\label{1nqn0}
%	q_n^{(\infty)}(0)&=\frac{1}{2}\left(\sqrt{\frac{\eta_2}{\eta_1}}-\sqrt{\frac{\eta_1}{\eta_2}}\right)\frac{\vartheta_4(\frac{(n+1)\Omega+\Delta}{2\pi},\tau)\vartheta(0,\tau)}{\vartheta_4(0,\tau)\vartheta(\frac{(n+1)\Omega+\Delta}{2\pi},\tau)}\nonumber\\
%	&=\left(\sqrt{\frac{\eta_2}{\eta_1}}-\sqrt{\frac{\eta_1}{\eta_2}}\right)\frac{\Jac(\frac{K(k)((n+1)\Omega+\Delta)}{\pi},k)}{2}.
%\end{align}
Using \eqref{eerror} and \eqref{eqn},
we obtain from \eqref{rec} that
\begin{align*}
	q_n(0)=\ii^{n+1}Z_{12}^{(\infty)}(0;n+1)+\oo(n^{-1}),
\end{align*}
which is \eqref{eq:qn0negn} by  \eqref{1nqn0}.

This completes the proof of Theorem \ref{Thm 1}.
\qed

\subsection{Proof of Theorem \ref{Thm 2}}	
\paragraph{Proof of \eqref{eq:decay}}
For 
$
 \xi >-\frac{\eta_{1}-\eta_{1}^{-1}}{\ln\eta_{1}},
$
% we derive that when
%  \begin{align*}
% 	\xi>-\frac{\eta_{1}-\eta_{1}^{-1}}{\ln\eta_{1}},
% \end{align*}
the jump matrix of $Z(\lb;n+1,t)$ tends to the identity matrix exponentially fast as $t\to+\infty$, which implies that
\begin{equation}
    Z(\lb;n+1,t)=I+\mathcal{O}(\E^{-ct}), \qquad t\to +\infty,
\end{equation}
for some positive constant $c$. This, together with \eqref{rec}, gives us \eqref{eq:decay}.

\paragraph{Proof of \eqref{mr1}}
For $\xi \in T_{I}^{(m)}$ with $m \in \{0\} \cup \mathbb{N}$, by tracking back the transformations  \eqref{def:err} and \eqref{defZ1T1}, we conclude that for $\lb\in \mathbb{C}\setminus U$ with $U=\bigcup_{p=\ii\eta_1^{\pm1}}U(p)$, the solution to  RH problem \ref{RHP0} is given by
\begin{align*}
	Z(\lb;n+1,t)=		E(\lb)Z^{(\infty)}(\lb),
\end{align*}
where $E$ and $Z^{(\infty)}$ are given by \eqref{defEc2} and \eqref{T1glo} respectively.
Substituting  $Z$ into the reconstruction formula \eqref{rec}, we obtain \eqref{mr1} from \eqref{T1glo} and the estimate \eqref{asy ET1}.

% via %explicit expressions \eqref{T1glo} and
% estimate \eqref{asy ET1}, we obtain 

% the asymptotic expression as $t\to+\infty$:
% \begin{align*}
% 		q_n(t)=	\oo (\min( \E^{(2m-1)\ln t +  t(\xi+\frac{\eta_1-\eta_1^{-1}}{\ln \eta_1})\ln \eta_1},\E^{-(2m+1)\ln t - t(\xi+\frac{\eta_1-\eta_1^{-1}}{\ln \eta_1})\ln \eta_1} )).
% \end{align*}

\paragraph{Proof of \eqref{mr2} and \eqref{mr4}}
For $\xi \in H_{I} \cup H_{II}$, by tracking back the transformations \eqref{trans1} and \eqref{def:err}, we conclude that for $\lb\in \mathbb{C}\setminus U$ with $U$ defined in \eqref{defun} (where recall that $\alpha(\xi)=\eta_2$ for $\xi \in H_{II}$), the solution to  RH problem \ref{RHP0} is given by
\begin{align}\label{zlb}
	Z(\lb;n+1,t)=	\delta(\infty)^{-\sigma_{3}}\E^{-tg^{(\infty)}\sigma_3}	E(\lb)Z^{(\infty)}(\lb)\E^{ (tg(\lb)-\phi(\lb)/2)\sigma_3}G(\lb)^{-1}\delta(\lb)^{\sigma_{3}},
\end{align}
where  $g^{(\infty)}$, $\delta(\infty)$, $\delta$ and $g$ can be found in Section \ref{aux2}, and $G$, $E$  and $Z^{\infty}$ are given by \eqref{dg},  \eqref{defEc2} and \eqref{zinftyc2}, respectively.
From the relations \eqref{g0}, \eqref{Imginf}, and \eqref{edeltasym2}, it follows that 
\begin{align*}
	g^{(\infty)}+\lim_{\lb\to0}\left(g(\lb)-\frac{\ln\lb}{2}\right)=\E^{2\ii t(1+\frac{\pi(n+1)}{4t})},\qquad \delta(\infty)\delta(0)=1.
\end{align*}
A combination of \eqref{zlb}, \eqref{rec} and the above formulas yields
\begin{align}\label{qnt2}
		q_n(t)=\E^{2\ii t(1+\frac{\pi(n+1)}{4t})} \left( E(0;n+1,t)Z^{(\infty)}(0;n+1,t)\right)_{12}.
\end{align}
By \eqref{eq:Zinf12}, we have
\begin{align}\label{z12i02}
    Z^{(\infty)}_{12}(0;n+1,t) = \frac{1}{2}\left(\sqrt{\frac{\alpha(\xi)}{\eta_1}}-\sqrt{\frac{\eta_1}{\alpha(\xi)}}\right)	\Jac\left(\frac{K(k(\xi))(t\Omega+\Delta)}{\pi},k(\xi)\right),
\end{align}
where $k(\xi)$ is given in \eqref{equ xialpha2}, $K$ denotes the complete elliptic integral  defined in \eqref{ei1},
$\Jac$ is the Jacobi elliptic function introduced in \eqref{ndjac} ,  and $\Omega$ and $\Delta$ are constants given by \eqref{def Omegj} and \eqref{def Delta}, respectively.
Using \eqref{asy E} and \eqref{z12i02}, \eqref{qnt2}  gives us \eqref{mr2} and \eqref{mr4}. 

%On account of the reconstruction formula \eqref{rec}, we obtain from the estimate \eqref{asy E} that
%\begin{align*}
%		q_n(t)=\E^{-tg^{(\infty)}-t \lim_{\lb\to0}\left( g(\lb)-\frac{\phi(\lb)}{2t}\right)}Z^{(\infty)}_{12}(0;n+1,t)+\oo(t^{-1}).
%\end{align*}
%A combination of \eqref{g0}, \eqref{Imginf}, \eqref{eq:Zinf12} and the above formula gives us \eqref{mr2} and \eqref{mr4}. 

% For convenience, denote $q^{(\infty)}_n(t)$ as the result of substituting $Z^{(\infty)}$      into the reconstruction formula \eqref{rec}, whose explicit expression is given by
% 	\begin{align}\label{uglo}
% 		q^{(\infty)}_n(t)=\frac{1}{2}\left(\sqrt{\frac{\alpha(\xi)}{\eta_1}}-\sqrt{\frac{\eta_1}{\alpha(\xi)}}\right)\Jac\left(\frac{K(k(\xi))(t\Omega+\Delta)}{\pi},k(\xi)\right),
% 	\end{align}		
%     where $\Jac$ is the subsidiary  Jacobi elliptic function,  $\alpha(\xi)$ and $k(\xi)$  are defined in \eqref{def alpha} and  \eqref{equ xialpha2}, respectively, while $\Omega$ and $\Delta$  are constants given in \eqref {def Omegj} and  \eqref{def Delta}.
% Then, substituting  $Z$ into the reconstruction formula \eqref{rec}, via explicit expressions \eqref{asy E} and estimate \eqref{uglo}, we obtain the asymptotic expression as $t\to+\infty$:
% \begin{align*}
% 		q_n(t)=&\E^{-tg^{(\infty)}-t \lim_{\lb\to0}\left( g(\lb)-\frac{\phi(\lb)}{2t}\right)}q^{(\infty)}_n(t)+\oo(t^{-1}).
% \end{align*}
%the exact expression of  $q^{(\infty)}_n(t)$ is given in \eqref{uglo} and
%$g^{(\infty)}$ is defined in \eqref{def ginf}.

\paragraph{Proof of \eqref{mr3}}
For $\xi \in T_{II}$, we still have \eqref{zlb}. Since the estimate of $E$ now is given by \eqref{asy ET2}, we are led to \eqref{mr3}.

\vspace{2mm}
% Then,
% by explicit expressions \eqref{asy ET2} and estimate \eqref{uglo}, we obtain the asymptotic expression as $t\to+\infty$:
%  \begin{align*}
% 		q_n(t)=&\E^{-tg^{(\infty)}-t \lim_{\lb\to0}\left( g(\lb)-\frac{\phi(\lb)}{2t}\right)}q^{(\infty)}_n(t)+\oo(t^{-1/3}).
% \end{align*}
This completes the proof of Theorem \ref{Thm 2}. \qed 

\begin{remark}
   It is noticed that $Z^{(\infty)}_{12}(0;n+1,t)$ is  a genus-$1$ algebraic-geometric solution of the AL system \eqref{e1.2}. Indeed, by a similar proof to Proposition \ref{Pro 1}, $Z^{(\infty)}(\lb)\E^{\frac{\phi(\lb)}{2}\sigma_3}$ admits a Lax pair of the form \eqref{lax Z}. Thus, its compatibility condition yields the AL system \eqref{e1.2}, namely, $Z^{(\infty)}_{12}(0;n+1,t)$ obtained from the reconstruction formula  solves the AL system \eqref{e1.2}.
   % By a similar proof to \hl{Proposition} \ref{Pro 1}, $Z^{(\infty)}(\lb)\E^{\phi(\lb)\sigma_3/2}$ admits the Lax pair like the form \eqref{lax Z}. So $q_n^{(\infty)}(0)$ is the initial data of a genus-$1$ algebraic-geometric solution of the AL system \eqref{e1.2}.
\end{remark}

\appendix

\section{The model RH problems}\label{app1}
In this section, we list the model RH problems used in the construction of local parametrices. 
\subsection{The Bessel parametrix}\label{app bessel}
\begin{figure}[h]
\centering
		\begin{tikzpicture}
		\draw[dashed] (-2,0)--(2,0);
		\draw[thick,->] (-2,0)--(-1,0)node[below]{$\Sigma_0$};
		\draw [thick,] (-2,0)--(0,0);
		\draw[thick,->](-1,1.713)--(-0.5,0.857) node[right]{$\Sigma_+$} ;
		\draw[thick] (-1,1.713)--(0,0);
		\draw[thick,->](-1,-1.713)--(-0.5,-0.857)node[right]{$\Sigma_-$} ;
		\draw[thick] (-1,-1.713)--(0,0);
		\draw (0,-0.1)node [below]{$0$};
	\end{tikzpicture}
	\caption{\footnotesize The jump contour $\Sigma_\Psi=\Sigma_+\cup\Sigma_-\cup\Sigma_0$ of $\Psi_{\Bes}$.}
    \label{fig bessel}
\end{figure}
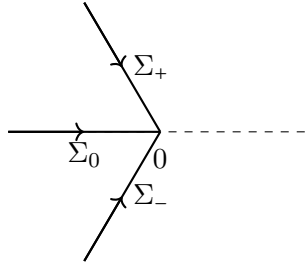
The Bessel parametrix $\Psi_{\Bes}$ is the unique solution of the following RH problem. 
\begin{Rhp} \label{rhp:Bes} 
\ 
\begin{itemize}
	\item $\Psi_\Bes(\zeta)$ is analytic for $\zeta \in \mathbb{C} \setminus \Sigma_{\Psi}$, where $\Sigma_\Psi$ is the union of  three contours $\Sigma_\pm = \left\{\zeta: \arg \zeta = \pm\frac{2\pi }{3}\right\} $ and $\Sigma_0 = \left\{ \zeta: \arg \zeta = \pi \right\}$ as shown in Figure \ref{fig bessel}.
    %, and take the argument of $\zeta$ as $\arg\zeta\in(-\pi,\pi)$.% (see \figurename \ \ref{step2});
	\item $\Psi_\Bes(\zeta)$ satisfies the following jump condition
	\begin{align*}
		\Psi_{\Bes, +}(\zeta) =  \Psi_{\Bes,-}(\zeta) \begin{cases} \begin{pmatrix} 1 & 0 \\ 1 & 1 \end{pmatrix}, & \zeta \in \Sigma_+ \cup \Sigma_-,\\
			\begin{pmatrix} 0 & 1 \\ -1\ & 0 \end{pmatrix}, & \zeta \in \Sigma_0.
		\end{cases}
	\end{align*}
    \item As $\zeta \to \infty$, we have
    \begin{align}\label{asy bessel}
	\Psi_\Bes(\zeta) = \left( 2\pi \zeta^{\frac{1}{2}}\right)^{-\frac{1}{2}\sigma_3} \frac{1}{\sqrt{2}} \begin{pmatrix}
		\displaystyle 1& \ii  \displaystyle \\
		\displaystyle \ii& \displaystyle 1
	\end{pmatrix} \left(I +  \oo \left(\frac{1}{\zeta^{\frac{1}{2}}}\right)\right)\E^{2\zeta^{\frac{1}{2}}\sigma_3}.
 \end{align}
	\item As $\zeta \rightarrow 0$, we have 
	\begin{align*}
		\Psi_\Bes(\zeta) =  \oo \left(\ln |\zeta| \right).
 %        \begin{pmatrix}  & \oo \left(\ln |\zeta| \right) \\ \oo \left( \ln|\zeta| \right)& \oo \left( \ln |\zeta|\right)  \end{pmatrix} \ .
	\end{align*}
\end{itemize}	
\end{Rhp}
By \cite[formul\ae \ (6.16)--(6.20)]{KMcLVAV} (with $\alpha=0$ therein), one has 
\begin{align}\label{def:Bes}
	\Psi_\Bes(\zeta) =
	\begin{cases}
		\begin{pmatrix} \displaystyle I_0 (2\zeta^{\frac{1}{2}}) &\displaystyle \frac{\ii }{\pi} K_0 (2\zeta^{\frac{1}{2}})\\
			\displaystyle 2\pi \ii  \zeta^{\frac{1}{2}}I'_0 (2\zeta^{\frac{1}{2}}) & \displaystyle - 2\zeta^{\frac{1}{2}}K'_0(2\zeta^{\frac{1}{2}})
		\end{pmatrix}, & -\frac{2\pi }{3}<\arg \zeta < \frac{2\pi }{3}, \\	
		\begin{pmatrix} \displaystyle \frac{1}{2}H^{(1)}_0 (2(-\zeta)^{\frac{1}{2}}) &\displaystyle \frac{1}{2} H^{(2)}_0 (2(-\zeta)^{\frac{1}{2}})\\
			\displaystyle \pi \zeta^{\frac{1}{2}} (H^{(1)}_0)' (2(-\zeta)^{\frac{1}{2}}) & \displaystyle  \pi \zeta^{\frac{1}{2}} (H^{(2)}_0)'(2(-\zeta)^{\frac{1}{2}})
		\end{pmatrix},  &    \frac{2\pi }{3} <\arg \zeta < \pi, \\
		\begin{pmatrix} \displaystyle \frac{1}{2}H^{(2)}_0 (2(-\zeta)^{\frac{1}{2}}) &\displaystyle -\frac{1}{2} H^{(1)}_0 (2(-\zeta)^{\frac{1}{2}})\\
			\displaystyle -\pi \zeta^{\frac{1}{2}} (H^{(2)}_0)' (2(-\zeta)^{\frac{1}{2}}) & \displaystyle  \pi \zeta^{\frac{1}{2}} (H^{(1)}_0)'(2(-\zeta)^{\frac{1}{2}})
		\end{pmatrix},  &  -\pi < \arg \zeta< - \frac{2\pi }{3},
	\end{cases} 
\end{align}
% with
% \begin{align}\label{asy bessel}
% 	\Psi_\Bes(\zeta) = \left( 2\pi \zeta^{\frac{1}{2}}\right)^{-\frac{1}{2}\sigma_3} \frac{1}{\sqrt{2}} \begin{pmatrix}
% 		\displaystyle 1& \ii  \displaystyle \\
% 		\displaystyle \ii& \displaystyle 1
% 	\end{pmatrix} \left(I +  \oo \left(\frac{1}{\zeta^{\frac{1}{2}}}\right)\right)\E^{2\zeta^{\frac{1}{2}}\sigma_3}
% \end{align}
% uniformly as $\zeta \rightarrow \infty$ everywhere in the complex plane aside from the jumps. 
where $I_0(\zeta)$ and $K_0(\zeta)$ are the modified Bessel functions of order $0$, $H^{(j)}_0(\zeta)$, $j=1,2$, are the Hankel functions (cf. \cite[Chapter 10]{dlmf}) and the principal branch is taken for $\zeta^{\frac{1}{2}}$.

\subsection{The Airy parametrix}\label{app airy}

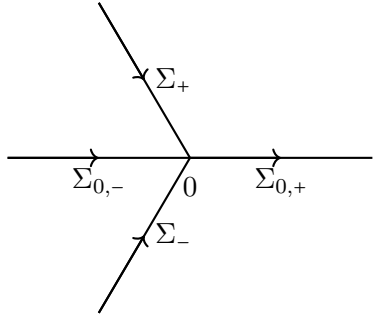
\begin{figure}[h]
	\centering
	\begin{tikzpicture}[scale=1.2]

		\draw[thick,->] (-2,0)--(-1,0)node[below]{$\Sigma_{0,-}$};
		\draw [thick] (-2,0)--(0,0);
		\draw[thick,->](-1,1.713)--(-0.5,0.857) node[right]{$\Sigma_+$} ;
		\draw[thick] (-1,1.713)--(0,0);
		\draw[thick,->](-1,-1.713)--(-0.5,-0.857)node[right]{$\Sigma_-$} ;
		\draw[thick] (-1,-1.713)--(0,0);
		\draw (0,-0.1)node [below]{$0$};
\draw [thick] (0,0)--(2,0);
   \draw[thick,->] (0,0)--(1,0)node[below]{$\Sigma_{0,+}$};     
	\end{tikzpicture}
	\caption{\footnotesize The jump contour $\Sigma_\Psi=\Sigma_+ \cup\Sigma_-\cup \Sigma_{0,+} \cup \Sigma_{0,-}$ of $\Psi_{\Ai}$ and $M^{\rm{P}_{34}}$.}
	\label{Airy}
\end{figure}

The Airy parametrix $\Psi_{\Ai}$ is the unique solution of the following RH problem.
%for Airy function \cite{Deift1999, DKMVZ1999}:
\begin{Rhp}\label{RHairy}\
    \begin{itemize}
        \item $\Psi_{\Ai}(\zeta)$ is analytic for $\zeta\in\mathbb{C}\setminus\Sigma_\Psi$, where $\Sigma_\Psi$ is the union of four contours $\Sigma_\pm=\{\zeta: \arg \zeta=\pm \frac{2\pi}{3} \}$, $\Sigma_{0,-}=\{\zeta: \arg \zeta = \pi \}$ and      $\Sigma_{0,+}=\{\zeta: \arg \zeta = 0 \}$ as shown in Figure \ref{Airy}.
  \item $\Psi_{\Ai}(\zeta)$ satisfies the jump condition 
\begin{align*}
	\Psi_{\Ai,+}(\zeta)=\Psi_{\Ai,-}(\zeta) \left\{\begin{array}{lll}
		\left(\begin{array}{cc}
			1 & 0\\
			1 & 1
		\end{array}\right),  &\zeta\in \Sigma_+\cup\Sigma_-,\\
		\left(\begin{array}{cc}
			0 & 1\\
			-1 & 0
		\end{array}\right),  &\zeta\in \Sigma_{0,-},\\
		\left(\begin{array}{cc}
			1 & 1\\
			0& 1
		\end{array}\right),  &\zeta\in \Sigma_{0,+}.\\
	\end{array}\right.
\end{align*}
\item As $\zeta\to \infty$, we have 
\begin{equation}\label{asymlo}
	\Psi_{\Ai}(\zeta) = \zeta^{-\frac{\sigma_3}{4}} \frac{1}{\sqrt{2}}\begin{pmatrix}
		1	&	\ii \\
		\ii &1\\
	\end{pmatrix} \left(I+\oo\left(\zeta^{-\frac{3}{2}}\right)\right) \e^{-\frac{2}{3}\zeta^{\frac{3}{2}}\sigma_3}.
\end{equation}
\item $\Psi_{\Ai}$ remains bounded as $\zeta\to 0$, $\zeta \in \C \setminus \Sigma_{\Psi}$.
    \end{itemize}
\end{Rhp}

By \cite{Deift1999, DKMVZ1999}, the solution to RH problem \ref{RHairy} is given by 
%is constructed with the help of Airy functions. Setting $\omega=\e^{\frac{2\pi \i}{3}}$, we have  
\begin{align}\label{def:Airy}
    \Psi_{\Ai}(\zeta) =
	\begin{cases}
		\sqrt{2 \pi}\begin{pmatrix} \displaystyle \Ai(\zeta) &\displaystyle - \omega^2\Ai (\omega^2\zeta)\\
			\displaystyle - \ii  \Ai'(\zeta) & \displaystyle \i \omega   \Ai'(\omega^2\zeta)
		\end{pmatrix}, & 0<\arg \zeta < \frac{2\pi }{3}, \\
		\sqrt{2 \pi}\begin{pmatrix} \displaystyle -\omega\Ai(\omega\zeta) &\displaystyle - \omega^2\Ai (\omega^2\zeta)\\
			\displaystyle  \ii\omega^2  \Ai'(\zeta) & \displaystyle \i \omega   \Ai'(\omega^2\zeta)
		\end{pmatrix},  &    \frac{2\pi }{3} <\arg \zeta < \pi, \\
		\sqrt{2 \pi}\begin{pmatrix} \displaystyle  -\omega^2\Ai(\omega^2\zeta) &\displaystyle  \omega\Ai (\omega\zeta)\\
			\displaystyle \ii \omega \Ai'(\omega^2\zeta) & \displaystyle -\i \omega^2   \Ai'(\zeta)
		\end{pmatrix},  &  -\pi < \arg \zeta< - \frac{2\pi }{3},\\
        \sqrt{2 \pi}\begin{pmatrix} \displaystyle \Ai(\zeta) &\displaystyle  \omega\Ai (\omega\zeta)\\
			\displaystyle - \ii  \Ai'(\zeta) & \displaystyle -\i \omega^2   \Ai'(\zeta)
		\end{pmatrix},  &   - \frac{2\pi }{3}< \arg \zeta< 0,
	\end{cases} 
\end{align}
where $\Ai(\zeta)$ is the Airy function (cf. \cite[Chapter 9]{dlmf}) and $\omega=\e^{2\pi \i/3}$.

\subsection{The generalized Laguegrre polynomial parametrix}\label{Lague}
The generalized Laguerre polynomials with index $0$ are given by 
\begin{equation*}
L_m(\zeta) = \frac{ \E^{\zeta}}{m!} \frac{\ddd ^m}{\ddd  \zeta^m} \left(\E^{-\zeta} \zeta^{m}\right), \qquad m=0,1,\ldots,
\end{equation*}
which satisfy the orthogonality condtions
\begin{equation*}
\int_{0}^\infty  \E^{-\zeta} L_m(\zeta)  L_n(\zeta) \ddd  \zeta = \frac{\Gamma(m+1)}{m!} \delta_{m,n}.
\end{equation*}
% For convenience, let
% \begin{equation*}
% \pi_m(\zeta) = (-1)^{m}m!L_m(\zeta).
% \end{equation*}
% Thus,
% \begin{equation*}
% \int_{0}^\infty \E^{-\zeta} \pi_m(\zeta)  \pi_n(\zeta) \ddd  \zeta = \Gamma(m+1) m! \delta_{m,n}.
% \end{equation*}
It comes out that $L_m$ is characterized by the following RH problem, which we call the generalized Laguegrre polynomial parametrix.
%Then $L(\zeta)$ is the solution of the following model RH problem for the generalized Laguerre polynomials with index $0$ and degree $m$ \cite{its1992}:
\begin{Rhp}\label{Lzeta}\
	%Find a $2\times 2$ matrix valued function $L(\zeta)$ with the following properties:
	\begin{itemize}
		\item $L(\zeta)$ is  analytic for $\zeta \in \mathbb{C} \setminus [0,+\infty).$
		\item $L$ satisfies the jump condition 
        %For $\zeta \in (0,+\infty)$,
		\begin{equation}\label{Ljumps}
		L_+(\zeta) = L_-(\zeta)
		\begin{pmatrix} 1 & 0 \\  \E^{-\zeta} & 1 \end{pmatrix}, \qquad \zeta\in (0,+\infty).
		\end{equation}	
		\item As $\zeta \to\infty$, we have 
        \begin{align}\label{Linf}
            L(\zeta)= \left( I + \oo(\zeta^{-1}) \right)\zeta^{-m \sigma_3},
        \end{align}
         where $m=0,1,  \ldots$.
         \item As $\zeta \to 0$, we have 
        \begin{align}
            L(\zeta)= \oo \left(\ln |\zeta| \right).
        \end{align}
	\end{itemize}
\end{Rhp}

By \cite{its1992}, we have 
%The solution of the above RH problem is written as follows: for $m \ge 1$,
\begin{align}\label{def:LagSol}
L
(\zeta)=\left\{
  \begin{array}{ll}
    \begin{pmatrix}
-\frac{1}{\Gamma(m)^2 }\int_{0}^{+\infty}\frac{ \pi_{m-1} (z) \E^{-z}}{z-\zeta} \ddd  z		& -\frac{2\pi \ii}{\Gamma(m)^2 } \pi_{m-1}(\zeta)\\
\frac{1}{2\pi \ii }\int_{0}^{+\infty}\frac{ \pi_m (z) \E^{-z}}{z-\zeta}\ddd  z	 &\pi_m(\zeta)\\
\end{pmatrix}, & \hbox{$m\geq 1$;} \\
    \begin{pmatrix}
1		& 0\\
\frac{1}{2\pi \ii }\int_{0}^{+\infty}\frac{  \E^{-z}}{z-\zeta}\ddd  z	 & 1\\
\end{pmatrix}, & \hbox{$m=0$,}
  \end{array}
 \right.
%\begin{pmatrix}
% -\frac{1}{\Gamma(m)^2 }\int_{0}^{+\infty}\frac{ \pi_{m-1} (z) \E^{-z}}{z-\zeta} \ddd  z		& -\frac{2\pi \ii}{\Gamma(m)^2 } \pi_{m-1}(\zeta)\\
% \frac{1}{2\pi \ii }\int_{0}^{+\infty}\frac{ \pi_m (z) \E^{-z}}{z-\zeta}\ddd  z	 &\pi_m(\zeta)\\
% \end{pmatrix},
\end{align}
where $\pi_m(\zeta) = (-1)^{m}m!L_m(\zeta)=\zeta^m+\cdots$ is the monic Laguerre polynomial. 
% and for $m=0$,
% \begin{align*}
% L(\zeta)=\begin{pmatrix}
% 1		& 0\\
% \frac{1}{2\pi \ii }\int_{0}^{+\infty}\frac{  \E^{-z}}{z-\zeta}\ddd  z	 & 1\\
% \end{pmatrix}.
% \end{align*}

\subsection{The Painlev\'e XXXIV parametrix}\label{app4}

The Painlev\'e XXXIV parametrix $M^{\rm{P}_{34}}$ solves the following RH problem.

% for 
%  Painlev\'e XXXIV equation \cite{FIKY2006,ikj2008}:

\begin{Rhp}\label{RHp34}\
%	Find a $2\times 2$ matrix valued function $M^{\rm{P}_{34}}(\zeta)$ with the following properties:
	\begin{itemize}
	 	\item $M^{\rm{P}_{34}}(\zeta)=M^{\rm{P}_{34}}(\zeta; b, \omega, s)$ is analytic for $\zeta \in
		\mathbb{C} \setminus \Sigma_\Psi$, where $\Sigma_\Psi$ is shown in Figure \ref{Airy}.
		
		\item $M^{\rm{P}_{34}}(\zeta)$  satisfies the jump condition
		\begin{equation*}
		M^{\rm{P}_{34}}_+ (\zeta)=M^{\rm{P}_{34}}_- (\zeta)
		\left\{ \begin{array}{ll}
		\begin{pmatrix}
		1 & \omega
		\\
		0 & 1
		\end{pmatrix}, &\quad \zeta \in \Sigma_{0,+},
		\\
		\begin{pmatrix}
		1 & 0 \\
		\E^{2b\pi \ii} & 1
		\end{pmatrix}, &\quad \zeta \in \Sigma_+,
		\\
		\begin{pmatrix}
		0 & 1 \\
		-1 & 0
		\end{pmatrix},& \quad
		\zeta \in \Sigma_{0,-},
		\\
		\begin{pmatrix}
		1 & 0 \\
		\E^{-2b\pi \ii} & 1
		\end{pmatrix},  & \quad \zeta \in \Sigma_-.
		\end{array}
		\right .
		\end{equation*}
		
		\item As $\zeta \to \infty$, there exists a function $a(s)=a(s;b,\omega)$ such that
		\begin{align} \label{eq:Psi-infinity}
		M^{\rm{P}_{34}}(\zeta) = \begin{pmatrix}
		1 & 0\\
		-\ii a(s) & 1
		\end{pmatrix}
		\left(I+\frac {M^{\rm{P}_{34}}_1(s)}{\zeta}
		+\oo \left( \zeta^{-2} \right) \right)
		\frac{\zeta^{-\frac{1}{4}\sigma_3}}{\sqrt{2}}
		\begin{pmatrix}
		1 & \ii
		\\
		\ii & 1
		\end{pmatrix} \E^{-(\frac{2}{3}\zeta^{3/2}+s\zeta^{1/2}) \sigma_3},
		\end{align}
		where
        % we take the  principle branch for the fractions and
		\begin{equation*}
		\left(M^{\rm{P}_{34}}_1(s)\right)_{12}=\ii a(s).
		\end{equation*}
		
		\item As $\zeta \to 0$, we have, if  $-1/2 < b < 0$
		\begin{equation*}
		M^{\rm{P}_{34}}(\zeta)=\oo({\zeta^{b}}),
		\end{equation*}
		and if $b \geq 0$,
		\begin{equation*}
		M^{\rm{P}_{34}}(\zeta)=\left\{ \begin{array}{ll}
		\begin{pmatrix}
		 \oo(\zeta^{b}) & \oo(\zeta^{-b})
		\\
		 \oo(\zeta^{b}) &  \oo(\zeta^{-b})
		\end{pmatrix}, &\quad  -\frac{2\pi}{3}<\arg \zeta <\frac{2\pi}{3},
		\\
		\oo(\zeta^{-b}), &\quad  \frac{2\pi}{3} <\arg < \pi \, \text{and}\,  -\pi <\arg \zeta < -\frac{2\pi}{3}.
		\end{array}
		\right .
		\end{equation*}
		\end{itemize}
\end{Rhp}

By \cite{FIKY2006,ikj2008,ikj2009,XZ2011}, the above RH problem is uniquely solvable for $b>-1/2$, $\omega \in \mathbb{C}\setminus (-\infty,0)$, and $s\in \mathbb{R}$. Moreover, with $a(s)$ given in
\eqref{eq:Psi-infinity}, the function
\begin{equation}\label{def:u}
    u(s):=u(s;b,\omega)=a'(s;b,\omega)-\frac s2
\end{equation}
satisfies the  Painlev\'{e}  XXXIV equation
\begin{equation}\label{def:P34}
u''(s)=4u(s)^2+2su(s)+\frac{u'(s)^2-(2b)^2}{2u(s)}.
\end{equation}

\begin{comment}
and is pole-free on the real axis.
Particularly, one has
\begin{equation}\label{eq:uasygen}
u(h;b, 0)=\left\{ \begin{array}{ll}
b/\sqrt{h}+\oo(h^{-2}), &\qquad h\to +\infty,
\\
-h/2+\oo(h^{-2}), &\qquad h\to -\infty.
\end{array}
\right .
\end{equation}
This, together with the fact that $a(h;b,0)\to 0$ as $h\to -\infty$, implies that
\begin{equation}\label{p34-a}
a(h;b,0)=\int_{-\infty}^h \left( u(t;b, 0) + \frac t2\right) \ddd  t.
\end{equation}
\end{comment}

\section{Unique solvability of the equation \eqref{equ xialpha} }\label{app5}
In this section, we show that \eqref{equ xialpha} admits a unique solution.
%give the proof of the existence for the solution $\alpha(\xi)$ of  equation \eqref{equ xialpha}. 
To proceed, we rewrite the right-hand of \eqref{equ xialpha} as 
$N(\alpha)/D(\alpha)$, where
% a function of $\alpha$ as 
% \begin{align}\label{E1}
% \frac{-\alpha^2-\frac{1}{\alpha^2}+\frac{(\alpha +\frac{1}{\alpha}+\eta_1 +\frac{1}{\eta_1})}{2}(\alpha +\frac{1}{\alpha}+2-\frac{4\Pi(l_1^2,k(\alpha))}{K(k(\alpha))})+\frac{\int_{\eta_1}^{\eta_1^{-1}}\frac{s^2+s^{-2}}{R(s)}	\ddd  s}{\int_{\eta_1}^{\eta_1^{-1}}\frac{1}{R(s)}	\ddd  s}}{\alpha+\alpha^{-1}+2 - 4\frac{\Pi(l_1^2,k(\alpha))}{K(k(\alpha))}}
% \end{align}
% with
% \begin{align}
%     &k(\alpha)=\dfrac{l_1(\alpha+1)}{\alpha-1}.\label{equ kalpha2}
% \end{align}
% For convenience, we denote the numerator and denominator of \eqref{E1} as $N(\alpha)$ and $D(\alpha)$ respectively, namely,
\begin{align}
   &N(\alpha):=-\alpha^2-\frac{1}{\alpha^2}+\frac{(\alpha +\frac{1}{\alpha}+\eta_1 +\frac{1}{\eta_1})}{2}\left(\alpha +\frac{1}{\alpha}+2-\frac{4\Pi(l_1^2,k(\alpha))}{K(k(\alpha))}\right)+\frac{\int_{\eta_1}^{\eta_1^{-1}}\frac{s^2+s^{-2}}{R(s)}	\ddd  s}{\int_{\eta_1}^{\eta_1^{-1}}\frac{1}{R(s)}	\ddd  s}, \label{def:Nalpha}\\
   &D(\alpha):=  \alpha+\alpha^{-1}+2 - 4\frac{\Pi(l_1^2,k(\alpha))}{K(k(\alpha))}
\end{align}
with 
 \begin{align}
     &k(\alpha)=\dfrac{l_1(\alpha+1)}{\alpha-1} 
     \label{equ kalpha2}
 \end{align}
are two functions with respect to $\alpha$. From the properties of $\Pi(l_1^2,k(\alpha))$ and $K(k(\alpha))$ in \cite[Chapter 19]{dlmf},  it follows that $D(\alpha)$ is an increasing function with respect to $\alpha$. Moreover, as $\alpha\to\eta_1$, $k(\alpha) \to 1$ and then it holds that
\begin{align*}
    \frac{\Pi(l_1^2,k(\alpha))}{K(k(\alpha))}=\frac{1}{1-l_1^2}+\frac{l_1\ln\eta_1}{1-l_1^2}\frac{1}{\ln(1-k(\alpha)^2)}+\mathcal{O}\left(\frac{1-k(\alpha)^2}{\ln(1-k(\alpha)^2)}\right),
\end{align*}
from which we have
\begin{align}\label{Da}
   D(\alpha)=-\frac{4l_1\ln\eta_1}{1-l_1^2}\frac{1}{\ln(1-k(\alpha)^2)}+\mathcal{O}\left(\frac{1-k(\alpha)^2}{\ln(1-k(\alpha)^2)}\right).
\end{align}

As for $N(\alpha)$, it is noticed from \eqref{eRn}, \eqref{ei1}  and \eqref{ei2} that 
    \begin{align*}
        \int_{\eta_1}^{\eta_1^{-1}}\frac{s^2+s^{-2} }{R(s)}	\ddd  s&=\int_{-l_1}^{l_1}\sqrt{\frac{(1-l_1^2)(1-l(\alpha)^2)}{(l_1^2-x^2)(l(\alpha)^2-x^2)}}\left(1+\frac{8x^2}{(1-x^2)^2}\right)\ddd x\\
        &= 2\int_0^1\frac{\sqrt{(1-l_1^2)(1-l(\alpha)^2)}}{l(\alpha)\sqrt{(1-x^2)(1-k(\alpha)^2x^2)}}\left(1+\frac{8l_1^2x^2}{(1-l_1^2x^2)^2}\right)\ddd x\\
        &= 2\frac{\sqrt{(1-l_1^2)(1-l(\alpha)^2)}}{l(\alpha)}\left(K(k(\alpha))+8l_1^2\Pi_{l_1^2}(l_1^2,k(\alpha))\right),
      %  =&2\frac{\sqrt{(1-l_1^2)(1-l(\alpha)^2)}}{l(\alpha)}\left(K(k(\alpha))+\frac{4l_1^2}{(k(\alpha)^2-l_1^2)(l_1^2-1)}\left(E(k(\alpha))+\frac{k(\alpha)^2-l_1^2}{l_1^2}K(k(\alpha))-\frac{k(\alpha)^2-l_1^4}{l_1^2}\Pi(l_1^2,k(\alpha))\right)\right)
    \end{align*}
    where we recall that
    \begin{align*}
        l(\alpha)=\frac{\alpha-1}{\alpha+1},
    \end{align*}
    and 
    \begin{align*}
        \Pi_{l_1^2}(l_1^2,k(\alpha)) =\frac{1}{2(k(\alpha)^2-l_1^2)(l_1^2-1)}\left(E(k(\alpha))+\frac{k(\alpha)^2-l_1^2}{l_1^2}K(k(\alpha))-\frac{k(\alpha)^2-l_1^4}{l_1^2}\Pi(l_1^2,k(\alpha))\right).
    \end{align*}
    This, together with \eqref{int 1/R(s)}, implies that
    \begin{align*}
       &\frac{\int_{\eta_1}^{\eta_1^{-1}}\frac{s^2+s^{-2}}{R(s)}	\ddd  s}{\int_{\eta_1}^{\eta_1^{-1}}\frac{1}{R(s)}	\ddd  s}=2+\frac{8l_1^2}{(k(\alpha)^2-l_1^2)(l_1^2-1)}\left(\frac{E(k(\alpha))}{K(k(\alpha))}+\frac{k(\alpha)^2-l_1^2}{l_1^2}-\frac{k(\alpha)^2-l_1^4}{l_1^2}\frac{\Pi(l_1^2,k(\alpha))}{K(k(\alpha))}\right).
    \end{align*}
   Substituting the above equation into \eqref{def:Nalpha}, we have
   \begin{align*}
     N(\alpha)= & -\alpha^2-\frac{1}{\alpha^2}+\frac{(\alpha +\frac{1}{\alpha}+\eta_1 +\frac{1}{\eta_1})}{2}(\alpha +\frac{1}{\alpha}+2-\frac{4\Pi(l_1^2,k(\alpha))}{K(k(\alpha))})\\
     &+2+\frac{8l_1^2}{(k(\alpha)^2-l_1^2)(l_1^2-1)}\left(\frac{E(k(\alpha))}{K(k(\alpha))}+\frac{k(\alpha)^2-l_1^2}{l_1^2}-\frac{k(\alpha)^2-l_1^4}{l_1^2}\frac{\Pi(l_1^2,k(\alpha))}{K(k(\alpha))}\right),
   \end{align*}
which is a decreasing function with respect to $\alpha$. Thus, the right-hand of \eqref{equ xialpha} is  a decreasing function with respect to $\alpha$. In addition, as $\alpha\to\eta_1$, we have
    \begin{align*}
       \frac{\int_{\eta_1}^{\eta_1^{-1}}\frac{s^2+s^{-2}}{R(s)}	\ddd  s}{\int_{\eta_1}^{\eta_1^{-1}}\frac{1}{R(s)}	\ddd  s}&= \eta_1^2+\eta_1^{-2}-\frac{\frac{8l_1^2}{(1-l_1^2)^2}\big(2-\frac{l_1^2+1}{l_1}\ln\eta_1\big)}{\ln(1-k(\alpha)^2)}+\mathcal{O}\left(\frac{1-k(\alpha)^2}{\ln(1-k(\alpha)^2)}\right).    \end{align*}
Combining this with \eqref{Da} yields, as $\alpha\to\eta_1$, 
  \begin{align*}
      \frac{N(\alpha)}{D(\alpha)}&=\eta_1 +\frac{1}{\eta_1}+\frac{ -\frac{\frac{8l_1^2}{(1-l_1^2)^2}\big(2-\frac{l_1^2+1}{l_1}\ln\eta_1\big)}{\ln(1-k(\alpha)^2)}+\mathcal{O}\left(\frac{1-k(\alpha)^2}{\ln(1-k(\alpha)^2)}\right)}{-\frac{4l_1\ln\eta_1}{1-l_1^2}\frac{1}{\ln(1-k(\alpha)^2)}+\mathcal{O}\left(\frac{1-k(\alpha)^2}{\ln(1-k(\alpha)^2)}\right)}\\
      &=-\frac{\eta_1-\eta_1^{-1}}{\ln\eta_1}+\mathcal{O}\left(1-k(\alpha)^2\right)
  \end{align*} 
Therefore, we conclude that $\frac{N(\alpha)}{D(\alpha)}$ is  a decreasing function on $(\eta_1,\eta_2)$ with boundary values
$$\frac{N(\eta_1)}{D(\eta_1)}=-\frac{\eta_1-\eta_1^{-1}}{\ln\eta_1}, \qquad\frac{N(\eta_2)}{D(\eta_2)}=\xi_{\mathrm{crit}},$$
where
$\xi_{\mathrm{crit}}$ is defined in \eqref{def xi}. Consequently, for each $\xi\in (\xi_{\mathrm{crit}},-\frac{\eta_1-\eta_1^{-1}}{\ln\eta_1})$, there exists a unique $\alpha(\xi)\in(\eta_1,\eta_2)$ as a solution of \eqref{equ xialpha}.       

\vspace{4mm}

\noindent{\bf Acknowledgements.}
The work of Chen is partially supported by NSFC under grant 12401311 and Natural Science Foundation of Fujian Province under grant 2024J01306.
The work of Fan is partially supported by NSFC under grants 12271104, 51879045. The work of Yang is partially supported by
Natural Science Foundation of Chongqing under grant  CSTB2025NSCQ-GPX0727, NSFC under grant 12501325  and Innovation Center for Mathematical Analysis of Fluid and Chemotaxis (Chongqing University).

\hspace*{\parindent}
\\


\begin{thebibliography}{10}



\bibitem{ABP2007}
{\small  \textsc{M. J. Ablowitz, G. Biondini and B. Prinari,} \ Inverse scattering transform for the integrable discrete nonlinear Schr\"{o}dinger equation with nonvanishing boundary conditions, \textit{Inverse Probl.,} \textbf{23} (2007), 1711-1758.}

\bibitem{AL1974}	
{\small  \textsc{M. J. Ablowitz and J. F.  Ladik,} \ Nonlinear differential-difference equations, \textit{J. Math. Phys.,} \textbf{16}  (1974), 598-603.}

\bibitem{AL1976}	
{\small  \textsc{M. J. Ablowitz and  J. F.  Ladik,}\  Nonlinear differential-difference equations and Fourier analysis, \textit{J. Math. Phys.,} \textbf{17} (1976), 1011-1018.}

\bibitem{APT2004}
{\small  \textsc{M. J. Ablowitz, B. Prinari and A. Trubatch,}\ Discrete and continuous nonlinear Schr\"{o}dinger systems, Cambridge University Press, UK, 2004.}

\bibitem{AA2025}
{\small  \textsc{A. Aggarwal,}\  Asymptotic scattering relation for the Toda lattice. arxiv2503.08018v1.}

\bibitem{AKV2020}	
{\small  \textsc{Y. Angelopoulos, R. Killip and M. Visan,} \  	Invariant measures for integrable spin chains and an integrable discrete nonlinear Schr\"{o}dinger equation, \textit{ SIAM J. Math. Anal.,}  \textbf{52} (2020), 135-163.}

%\bibitem{NLSgas2}
%{\small  \textsc{M. Bertola, T. Grava and G. Orsatti,}\ $\bar{\partial}$-problem for focusing nonlinear Schr\"odinger equation and soliton shielding, 	\textit{Proc. A.}, \textbf{481} (2025), 20240764.} %
	
%\bibitem{NLSgas3}
%{\small  \textsc{M. Bertola, T. Grava and  G. Orsatti,}\ Soliton shielding of the focusing nonlinear Schr\"{o}dinger equation, \textit{	Phys. Rev. Lett.}, \textbf{130} (2023), 127201.}%

\bibitem{BM2018}
{\small  \textsc{M. Bertola and  A. Minakov}, \
Laguerre polynomials and transitional asymptotics
of the modified Korteweg-de Vries equation for step-like initial data,
\textit{	Anal. Math. Phys.}, \textbf{9} (2019), 1761-1818.}

%\bibitem{NLSgas4}
%{\small  \textsc{G. Biondini, G. A. El, X.-D. Luo, J. Oregero and A. Tovbis,} \ Breather gas fission from elliptic potentials in self-focusing media, \textit{Phys. Rev. E}, \textbf{111} (2025),  014204.} %
		

%\bibitem{BKA2009}
%{\small  \textsc{Yu V. Bludov, V. V. Konotop and N. Akhmediev,} \ Rogue waves as spatial energy concentrators in arrays of nonlinear waveguides,  \textit{Opt. Lett.,} \textbf{34} (2009), 3015-3017.}

%	\bibitem{Monvel2011}
%A. Boutet de Monvel, V. P. Kotlyarov and D. Shepelsky, \newblock   Focusing NLS equation: long-time dynamics of step-like initial data.\textit{ Int. Math. Res. Not. IMRN}, 7(2011), 1613-1653.

%\bibitem{MonvelCMP1}
%A. Boutet de Monvel, J. Lenells and D. Shepelsky,
%\newblock  The focusing NLS equation with step-like oscillating background: scenarios of long-time Asymptotics.  \textit{Comm. Math. Phys.}, 383(2021), 893-952.

%\bibitem{DIW2015}
% {\small  \textsc{D. Dai, M. E. H. Ismail and J. Wang,} \
%Asymptotics for Laguerre polynomials with large order and parameters,
%\textit{J. Approx. Theory}, \textbf{193} (2015), 4-19.}








\bibitem{mis} 
{\small  \textsc{A. Boutet de Monvel, A. Its and  D. Shepelsky,} \  Painlev\'e-type asymptotics for the Camassa-Holm equation,  \textit{SIAM J. Math. Anal.,} \textbf{42} (2010), 1854-1873.}

%\bibitem{MonvelCMP2}
%{\small  \textsc{A. Boutet de Monvel, J. Lenells and D. Shepelsky,} \ The focusing NLS equation with step-like oscillating Background: The genus 3 sector, \textit{Comm. Math. Phys.},  \textbf{392} (2022), 1081-1148.}

\bibitem{CFW2025}
{\small  \textsc{M. S. Chen, E. G. Fan and Z. Y. Wang,}\ Long-time asymptotics for the Ablowitz-Ladik system with present of solitons, arXiv:2407.21526.}


	
%	\bibitem{fNLS}
%{\small  \textsc{M. Borghese, R. Jenkins, K. T. R. McLaughlin  and P. Miller,}	\newblock { Long-time asymptotic behavior of the focusing nonlinear Schr\"odinger equation, }	\textit{Ann. I. H. Poincar$\acute{e}$ Anal}, \textbf{35} (2018), 887-920.}
%	\bibitem{NLSSTEPLIKE}
%{\small \textsc{R. Buckingham and S. Venakides}, Long-time asymptotics of the nonlinear Schr\"odinger equation shock problem, \textit{Commun. Pure and Appl. Math.,}  \textbf{60} (2007), 1349-1414.}


\bibitem{CER2021}
{\small \textsc{T. Congy, G. El, and G. Roberti,}\ Soliton gas in bidirectional dispersive hydrodynamics, \textit{Phys. Rev. E}, \textbf{103} (2021), 042201.}

\bibitem{Deift1999}
{\small \textsc{P. Deift,} \ Orthogonal polynomials and random matrices: A Riemann–Hilbert approach, volume 3
of courant lecture notes, New York University, 1999}

\bibitem{DKMVZ1999}
{\small \textsc{P. Deift, T. Kriecherbauer, K.T-R McLaughlin, S. Venakides, and X. Zhou},\ Strong asymptotics of
orthogonal polynomials with respect to exponential weights, \textit{Comm. Math. Phys.}, \textbf{2} (1999), 1491-1552}

\bibitem{gfunction}
{\small \textsc{P. Deift, S. Venakides and X. Zhou,}\ The collisionless shock region for the long-time behavior of solutions of the KdV equation, \textit{Comm. Pure Appl. Math.}, \textbf{47} (1994), 199-206.}
	
\bibitem{RN6}
{\small \textsc{P. Deift and X. Zhou,}\ A steepest descent method for oscillatory Riemann-Hilbert problems, Asymptotics for the MKdV equation, \textit{Ann. Math.}, \textbf{137} (1993),  295-368.}
	
	
\bibitem{RN10}
{\small \textsc{P. Deift and X. Zhou,} \	Long-time asymptotics for solutions of the NLS equation with initial data in a weighted Sobolev space,	\textit{ Comm. Pure Appl. Math.,} \textbf{56} (2003), 1029-1077}.





%\bibitem{AAS2010}
%{\small  \textsc{A. Ankiewicz, N. Akhmediev, J. M. Soto-Crespo,} \  Discrete rogue waves of the Ablowitz–Ladik and Hirota equations, \textit{Phys. Rev. E,} \textbf{82} (2010) 026602.}

%\bibitem{WYM2016}
%{\small  \textsc{X.-Y. Wen, Z. Yan,} \ B.A. Malomed, Higher-order vector discrete rogue-wave states in the coupled Ablowitz–Ladik equations: Exact solutions and stability, \textit{Chaos,} \textbf{26} (2016) 123110.}

	
%\bibitem{BT2015}
%{\small  \textsc{V. M. Buchstaber, S. I. Tertychnyi, } \ Holomorphic solutions of the double confluent Heun equation associated with the RSJ model of the Josephson junction,  \textit{Theor. Math. Phys. (Russian Fed.),} \textbf{182} (2015), 329-355.}

	
	
%	\bibitem{BC1984}
%	{\small  \textsc{R. Beals and R. R. Coifman,} \  Scattering and inverse scattering for first order systems,
%		\textit{Commun.  Pure and Appl. Math.,} \textbf{37} (1984), 39-90.}
%	
%	\bibitem{BC1985}
%	{\small  \textsc{R. Beals and R. R. Coifman,} \ Inverse scattering and evolution equations,
%		\textit{Commun.  Pure and Appl. Math.,}   \textbf{38} (1985),  29-42.}
	
	%\bibitem{Gardner1967}  {\small  \textsc{C. S. Gardner,  J. M. Green,  M. D.  Kruskal  and R. M.  Miura,} Method for solving the Korteweg-de Vries equation, Phys. Rev. Lett. 19(1967), 1095-1097.}


\bibitem{DZZ16}{\small \textsc{S. Dyachenko, D. Zakharov and V. Zakharov,} \ Primitive potentials and bounded solutions of the KdV equation, \textit{ Phys. D}, \textbf{333} (2016),  148-156.}

%\bibitem{E03}
%{\small  \textsc{G. A. El}, \  The thermodynamic limit of the Whitham equations, \textit{ Phys. Lett. A}, \textbf{311} (2003),  374-383. }


\bibitem{EK05}
{\small  \textsc{G. A. El and A. M. Kamchantov, }\  Kinetic equation for a dense soliton gas, \textit{ Phys. Rev. Lett.}, \textbf{95} (2005), 204101.}

\bibitem{EKPZ11}
{\small  \textsc{G. A. El,  A. M. Kamchantov, M. V. Pavlov and S. A. Zykov,} \ Spectral theory of soliton and breather gases for the focusing nonlinear
Schr\"{o}inger equation, \textit{Phys. Rev. E,} \textbf{101} (2020), 052207.
}

%\bibitem{ET20}
%{\small  \textsc{G. A. El and A. Tovbis, }\ Kinetic equation for a soliton gas and its hydrodynamic reductions, \textit{J. Nonlinear Sci.}, \textbf{21} (2011), 151-191.
%}
	
\bibitem{EL2006}	
{\small  \textsc{N. M. Ercolani and G. Lozano,}\  A bi-Hamiltonian structure for the integrable, discrete non-linear Schr\"{o}dinger system, \textit{ Phys. D,}  \textbf{218} (2006), 105-121.}

%\bibitem{NLSgas1}
%{\small  \textsc{G. Falqui, T. Grava and  G. Orsatti,}\ Shielding of breathers for the focusing nonlinear	Schr\"odinger equation,	\textit{Phys. Rev. Lett.}, \textbf{130} (2023),  127201.} %

\bibitem{FLYZ2006}
{\small  \textsc{E. G. Fan,\ G. Z. Li,\ Y. L. Yang \  and  \ L. Zhang,}\ Painlev\'{e}  XXXIV  asymptotics for the defocusing nonlinear Schr\"odinger equation with a  finite-genus algebro-geometric  background,	  submitted, 2025.}

\bibitem{FIKY2006}
{\small  \textsc{A. S. Fokas, A. R. Its, A. A. Kapaev and V. Yu. Novokshenov}, \ Painlev\'e Transcendents: The Riemann-Hilbert Approach, AMS Mathematical Surveys and Monographs, Vol. 128, Amer. Math. Society, Providence R.I., 2006.}

\bibitem{its1992}
{\small  \textsc{A. S. Fokas, A. R. Its and  A. V. Kitaev}, The isomonodromy approach to matrix models in 2D quantum gravity. \textit{Comm. Math. Phys.}, \textbf{147} (1992), 395-430.}
%\bibitem{YYLAG} {\small  \textsc{E. G. Fan, G. Z. Li,  Y. L. Yang   and   L. Zhang,} Painlev\'{e}  XXXIV  asymptotics for the defocusing nonlinear Schr\"odinger equation with a  finite-genus Algebro-Geometric  background,  submitted, 2025.}


%	\bibitem{Fokas}	{\small  \textsc{A. S. Fokas,}		\newblock  On a class of physically important integrable equations,		{\it Phys. D, } \textbf{87} (1995), 145-150.}
		
%	\bibitem{FIKNBook}
%	{\small  \textsc{A. S. Fokas, A. R. Its, A. A. Kapaev and V. Y. Novokshenov,}
%	Painlev\'{e} Transcendents: The Riemann-Hilbert Approach, Math. Surv. Monog., \textit{ Amer. Math. Soc.},
%	Providence, RI, \textbf{128} (2006).}
	
%	\bibitem{itsCMP}
%	{\small  \textsc{A.S. Fokas, A.R. Its, A.V. Kitaev,}
%	The isomonodromy approach to matrix models in 2D quantum gravity,
%\textit{Comm. Math. Phys.}, \textbf{147} (1992),  395-430}

\bibitem{GDZ2007}
{\small  \textsc{X. Geng, H. H. Dai and J. Zhu,} \ Decomposition of the discrete Ablowitz–Ladik hierarchy, \textit{Stud. Appl. Math.,} \textbf{118} (2007), 281-312.}

\bibitem{kdvgas}	
{\small  \textsc{M. Girotti, T. Grava, R. Jenkins and K. T. R. McLaughlin,} \ Rigorous asymptotics of a KdV soliton gas, \textit{Comm.  Math. Phys.,} \textbf{384} (2021), 733-784.}
	
\bibitem{mkdvgas}	
{\small  \textsc{M. Girotti, T. Grava, R. Jenkins, K. T. R. McLaughlin and A. Minakov},\ Soliton versus the gas:
Fredholm determinants, analysis, and the rapid oscillations behind the kinetic equation, \textit{ Comm. Pure  Appl. Math.}, \textbf{76} (2023), 3233-3299.}

\bibitem{G2001}
{\small \textsc{T. Grava},\  From the solution of the Tsarev system to the solution of the Whitham equations,  \textit{Math. Phys. Anal. Geom.},  \textbf{4} (2001),  65-96.}
	
\bibitem{GT2002}
{\small \textsc{T. Grava and F.-R. Tian}, The generation, propagation, and extinction of multiphases in the KdV zero-dispersion limit, \textit{Comm. Pure Appl. Math.}, \textbf{55} (2002), 1569-1639.}
	
%\bibitem{Hastings1980ABV}
%{\small 	\textsc{S. Hastings and J. B. McLeod}, \
%A boundary value problem associated with the second Painlev{\'e} transcendent and the Korteweg-de Vries equation, \textit{Arch. Ration. Mech. Anal.}, \textbf{73} (1980), 31-51.} %

\bibitem{I1982}
{\small  \textsc{Y. Ishimori,} \ An integrable classical spin chain,  \textit{J. Phys. Soc. Japan,} \textbf{51} (1982),  3417-3418.}

	
\bibitem{ikj2008}
{\small  \textsc{A. R. Its, A. B. J. Kuijlaars  and J. \"{O}stensson}, \ Critical edge behavior in unitary
	random matrix ensembles and the thirty fourth Painlev\'{e}
	transcendent, \textit{ Int. Math. Res. Not.}, \textbf{2008} (2008), rnn017.}
	
	
	\bibitem{ikj2009}
	{\small  \textsc{A. R. Its, A. B. J. Kuijlaars  and J. \"{O}stensson}, \ Asymptotics for a special solution of the thirty fourth Painlev\'e equation, \textit{ Nonlinearity}, \textbf{22} (2009), 1523-1558.}
	




	
\bibitem{KC1986}
{\small  \textsc{V. M. Kenkre and D. K. Campbell,} \ Self-trapping on a dimer: Time-dependent solutions of a discrete nonlinear Schr\"{o}dinger equation, \textit{Phys. Rev. B,}  \textbf{34} (1986), 4959-4961.}	
	
	\bibitem{KMcLVAV}
		{\small  \textsc{A. Kuijlaars, K. McLaughlin, W. Van Assche and M. Vanlessen,} \ The Riemann-Hilbert approach to
	strong asymptotics for orthogonal polynomials on $[-1, 1]$, \textit{Adv. Math.}, \textbf{188} (2004), 337-398.}
	
\bibitem{mkdvtran}
{\small  \textsc{G. Z. Li, Z. Y. Wang, Y. L. Yang and L. Zhang,}\ Transient asymptotics of mKdV soliton gas,  submitted, 2025.}


\bibitem{LN2012}
{\small  \textsc{L.-C. Li and I. Nenciu,}\ The periodic defocusing Ablowitz–Ladik equation and the geometry of Floquet CMV matrices,  \textit{Adv. Math.,} \textbf{231} (2012), 3330–3388.}


\bibitem{LSG2024}
{\small  \textsc{H. Liu, J. Shen and X. G. Geng,}\ Riemann–Hilbert method to the Ablowitz–Ladik equation: Higher-order cases,   \textit{Stud. Appl. Math.,} \textbf{153} (2024), 12748.}


\bibitem{MEKL1995}
{\small  \textsc{P. D. Miller, N. M. Ercolani, I. M. Krichever and C. D. Levermore,}\ Finite genus solutions to the Ablowitz–Ladik equations, \textit{Comm. Pure Appl. Math.,} \textbf{48} (1995), 1369–1440.}


	



\bibitem{N1987}
{\small  \textsc{N. Papanicoulau,} \ Complete integrability for a discrete heisenberg chain, \textit{J. Phys. A: Math. Gen.,} \textbf{20} (1987), 3637–3652.}












\bibitem{P2016}
{\small  \textsc{B. Prinari,} \ Discrete solitons of the focusing Ablowitz-Ladik equation with nonzero boundary conditions via inverse scattering, \textit{J. Math. Phys.,} \textbf{57} (2016), 083510.}

\bibitem{PV2016}
{\small  \textsc{B. Prinari and F. Vitale,} \ Inverse scattering transform for the focusing Ablowitz–Ladik system with nonzero boundary conditions,   \textit{Stud. Appl. Math.,} \textbf{137} (2016), 28-52.}


%\bibitem{QW2008}
%{\small  \textsc{W. Y. Qiu and R. Wong, }\
%Global asymptotic expansions of the Laguerre polynomials-a Riemann-Hilbert approach, \textit{Numer. Algorithms}, \textbf{49} (2008), 331-372.}






\bibitem{OY2014}
{\small  \textsc{Y. Ohta and J. Yang,} \ General rogue waves in the focusing and defocusing Ablowitz–Ladik equations,  \textit{J. Phys. A: Math. Theor.,}  \textbf{47} (2014), 255201.}

\bibitem{dlmf}
{\small \textsc{F. W. J. Olver, A. B. Olde Daalhuis, D. W. Lozier, B. I. Schneider, R. F. Boisvert, C. W. Clark, B. R. Miller, B. V. Saunders (Eds.),} \ 
NIST digital library of mathematical functions, 2018, http://dlmf .nist .gov/, Release 1.0.21 of 2018-12-15.}

\bibitem{HM1981}
{\small  \textsc{H. Segur and M. L. Ablowitz,} \  Asymptotic solutions of nonlinear evolution equations and a Painlev\'{e} transcendent,   \textit{Phys. D,} \textbf{3} (1981), 165-184.}

 
\bibitem{TH1990}	
{\small  \textsc{S. Takeno and K. Hori,} \ A propagating self-localized mode in a one-dimensional lattice with quartic anharmonicity, \textit{ J. Phys. Soc. Japan,} \textbf{59} (1990), 3037-3040.}

%\bibitem{V2007}
%{\small  \textsc{M. Vanlessen,} \ 
%Strong asymptotics of Laguerre-type orthogonal polynomials and applications in random matrix theory, \textit{Constr. Approx.}, \textbf{25} (2007), 125-175.}


\bibitem{VK1992}
{\small  \textsc{V. E. Vekslerchik and V. V. Konotop,}\ Discrete nonlinear Schr\"{o}dinger equation under non-vanishing boundary conditions,  \textit{Inverse Probl.,} \textbf{8} (1992), 889-909.}
 	


    	\bibitem{wang2023defocusing}
  {\small  \textsc{Z. Y. Wang and E. G. Fan,}\   The defocusing nonlinear Schr\"{o}dinger equation with a nonzero background: Painlev\'e asymptotics in two
	transition regions, \textit{Comm. Math. Phys.}, \textbf{402} (2023),  2879-2930.}



		
\bibitem{XF2018}	
{\small  \textsc{B. Q.	Xia and A. Fokas,}\ Initial-boundary value problems associated with the Ablowitz-Ladik system, \textit{ Phys. D,} \textbf{364} (2018), 27-61.}
	

\bibitem{XZ2011}
{\small \textsc{S. X. Xu and Y. Q. Zhao,}
Painlev\'e  XXXIV asymptotics of orthogonal polynomials  for the Gaussian weight with a jump at the edge,
\textit{Stud. Appl. Math.}, \textbf{127} (2011), 67-105.}


\bibitem{xu2024transient}
	{\small  \textsc{T. Y. Xu, Y. L. Yang and  L. Zhang}, \  Transient asymptotics of the modified Camassa-Holm equation, \textit{ J. Lond. Math. Soc.}, \textbf{110} (2024), e12967.}


%	
%	\bibitem{Zhou1989CPAM}
%	{\small  \textsc{X. Zhou,} \  Direct and inverse scattering transforms with arbitrary spectral singularities,
%		\textit{Commun.  Pure   Appl. Math.,} \textbf{42} (1989), 895-938.}
%	
	
	
	\bibitem{Z1971}
	{\small \textsc{V. Zakharov},\ Kinetic equation for solitons, \textit{Sov. Phys. JETP}, \textbf{33}  (1971),  538-541.}
    
    
 %   \bibitem{Zhou1989} {\small \textsc{X. Zhou,} \ The Riemann-Hilbert problem and inverse scattering,  \textit{SIAM J. Math. Anal.,}  \textbf{20} (1989), 966-986.}
	

	
	%\bibitem{CA1998}
	%A.  Constantin and J. Escher,
	%\newblock   Wave breaking for nonlinear nonlocal shallow water equations,
	%\newblock {\em   Acta Math., }  191(1998), 229-243.
	
	%\bibitem{YA2000}
	%Y. A. Li and P. J. Olver,
	%\newblock   ell-posedness, blow-up solutions for an integrable nonlinearly dispersive modelwave equation,
	%\newblock {\em  J. Diff. Eqns, }  162(2000), 27-63.
	
	











%	\bibitem{Charlier2020}  C. Charlier, J. Lenells, Airy and Painlev\'e asymptotics for	the mKdV equation, J. Lond. Math. Soc. 101 (2020), 194-225.
	
%	\bibitem{hz} L. Huang, L. Zhang, Higher order Airy and Painlev\'e asymptotics for the mKdV hierarchy, SIAM J. Math. Anal. 54 (2022), 5291-5334.
	
	







%\bibitem{CG2009}T. Claeys, T. Grava, Universality of the break-up profile for the KdV equation in the small dispersion limit using the Riemann-Hilbert approach, Comm. Math. Phys. 286 (2009), 979-1009.	

%\bibitem{MLS2025} A. Boutet de Monvel, J. Lenells, D. Shepelsky, The focusing NLS equation with step-like oscillating background: asymptotics in a transition zone, J. Differential Equ. (429) 2025, 747-801.













%\bibitem{IKO2009} A. R. Its, A. B. J. Kuijlaars, J. \"{O}stensson, Asymptotics for a special solution of the thirty fourth Painlev\'e equation, Nonlinearity 22 (2009), 1523-1558.

















\end{thebibliography}
\end{document}